\documentclass[a4paper,UKenglish,cleveref, autoref, thm-restate, nolineno]{socg-lipics-v2019}
\usepackage[T1]{fontenc}
\usepackage[utf8]{inputenc}
\usepackage{color}
\usepackage{babel}
\usepackage{verbatim}
\usepackage{multirow}
\usepackage{amsmath}
\usepackage{amsthm}
\usepackage{amssymb}
\usepackage{graphicx}

\hideLIPIcs

\makeatletter

\providecommand{\tabularnewline}{\\}

\theoremstyle{plain}
\newtheorem{thm}{\protect\theoremname}
\theoremstyle{plain}
\newtheorem{lem}[thm]{\protect\lemmaname}

\@ifundefined{date}{}{\date{}}

\usepackage{color}
\usepackage{babel}
\usepackage{multirow}
\usepackage{url}

\usepackage{tikz}
\usepackage{pgfplots}\usetikzlibrary{patterns}
\usepackage{subcaption}
\usepackage{cleveref}

\usepackage{xspace}
\renewcommand{\epsilon}{\varepsilon}
\newcommand{\eps}{\varepsilon}
\newcommand{\epsl}{\eps_{large}}
\newcommand{\epss}{\eps_{small}}

\newcommand{\epsst}{\eps_{strip}}

\newcommand{\opt}{OPT}
\newcommand{\optco}{\opt_{corr}}
\newcommand{\optsk}{\opt_{skew}}
\newcommand{\optla}{\opt_{large}}
\newcommand{\optsm}{\opt_{small}}
\newcommand{\optho}{\opt_{hor}}
\newcommand{\optve}{\opt_{ver}}
\newcommand{\optin}{\opt_{int}}

\newcommand{\R}{I}
\newcommand{\width}{w}
\newcommand{\height}{h}
\newcommand{\profit}{p}
\newcommand{\bottomc}{bc}

\newcommand{\leftc}{lc}

\newcommand{\Rco}{\R_{corr}}
\newcommand{\Rsm}{\R_{small}}
\newcommand{\Rla}{\R_{large}}
\newcommand{\Rho}{\R_{hor}}
\newcommand{\Rve}{\R_{ver}}
\newcommand{\Rsk}{\R_{skew}}
\newcommand{\Rin}{\R_{int}}

\newcommand{\typeC}{\textsf{type}}
\newcommand{\shapeC}{\mathsf{shape}}

\usepackage[textsize=tiny,textwidth=2cm,color=green!50!gray]{todonotes}

\definecolor{cadmiumgreen}{rgb}{0.0, 0.42, 0.24}

\title{Improved Approximation Algorithms for 2-Dimensional Knapsack:  Packing into Multiple L-Shapes, Spirals, and More}

\titlerunning{Improved Approximation Algorithms for 2-Dimensional Knapsack}

\author{Waldo G\'alvez}{Department of Computer Science, Technical University of Munich, Munich, Germany}{galvez@in.tum.de}{https://orcid.org/0000-0002-6395-3322}{Supported by the European Research Council, Grant Agreement No. 691672, project APEG.}
\author{Fabrizio Grandoni}{IDSIA, USI-SUPSI, Lugano, Switzerland}{fabrizio@idsia.ch}{}{Partially supported by the SNSF Excellence Grant 200020B\_182865/1.}

\author{Arindam Khan}{Department of Computer Science and Automation, Indian Institute of Science, Bangalore, India}{arindamkhan@iisc.ac.in}{https://orcid.org/0000-0001-7505-1687}{}
\author{Diego Ramírez-Romero}{Department of Mathematical Engineering, Universidad de Chile, Santiago, Chile}{dramirez@dim.uchile.cl}{}{}				
\author{Andreas Wiese}{Department of Industrial Engineering and Center for Mathematical Modeling, Universidad de Chile, Santiago, Chile}{awiese@dii.uchile.cl}{}{Partially supported by the ANID Fondecyt Regular grant 1200173.}

\authorrunning{W. Gálvez, F. Grandoni, A. Khan, D. Ramírez and A. Wiese}



\providecommand{\lemmaname}{Lemma}
\providecommand{\theoremname}{Theorem}

\makeatother

\providecommand{\lemmaname}{Lemma}
\providecommand{\theoremname}{Theorem}

\begin{document}
\global\long\def\N{\mathbb{N}}%
\global\long\def\C{\mathcal{C}}%
\global\long\def\B{\mathcal{B}}%
\global\long\def\L{\mathcal{L}}%
\global\long\def\P{\mathcal{P}}%
\global\long\def\Rh{I_{\mathrm{high}}}%
\global\long\def\Rl{I_{\mathrm{low}}}%
\global\long\def\Rhone{I_{\mathrm{high},1}}%
\global\long\def\Rhtwo{I_{\mathrm{high},2}}%
\global\long\def\area{a}%
\global\long\def\Rlo{I_{\mathrm{lonely}}}%
\global\long\def\ratioweak{1.6+\epsilon}%
\maketitle
\global\long\def\S{\mathcal{S}}%

\thispagestyle{empty} 
\begin{abstract}
\noindent In the \textsc{2-Dimensional Knapsack} problem (2DK) we
are given a square knapsack and a collection of $n$ rectangular items
with integer sizes and profits. Our goal is to find the most profitable
subset of items that can be packed non-overlappingly into the knapsack.
The currently best known polynomial-time approximation factor for
2DK is $17/9+\epsilon<1.89$ and there is a $(3/2+\eps)$-approximation
algorithm if we are allowed to rotate items by 90 degrees~{[}Gálvez
et al., FOCS 2017{]}. In this paper, we give $(4/3+\epsilon)$-approximation
algorithms in polynomial time for both cases, assuming that all input
data are {integers polynomially bounded in $n$}.

Gálvez et al.'s algorithm for 2DK partitions the knapsack into a constant
number of rectangular regions plus \emph{one} L-shaped region and
packs items into those {in a structured way}. We generalize this approach by allowing up
to a \emph{constant} number of {\emph{more general}} regions that can have the shape of
an L, a U, a Z, a spiral, and more, and therefore obtain an improved
approximation ratio. {In particular, we present an algorithm that
computes the essentially optimal structured packing into these regions.
}

\end{abstract}

\section{Introduction}

The \textsc{2-Dimensional (Geometric) Knapsack} problem (2DK) is a
natural geometric generalization of the fundamental (one-dimensional)
\textsc{Knapsack} problem. In 2DK we are given a set $\R$ of $n$
items $i$ which are axis-parallel rectangles specified by their width
$\width(i)\in\N$, height $\height(i)\in\N$, and profit $\profit(i)\in\N$.
Furthermore, we are given an axis-parallel square knapsack $K=[0,N]\times[0,N]$
for some $N\in\N$. The goal is to select a subset $\R'\subseteq\R$
of maximum total profit $\profit(\R'):=\sum_{i\in\R'}\profit(i)$
that can be placed non-overlappingly inside $K$. 
Formally, for each $i\in\R'$ we have to define a pair $(\leftc(i),\bottomc(i))$
that specifies the left and bottom coordinates of $i$, respectively,
such that $i$ is placed inside $K$ as $R(i):=(\leftc(i),\bottomc(i))\times(\leftc(i)+\width(i),\bottomc(i)+\height(i))$;
we require that $R(i)\subseteq K$ and also that for any two $i,j\in\R'$
it holds that $R(i)\cap R(j)=\emptyset$. We will consider also the
case \emph{with rotations}, in which each item $i\in\R$ can be rotated
by $90$ degrees, i.e., $i$ can be replaced by a rectangle with width
$\height(i)$, height $\width(i)$ (and profit~$\profit(i)$). In
the cardinality (or unweighted) setting of the problem all profits
are $1$.

2DK has several applications. For example, the rectangles can model
banners out of which one wants to place the most profitable subset
on a website or an advertisement board. Also, they can model pieces
that one wants to cut out of some raw material like wood or steel.
In addition, there are scheduling settings in which jobs need a consecutive
amount of some given resource (e.g., a frequency bandwidth) for some
amount of time; thus, each job can be modeled via a rectangle.

Most algorithms for 2DK and related problems work as follows: they
guess a partition of the knapsack into $O_{{\epsilon}}(1)$ rectangular
boxes, for some small constant $\eps>0$. Inside each box the items
are packed greedily using the Next-Fit-Decreasing-Height algorithm
\cite{CGJT80}, or even simpler by stacking items on top of each other
or next to each other. Implicitly, Jansen and Zhang use this strategy
to obtain a $(2+\epsilon)$-approximation algorithm for 2DK~\cite{JZ04swat,JZ04soda}.
The same approach, however using $(\log N)^{O_{\epsilon}(1)}$ boxes,
is used in a QPTAS which assumes that $N$ is quasi-polynomially bounded
in $n$ \cite{AW15}. Finding a PTAS for 2DK or ruling it out is a
major open problem in the area. This question is also open if $N$
is polynomially bounded {in~$n$}, i.e., for pseudo-polynomial time algorithms.

One might wonder whether a PTAS can be constructed using $O_{\epsilon}(1)$
boxes only. Unfortunately, as observed in \cite{GGHIKW17}, essentially
no better approximation ratio than 2 is achievable {in} this way. 
Hence a different type of packing is needed to breach this approximation
barrier (in polynomial time). This was recently achieved by Gálvez
et al.~\cite{GGHIKW17}, where the authors pack the items into $O_{\epsilon}(1)$
boxes and additionally one container with the shape of an L (which
is packed with an ad-hoc, more complex algorithm). 

This yields an approximation ratio of $17/9+\epsilon<1.89$ ({and} $558/325+\epsilon<1.72$
in the unweighted case). The authors also {present} $(3/2+\epsilon)$-
and $(4/3+\epsilon)$-approximation algorithms for 2DK with rotations
in the weighted and unweighted case, respectively.

Gálvez et al.~\cite{GGHIKW17} pose as an open problem how to efficiently
pack items into a \emph{constant} number of L-shaped containers, and observe
that this would lead to improved approximation algorithms for 2DK.
This problem was open even for just two L-shaped {containers}
and using pseudo-polynomial time $(nN)^{O_{\epsilon}(1)}$. 
{In this paper we solve (a generalization of) this problem, {and hence obtain} an improved approximation ratio.}

\subsection{Our contribution}

In this paper, we present a $(4/3+\epsilon)$-approximation algorithm
with a (pseudo-polynomial) running time {of} $(nN)^{O_{\epsilon}(1)}$
for (weighted) 2DK. We also achieve improved $(4/3+\epsilon)$- and
$(5/4+\epsilon)$-approximation algorithms for 2DK \emph{with rotations}
in the weighted and unweighted case, {resp., with the same running
time.} See Table~\ref{tab:results} for an overview of our and previous
results in the respective settings.

\begin{table}
\begin{centering}
\begin{tabular}{|c|c|c|c|}
\hline 
\multicolumn{2}{|c|}{Setting} & Known result  & Our result\tabularnewline
\hline 
\multirow{2}{*}{Without rotations} & Weighted  & $17/9+\epsilon<1.89$ \cite{GGHIKW17}  & $4/3+\epsilon$\tabularnewline
\cline{2-4} \cline{3-4} \cline{4-4} 
 & Unweighted  & $558/325+\epsilon<1.72$ \cite{GGHIKW17}  & $4/3+\epsilon$\tabularnewline
\hline 
\multirow{2}{*}{With rotations} & Weighted  & $3/2+\epsilon$ \cite{GGHIKW17}  & $4/3+\epsilon$\tabularnewline
\cline{2-4} \cline{3-4} \cline{4-4} 
 & Unweighted  & $4/3+\epsilon$ \cite{GGHIKW17}  & $5/4+\epsilon$\tabularnewline
 \hline
\end{tabular}
\par\end{centering}
\caption{\label{tab:results}A summary of our approximation ratios compared
to the best known results with running time $(nN)^{O_{\epsilon}(1)}$
for the respective settings.}
\end{table}

Our algorithms use $O_{\epsilon}(1)$ boxes and in addition, rather
than one single L-shaped container as in \cite{GGHIKW17}, a combination of $O_{\epsilon}(1)$ containers with the shape of an
L or even more complicated shapes. 
The latter are intuitively thin corridors with the property that if
we traverse them, we change the orientation at the turns (i.e., clockwise
or counter-clockwise) at most once. For example, {they} can
have figuratively the shape of a U, a Z, or a spiral (see Figure~\ref{fig:packings}).
{Since they are thin, they help us to distinguish the parts
of $K$ that are used by items that are wide and thin (horizontal
items) and items that are high and narrow (vertical items). The interaction
of these types of items is a major difficulty in 2DK.}

By standard arguments (in particular building upon the corridor decomposition
in~\cite{AHW19}), it is not hard to show that we can partition $K$
into a constant number of boxes and corridors of the allowed {types},
so that there exists a feasible packing of a $(4/3+\epsilon)$-approximate
solution into them. The non-trivial part is how to efficiently pack
items into our corridors. 
Here we cannot exploit the L-packing algorithm in~\cite{GGHIKW17}.
Indeed, the latter algorithm does not seem to generalize even to two
L-shaped corridors (even {if} $N=n^{O(1)}$), while we need
to handle $\Theta_{\epsilon}(1)$ corridors with possibly even more
general shapes.

Our strategy is to partition each of our corridors into $O_{\epsilon}(\log N)$
rectangular boxes. Using the properties of their shapes, we show that
this is indeed possible by losing only a factor of $1+\epsilon$ in
the approximation guarantee. Guessing these boxes explicitly would
take $N^{O_{\epsilon}(\log N)}$ time which is too slow. Instead,
we show that we can guess the \emph{sizes} of almost all of the boxes
in time $(nN)^{O_{\epsilon}(1)}$. Then, we place them into the corridors
in polynomial time using a dynamic program based on color-coding,
using that in total there are only $O_{\epsilon}(\log N)$ boxes to
place.

We remark that the above approach compromises between corridor shapes
that are general enough to allow for %
\mbox{%
$(4/3+\epsilon)$%
}-approximate packings, and at the same time simple enough so that
we can partition them into $O_{\epsilon}(\log N)$ boxes that we can
essentially guess in time~$(nN)^{O_{\epsilon}(1)}$. It is not clear
how to extend our {algorithm to }$\Theta(\log^{1/\epsilon}N)$,
or just even $\Theta(\log^{2}N)$ such boxes. However, this would
allow us to exploit corridors of more general shapes (say $W$-shaped),
hence achieving better approximation ratios. We leave this as an
interesting open problem.

\begin{figure}
\begin{centering}
\includegraphics[scale=0.5]{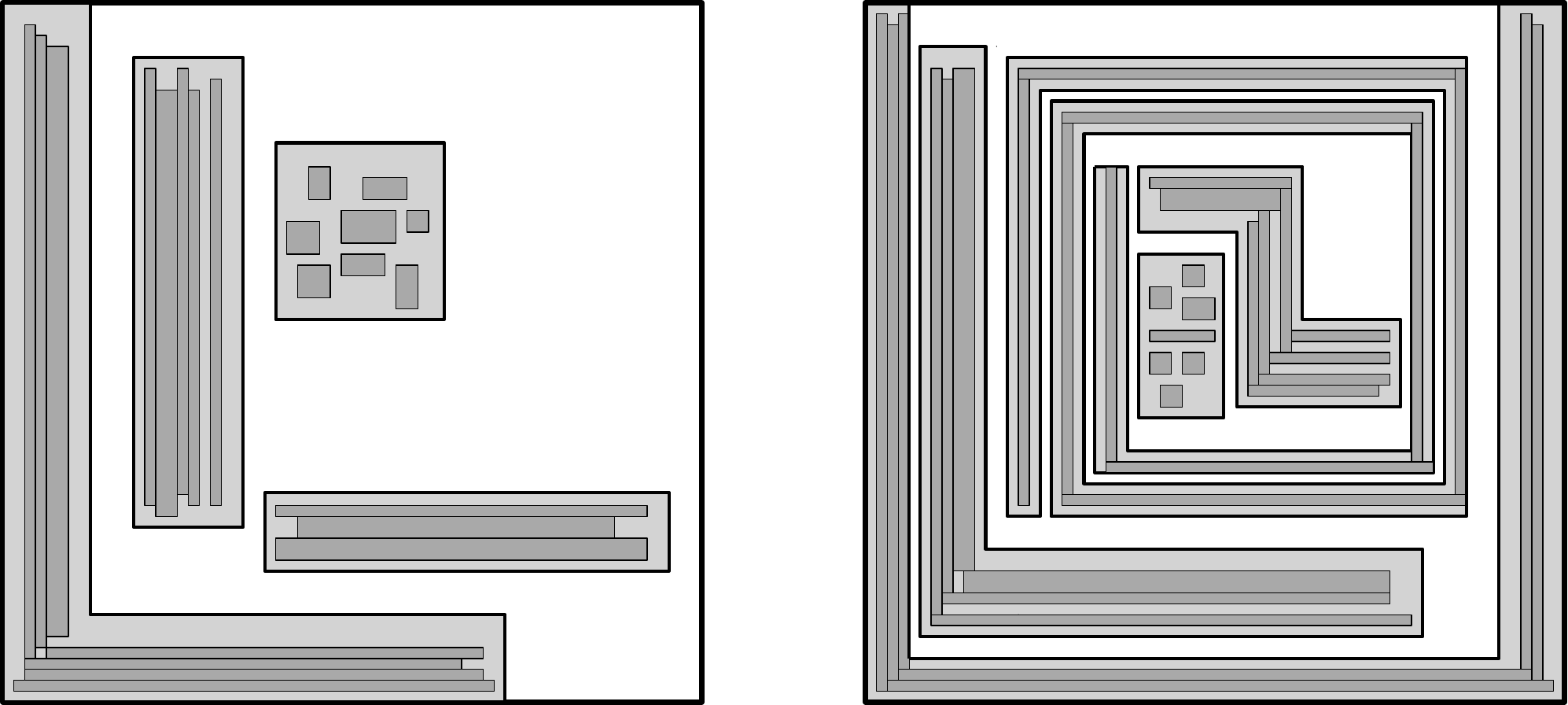} 
\par\end{centering}
\caption{\label{fig:packings}Left: a packing based on a single L-shaped container
and boxes as it was used in previous work. Right: A packing based
on a box and corridors with the shapes of an L, a U, a Z, and a spiral,
as we use them in our algorithms.}

\end{figure}

\subsection{Related Work}	
The QPTAS in \cite{AW15} (though with some restrictions on $N$)
suggests that 2DK most likely admits a PTAS. This is already known
for some relevant special cases: if the profit of each item equals
its area~\cite{BCJPS09}, if the size of the knapsack can be slightly
increased (resource augmentation)~\cite{FGJS05,JS07}, if all items
are relatively small \cite{FGJ05}, or squares~\cite{HW17,JS08}.

One might consider packing geometric objects other than rectangles. In particular,
there are constant approximation algorithms for packing triangles
and also arbitrary convex polygons under resource augmentation,
both assuming that arbitrary rotations are allowed \cite{MW2020}.
Also, for circles a $(1+\eps)$-approximation
is known under resource augmentation in one dimension if the profit
of each circle equals its area~\cite{lintzmayer2018two}. One can
consider natural generalizations of 2DK to a higher number of dimensions.
In particular, the 3-dimensional case, 3DK, has applications like
packing containers into a ship or cargo into a truck. 3DK is known
to be APX-hard \cite{CC09jda}, and constant approximation algorithms
are known \cite{DHJTT08,H09}. Khan et al.~\cite{KSS21} have  
given a $(2+\eps)$-approximation algorithm for a generalization of 2DK, which generalizes geometric 
packing and vector packing.

A parameterized version of 2DK for rectangles (where the
parameter is the number $k$ of packed items) is studied in \cite{GKW19}.
The authors show that the problem is $W[1]$-hard (both with and without
rotations). Furthermore, they provide an FPT $(1+\eps)$-approximation
for the case with rotations. Achieving a similar result for the case
without rotations is open.

A packing is called a guillotine packing if all rectangles can be separated by a sequence of 
end-to-end (guillotine) cuts~\cite{KMR20}. Abed et al.~\cite{AbedCCKPSW15} have 
given QPTAS for 2DK satisfying 
guillotine packing constraints, assuming the input data is quasi-polynomially bounded. 
Recently, Khan et al.~\cite{KMSW21} have shown a pseudo polynomial-time
approximation scheme for 2DK satisfying  guillotine packing constraints. 

In the \textsc{2-Dimensional Bin Packing} problem we are given a collection
of items similarly to 2DK, and copies of the same square knapsack
(the bins). Our goal is to pack \emph{all} the items using the smallest possible number of bins. 
The best known (asymptotic) result for this problem is due to Bansal and Khan~\cite{BK14}: 
they achieve a $1.405$ approximation based on a configuration-LP. This improves a 
series of previous results \cite{BCS09,C02,CGJ82,JS07,KR00}. 

Another closely related problem is \textsc{Strip Packing}. Informally,
we are given a knapsack of width $N$ and infinite
height,
and we wish to pack \emph{all} items so that the topmost coordinate
is as small as possible. This problem admits a $\left(5/3+\eps\right)$-approximation
\cite{HJPS14} (improving on \cite{BCR80,CGJT80,HV09,S94,S80,S97}) in the general case, 
$\left(3/2+\eps\right)$-approximation \cite{Galvez0AJ0R20} when none of the items are large, 
and it is NP-hard to approximate it below a factor $3/2$ by a simple reduction
from \textsc{Partition}. However, strictly better approximation ratios
can be achieved in pseudo-polynomial time $(Nn)^{O_{\eps}(1)}$~\cite{GGIK16,JR17,NW16},
being $5/4+\eps$ essentially the best possible ratio achievable in
this setting~\cite{HJRS18,JR19}. This shows that pseudo-polynomial
time can make the difference for rectangle packing problems. It would
be interesting to understand whether the techniques in our paper (as
well as those in \cite{AW15}) can be strengthened so as to run in polynomial
time for arbitrary $N$. \textsc{Strip Packing} was also studied in
the asymptotic setting \cite{JS07,KR00} and in the case with rotations
\cite{JS05stoc}.

Another related problem is the \textsc{Independent Set of Rectangles}
problem: here we are given a collection of axis-parallel rectangles
embedded in the plane, and we need to find a maximum cardinality/weight
subset of non-overlapping rectangles \cite{AHW19,AW13,CC09,chuzhoy2016approximating}.
The problem has also been studied for squares, disks, and pseudo-disks,
see e.g., \cite{erlebach2005polynomial,hunt1998nc,CH12}.

We refer the readers to \cite{CKPT17,K16} for surveys on geometric
packing problems.

\section{Preliminaries}

We start with a classification of the input items according to their
heights and widths. Let $\epsilon>0$. For two constants $1\geq\epsl>\epss>0$
to be defined later, we classify an item $i$ as:
\begin{itemize}
\item \emph{small}: if $\height(i),\width(i)\leq\epss N$; 
\item \emph{large}: if $\height(i),\width(i)>\epsl N$; 
\item \emph{horizontal}: if $\width(i)>\epsl N$ and $\height(i)\leq\epss N$; 
\item \emph{vertical}: if $\height(i)>\epsl N$ and $\width(i)\leq\epss N$; 
\item \emph{intermediate}: otherwise, i.e., {the length of at least one edge is in}
$(\epss N,\epsl N]$.
\end{itemize}
We call \emph{skewed} {the} items that are either horizontal or vertical.
We let $\Rsm$, $\Rla$, $\Rho$, $\Rve$, $\Rsk$, and $\Rin$ be
the items which are small, large, horizontal, vertical, skewed, and
intermediate, respectively. The corresponding intersection with the
optimal solution $\opt$ defines the sets $\optsm$, $\optla$, $\optho$,
$\optve$, $\optsk$ and $\optin$, respectively.

In order to describe our main ideas, we will start by considering the cardinality case of the problem (i.e., $\profit(i)=1$ for each item $i\in\R$) without rotations. In Sections~\ref{sec:partition-corridors}, \ref{sec:large-OPT},
and \ref{sec:DP-color-coding}, we will present a simplified algorithm that yields a $\ratioweak$ approximation. 

Notice that the optimum solution can contain at most $1/\epsl^{2}$ large items. Thus, unless $|\opt|\le1/(\epsilon\cdot\epsl^{2})$ (in which case we can solve the problem optimally in {time $n^{O(1/(\epsilon\cdot\epsl^{2}))}$ by complete enumeration}), we can drop all large
items by losing only a factor of $1+\epsilon$ in the approximation. Similarly, by standard shifting techniques (see Lemma \ref{lem:item-classification} for details), when defining $\epsl$
and $\epss$ one can ensure that the intermediate items can be neglected
by losing only a factor of $1+\epsilon$ in the approximation ratio (while maintaining that $\epsl$ and $\epss$ are lower-bounded by
some constant depending only on $\epsilon$). 

Hence, w.l.o.g.  we can assume that all items are small or skewed. It is possible to deal with small items by standard techniques from the literature, however this would make our exposition much more technical without introducing substantially new ideas. Hence, for the sake of simplicity, we will assume that there are no small items, i.e., all items are skewed. We remark that, even with the mentioned restrictions, the problem is far from {being} trivial. In particular, the best known approximation for the considered setting is $\frac{558}{325}+\epsilon\approx1.72$ \cite{GGHIKW17}. Our simplified algorithm has a better approximation ratio $\ratioweak$, {and it is also substantially simpler}. 

Later, in Sections~\ref{sec:improve-ratio}-\ref{sec:rotations}, we will present our main results: $(4/3+\epsilon)$-approximation
algorithms for the general case with and without rotations, and a $(5/4+\eps)$-approximation algorithm for the cardinality case with rotations respectively.

\section{\label{sec:partition-corridors}Partition into LU-corridors}

Our strategy is to partition the knapsack into $O_{\epsilon}(1)$
thin corridors, {each having the shape of an L or a U,}
such that\emph{
}there exists a $(\ratioweak)$-approximate solution in which each
item is contained in one of these corridors (see Figure~\ref{fig:lupacking-intro}). In this
section, we will make {this} precise and show that such a partition indeed
exists. Our algorithm will then guess this partition in polynomial
time. In Sections~\ref{sec:large-OPT} and \ref{sec:DP-color-coding}
we will show how to find the corresponding solution afterwards efficiently.
\begin{figure}[!ht]
	\centering
	\includegraphics[width=0.7\linewidth]{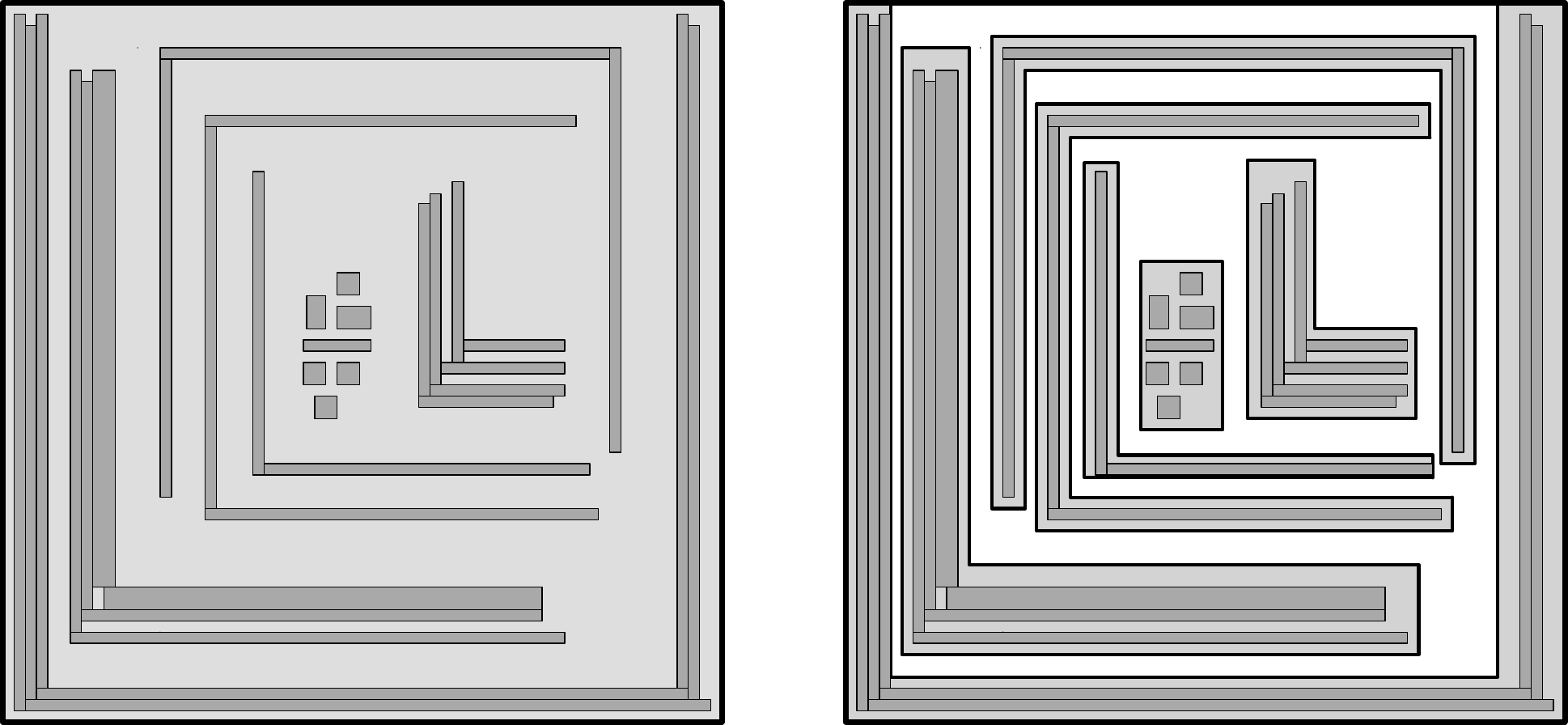}
	\caption{Left: a packing of horizontal, vertical and small items into the knapsack. Right: a decomposition of the knapsack into box-shaped, L-shaped and U-shaped corridors which contain the previous items.}
	\label{fig:lupacking-intro}
\end{figure}

Intuitively, a \emph{path corridor} is a polygon inside $K$ that describes
a path of width at most $\epsilon\cdot\epsl$ and that is allowed
to have bends, see Figure~\ref{fig:luz-1}. 

Formally, it is a simple
rectilinear polygon within $K$ with $2k$ edges $e_{0},\ldots,e_{2k-1}$
for some integer $k\ge2$, such that for each pair of horizontal (resp.,
vertical) edges $e_{i},e_{2k-i}$, $i\in\{1,...,k-1\}$ there exists
a vertical (resp., horizontal) line segment $\ell_{i}$ of length
less than $\epsilon\cdot\epsl$ such that both $e_{i}$ and $e_{2k-i}$
intersect $\ell_{i}$ and $\ell_{i}$ does not intersect any other
edge, and we require $e_{i}$ and $e_{2k-i}$ to have length at least
$\epsl/2$. Note that $e_{0}$ and $e_{k}$ are not required to satisfy
these properties. We say that such a path corridor $C$ has $s(C):=k-1$
\emph{subcorridors. }We say that a \emph{box }is a path corridor with
only one subcorridor, i.e., it is simply a rectangle.

Similarly, a \emph{cycle corridor }is intuitively a path corridor
in which the start and end point of the path coincide (see Figure~\ref{fig:luz-1}).
Formally, we define it to be a face bounded by two simple non-intersecting
rectilinear polygons defined by edges $e_{0},e_{1},\ldots,e_{k-1}$
and $e'_{0},e'_{1},\ldots,e'_{k-1}$, each of them of length at least
$\epsl/2$, such that the second polygon is contained in the first
one, and for each pair of corresponding horizontal (vertical) edges
$e_{i},e'_{i}$ for $i\in\{0,\ldots,k-1\}$ 
there is a vertical (horizontal, respectively) line segment $\ell_{i}$
of length less than $\epsilon\cdot\epsl$ such that both edges $e_{i}$
and $e'_{i}$ intersect $\ell_{i}$ and $\ell_{i}$ does not intersect
any other edge of the cycle corridor. We say that the resulting cycle
corridor $C$ has $s(C):=k$ subcorridors\emph{. }

We use a result from \cite{AHW19} that implies directly that there
there exists a partition of $K$ into $O_{\epsilon}(1)$ corridors
and a near-optimal solution in which each item is contained in some
corridor.
\begin{lem}[\cite{AHW19}]
\label{lem:corridor-partition}There exists a solution $\overline{\opt}\subseteq\opt$
with $\left|\opt\right|\le(1+\epsilon)\left|\overline{\opt}\right|$
and a partition of $K$ into a set of corridors $\bar{\C}$ with $|\bar{\C}|\le O_{\epsilon}(1)$
where each corridor $C\in\bar{\C}$ has at most $1/\epsilon$ subcorridors
each and such that each item $i\in\overline{\opt}$ is contained in
one corridor of~$\C$. 
\end{lem}

Algorithmically, we could guess $\bar{\C}$ in time $N^{O_{\epsilon}(1)}$
and then try to compute a profitable solution in which each item is
contained in one corridor in $\bar{\C}$, i.e., mimicking $\overline{\opt}$.
However, it is not clear how to compute such a solution in polynomial
time. Therefore, we partition $\bar{\C}$ further into a set of smaller
corridors $\C$ such that each resulting corridor has the shape of
an L or a U. We will ensure that there exists a $(\ratioweak)$-approximate
solution in which each item is contained in one corridor of $\C$.
Since the corridors in $\C$ are simpler than the corridors in $\bar{\C}$,
we will be able to compute in polynomial time essentially the most
profitable solution in which each item is contained in a corridor
of $\C$ (see in Sections~\ref{sec:large-OPT} and \ref{sec:DP-color-coding}). 

We say that a path corridor $C$ is an \emph{L-corridor }if $C$ has
exactly two subcorridors, 
and then figuratively it has the shape of
an L (see Figure~\ref{fig:luz-1}). Intuitively, we define that a path
corridor is a U-corridor if it has the shape of a U. Formally,
let $C$ be a a path corridor with exactly three subcorridors, and
hence $C$ is defined via edges $e_{0},...,e_{7}$. Assume w.l.o.g.~that
$e_{2}$ and $e_{6}$ are horizontal. We project $e_{2}$ and $e_{6}$
on the $x$-axis, let $I_{2}$ and $I_{6}$ be the resulting intervals.
Then $C$ is a \emph{U-corridor }if $I_{2}\subseteq I_{6}$ or $I_{6}\subseteq I_{2}$. 

In order to partition $\bar{\C}$, we first define a partition of
each {path/cycle} corridor $C\in\bar{\C}$ into $s(C)$ subcorridors
(see Figure~\ref{fig:luz-1}). We say that a \emph{subcorridor} of
a corridor $C$ is a simple polygon $P\subseteq C$ whose boundary
consists of two parallel edges (in most cases these will be edges
of $C$) and of two monotone axis-parallel curves, i.e., sets of axis-parallel
line segments such that either for any two points $(x_{1},y_{1}),(x_{2},y_{2})\in P$
where $x_{1}<x_{2}$ we have $y_{1}\le y_{2}$, or for any such two
points we have $y_{1}\ge y_{2}$. We require that each vertex of $P$
has integral coordinates. {Given} a corridor $C$
and a set of non-overlapping items $\R'$ inside $C$, we say that
a partition of $C$ into a set of {subcorridors} is \emph{nice for $\R'$}
if each {subcorridor} either intersects only items from $\R'\cap\Rho$ or
only items from $\R'\cap\Rve$.
\begin{lem}[\cite{AW15}, Lemma 2.4]

\label{lem:corridor-pieces-1}Let $C$ be a {path/cycle} corridor
containing a set of items $\R'$. There is a partition of $C$ into
$s(C)$ subcorridors that is nice for $\R'$.
\end{lem}

Now we take each path corridor $C\in\bar{\C}$ and delete the items
in every third {subcorridor}, starting with the $\alpha$-th {subcorridor} for some
offset $\alpha\in\{1,2,3\}$. Then we can divide $C$ into L-corridors
such that each remaining item is contained in one of these L-corridors.
Each item $i\in\overline{\opt}$ contained in $C$ is deleted only
for one choice of $\alpha$, and hence there is a choice for $\alpha$
such that we lose at most {one third of the profit} due to this step. Now
consider a cycle corridor $C\in\bar{\C}$. Note that $s(C)$ is even
and $s(C)\ge4$. If $s(C)=4$ (i.e., $C$ is a ring), we delete the
items in one of its four subcorridors, losing at most {one quarter of the profit}, and obtain a U-corridor. If $s(C)\ge6$ and $s(C)$ is divisible
by $3$ we do the same operation as for path corridors, losing at most
{one third of the profit}. For all other values of $s(C)$ we might lose
a larger factor since then $s(C)$ is not divisible by 3 and $P_{0}$
and $P_{s(C)-1}$ are adjacent, e.g., if $s(C)=8$. However, a case
distinction shows that we can still partition $C$ into L- and U-corridors
while {decreasing the profit} at most by a factor of $1.6$.
\begin{lem}
\label{lem:partition-knapsack}There exists a solution $\opt'\subseteq\opt$
with $|\opt|\le(\ratioweak)|\opt'|$ and a partition of $K$ into
a set $\C$ of $O_{\epsilon}(1)$ L- and U-corridors 
such that each item $i\in\opt'$ is contained in one corridor of~$\C$.
\end{lem}

\begin{figure}[t]
\centering \includegraphics[width=0.9\linewidth]{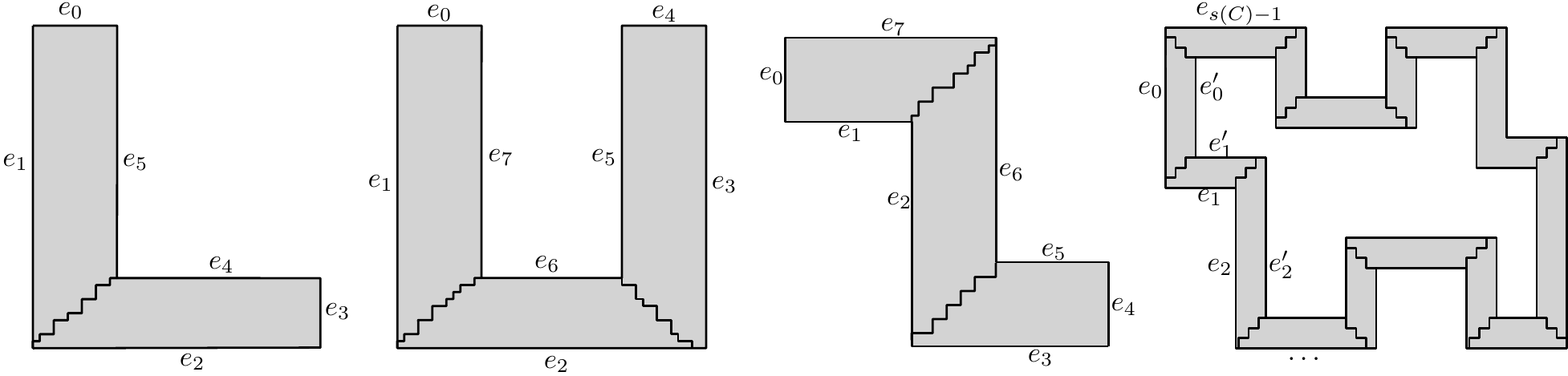}
\caption{{An L-corridor, a U-corridor, another path corridor with two bends, and a cycle corridor. The 
monotone axis-parallel curves indicate the boundaries of the subcorridors.}}
\label{fig:luz-1}
\end{figure}

The first step in our algorithm is to guess $\C$ which can be done
in time $N^{O_{\epsilon}(1)}$. The next step is to compute a solution
with at least $(1-\epsilon)|\opt'|$ items in which each item is contained
in one corridor of $\C$. For this, we consider two cases separately,
which are intuitively the case that $|\opt'|\ge\Omega_{\epsilon}(\log N)$
and $|\opt'|\le O_{\epsilon}(\log N)$, and they are treated in Sections~\ref{sec:large-OPT}
and \ref{sec:DP-color-coding}, respectively.

\section{\label{sec:large-OPT}Packing via guessing slices}

In this section, we assume that $|\opt'|>c_{\epsilon}\cdot\log N$
for some constant $c_{\epsilon}$ to be defined later. We describe
an algorithm that computes a solution of size $(1-\epsilon)|\opt'|$
such that each item of this solution is contained in
a corridor in $\C$.

First, we group the items into $O_{\epsilon}(\log N)$ groups where
we group the items in $\Rho$ according to their heights and the items
in $\Rve$ according to their widths. Formally, for each $\ell\in\{0,...,\left\lfloor \log_{1+\epsilon}N\right\rfloor \}$
we define $\Rho^{(\ell)}:=\{i\in\Rho|\height(i)\in[(1+\epsilon)^{\ell},(1+\epsilon)^{\ell+1})\}$
and $\Rve^{(\ell)}:=\{i\in\Rve|\width(i)\in[(1+\epsilon)^{\ell},(1+\epsilon)^{\ell+1})\}$.

So intuitively, for each $\ell$ the items in $\Rho^{(\ell)}$ essentially
all have the same height and the items in $\Rve^{(\ell)}$ essentially
all have the same width. Now, for the groups $\Rho^{(\ell)}$, $\Rve^{(\ell)}$
we guess estimates $\mathrm{opt}_{hor}^{(\ell)},\mathrm{opt}_{ver}^{(\ell)}$
for $|\Rho^{(\ell)}\cap\opt'|$, $|\Rve^{(\ell)}\cap\opt'|$, respectively.
Even though there can be $\Theta_{\epsilon}(\log N)$ of these groups
and for each guessed value there are potentially $\Omega(n)$ options,
we guess the estimates for \emph{all }groups in parallel in time $(nN)^{O_{\epsilon}(1)}$,
adapting a technique from~\cite{chekuri2005polynomial}. 
\begin{lem}	\label{lem:estimate-sizes-1}
	In time $(nN)^{O_{\eps}(1)}$ we
	can guess the values for all pairs $\mathrm{opt}_{hor}^{(\ell)},\mathrm{opt}_{ver}^{(\ell)}$
	{with $\ell\in\{0,...,\left\lfloor \log_{1+\eps}N\right\rfloor \}$}
	such that 
\begin{itemize}
\item $\sum_{\ell}\mathrm{opt}_{hor}^{(\ell)}+\mathrm{opt}_{ver}^{(\ell)}\ge(1-\epsilon)|\opt'|$
and 
\item $\mathrm{opt}_{hor}^{(\ell)}\le|\opt'\cap\Rho^{(\ell)}|$ and $\mathrm{opt}_{ver}^{(\ell)}\le|\opt'\cap\Rho^{(\ell)}|$
for each $\ell\in\{0,...,\left\lfloor \log_{1+\epsilon}N\right\rfloor \}$. 
\end{itemize}
\end{lem}
\begin{proof}
 First, we guess $|\opt'|$ for which there are only $n$ options.
 {We show now how to guess in time }

 $N^{O_\epsilon(1)}$ {the} feasible values for $\mathrm{opt}_{hor}^{(\ell)}$; a symmetric argument holds for $\mathrm{opt}_{ver}^{(\ell)}$. Let us define each $\mathrm{opt}_{hor}^{(\ell)}$ as the largest integer of the form $k_{hor}^{(\ell)}\cdot\frac{\epsilon}{{4}\log_{1+\epsilon}N}|\opt'|$
which is {upper} bounded by $|\opt'\cap\Rho^{(\ell)}|$, where $k_{hor}^{(\ell)}$ is a non-negative integer. Notice that trivially $\sum_\ell k_{hor}^{(\ell)}\leq \frac{{4}\log_{1+\epsilon}N}{\epsilon}$. We encode all such values $k_{hor}^{(\ell)}$  as a single binary string as follows: we represent each $k_{hor}^{(\ell)}$ as a string of $k_{hor}^{(\ell)}$ $0$-bits follow{ed} by one $1$-bit, and then chain such strings according to the index $\ell$. 

The final bit string encodes the solution. Notice that {this} string contains at most $\log_{1+\epsilon}N{+1}+\sum_{\ell}k_{hor}^{(\ell)}=O(\frac{\log_{1+\epsilon}N}{\epsilon})$ bits, hence we can guess it in time ${N}^{O_{\epsilon}(1)}$.

The claim follows since

\[|\opt'|-\sum_{\ell}\left(\mathrm{opt}_{hor}^{(\ell)}+\mathrm{opt}_{ver}^{(\ell)}\right)\le 2(\log_{1+\epsilon}N{+1})\cdot\frac{\epsilon}{{4}\log_{1+\epsilon}N}|\opt'|{\leq}\epsilon|\opt'|.\]
 \end{proof}

\subparagraph{Definition of slices.}

Next, for each group $\Rho^{(\ell)}$ we define slices that together
are essentially as profitable as the items in $\opt'\cap\Rho^{(\ell)}$.
We first order the items in $\Rho^{(\ell)}$ non-decreasingly by width
and select the first $\frac{1}{1+\epsilon}\mathrm{opt}_{hor}^{(\ell)}$
items. One can show {easily} that their total height is at most 
$\height(\opt'\cap\R_{hor}^{(\ell)})$.
Let $i$ be one of these items. Intuitively, we slice $i$ horizontally
into slices of height~1. Formally, for $i$ we introduce $\height(i)$
items of height~1 and profit $1/(1+\epsilon)^{\ell}$ each. Let $\hat{\R}_{hor}^{(\ell)}$
denote the resulting set of slices. We do this procedure for each $\ell$
and a symmetric procedure for the group $\Rve^{(\ell)}$ for each
$\ell$, resulting in {a} set of slices {$\hat{\R}_{ver}^{(\ell)}$}.
\begin{lem}
\label{lem:place-pre-slices}It is possible to place the slices in
$\left\{ \hat{\R}_{hor}^{(\ell)},\hat{\R}_{ver}^{(\ell)}\right\} _{\ell}$
non-overlappingly inside $K$ such that each slice is contained in
some corridor in $\C$. Also, we have that $\sum_{\ell}\left(\profit(\hat{\R}_{hor}^{(\ell)})+\profit(\hat{\R}_{ver}^{(\ell)})\right)\ge\frac{1}{1+O(\epsilon)}\left|\opt'\right|$.
\end{lem}

\begin{proof}
	Let $\ell\in\{0,...,\left\lfloor \log_{1+\epsilon}N\right\rfloor \}$.
	Recall that $\mathrm{opt}_{hor}^{(\ell)}\le|\opt'\cap\Rho^{(\ell)}|$,
	all items in $\Rho^{(\ell)}$ have the same height (up to a factor
	of $1+\epsilon$), and we selected the $\frac{1}{1+\epsilon}\mathrm{opt}_{hor}^{(\ell)}$
	items in $\Rho^{(\ell)}$ of minimum width.
	Now consider the items in $\opt'\cap\Rho^{(\ell)}$ in non-decreasing width.
	So if we consider the slices in $\hat{\R}_{hor}^{(\ell)}$ in non-decreasing width, they fit into the space that
	is occupied by the items in $\opt'\cap\Rho^{(\ell)}$ in $\opt'$.
	Also, $\frac{1}{1+O(\epsilon)}\height(\opt'\cap\R_{hor}^{(\ell)})\le|\hat{\R}_{hor}^{(\ell)}|\le\height(\opt'\cap\R_{hor}^{(\ell)})$
	and each slice in $\hat{\R}_{hor}^{(\ell)}$ has a profit of $1/(1+\epsilon)^{\ell}$.
	Hence,  $\profit(\hat{\R}_{hor}^{(\ell)}) \ge \frac{1}{1+O(\eps)} \profit(\opt'\cap\R_{hor}^{(\ell)})$.
	Summing this over each $\ell$ and a similar statement for the vertical
	items,  we obtain that that $\sum_{\ell}\left(\profit(\hat{\R}_{hor}^{(\ell)})+\profit(\hat{\R}_{ver}^{(\ell)})\right)\ge\frac{1}{1+O(\epsilon)}\left|\opt'\right|$.
\end{proof}

Next, for each $\ell\in\{0,...,\left\lfloor \log_{1+\epsilon}N\right\rfloor \}$
we round the widths of the slices in $\hat{\R}_{hor}^{(\ell)}$ via
linear grouping such that they have at most $1/\epsilon$ different
widths and we lose at most a factor of $1+O(\epsilon)$ in their profit
due to this rounding. Formally, we sort the slices in $\hat{\R}_{hor}^{(\ell)}$
non-increasingly by width and then partition them into $1/\epsilon+1$
groups such that each group contains $\left\lceil \frac{1}{1/\epsilon+1}|\hat{\R}_{hor}^{(\ell)}|\right\rceil $
slices (apart from possibly the last group which might contain fewer
slices). Let $\hat{\R}_{hor}^{(\ell)}=\hat{\R}_{hor,1}^{(\ell)}\dot{\cup}...\dot{\cup}\hat{\R}_{hor,1/\epsilon+1}^{(\ell)}$
denote the resulting partition. We drop the slices in $\hat{\R}_{hor,1}^{(\ell)}$
(whose total profit is at most $\epsilon\cdot\profit(\hat{\R}_{hor}^{(\ell)})$).
Then, for each $j\in\{2,...,1/\epsilon+1\}$ we increase the width
of the slices in $\hat{\R}_{hor,j}^{(\ell)}$ to the width of the
widest slic{e} in $\hat{\R}_{hor,j}^{(\ell)}$. By construction, the
resulting slices have $1/\epsilon$ different widths. Let $\tilde{\R}_{hor}^{(\ell)}$
denote the resulting set and let $\tilde{\R}_{hor}^{(\ell)}=\tilde{\R}_{hor,1}^{(\ell)}\dot{\cup}...\dot{\cup}\tilde{\R}_{hor,1/\epsilon}^{(\ell)}$
denote a partition of $\tilde{\R}_{hor}^{(\ell)}$ according to the
widths of the slices, i.e., for each $j\in\{1,...,1/\epsilon\}$ the
set $\tilde{\R}_{hor,j}^{(\ell)}$ contains the rounded slices from
$\hat{\R}_{hor,j+1}^{(\ell)}$.%

We do this procedure for each $\ell$ and a symmetric procedure for
the group $\Rve^{(\ell)}$ for each $\ell$.

\begin{lem}
\label{lem:place-slices}It is possible to place the slices in $\left\{ \tilde{\R}_{hor,j}^{(\ell)},\tilde{\R}_{ver,j}^{(\ell)}\right\} _{\ell,j}$
non-overlappingly inside $K$ such that each slice is contained in
some corridor in $\C$. Also, we have $\sum_{\ell}\left(\profit(\tilde{\R}_{hor}^{(\ell)})+\profit(\tilde{\R}_{ver}^{(\ell)})\right)\ge\frac{1}{1+\epsilon}\sum_{\ell}\left(\profit(\hat{\R}_{hor}^{(\ell)})+\profit(\hat{\R}_{ver}^{(\ell)})\right)$.
\end{lem}

\begin{proof}
	For each $\ell\in\{0,...,\left\lfloor \log_{1+\epsilon}N\right\rfloor \}$
	and each $j\in\{2,...,1/\epsilon+1\}$ the slices in $\tilde{\R}_{hor,j}^{(\ell)}$
	fit into the space occupied by the slices in $\hat{\R}_{hor,j-1}^{(\ell)}$,
	as $\height(\tilde{\R}_{hor,j}^{(\ell)})\le h(\hat{\R}_{hor,j-1}^{(\ell)})$ and 
	width of maximum width rectangle in $\tilde{\R}_{hor,j}^{(\ell)}$ is at most the width of minimum width rectangles in $\hat{\R}_{hor,j-1}^{(\ell)}$.
	A similar argumentation holds for the sets $\Rve^{(\ell)}$ of vertical
	items.
	As we drop the slices in $\hat{\R}_{hor,1}^{(\ell)}$
	(whose total profit is at most $\eps \cdot\profit(\hat{\R}_{hor}^{(\ell)})$),
	we totally loose only $1/(1+\eps)$ fraction of the profit due to them. Summing these for vertical items and over all $\ell$,  give us that $\sum_{\ell}\left(\profit(\tilde{\R}_{hor}^{(\ell)})+\profit(\tilde{\R}_{ver}^{(\ell)})\right)\ge\frac{1}{1+\epsilon}\sum_{\ell}\left(\profit(\hat{\R}_{hor}^{(\ell)})+\profit(\hat{\R}_{ver}^{(\ell)})\right)$. 
\end{proof}
We fix a partition of each corridor $C\in\C$ into subcorridors that
is nice for the slices $\left\{ \tilde{\R}_{hor,j}^{(\ell)},\tilde{\R}_{ver,j}^{(\ell)}\right\} _{\ell,j}$.
Note that we do not compute this partition explicitly but we use it
for our analysis and as guidance for our algorithm. For each corridor
$C\in\C$ or subcorridor ${S}$ denote by $\tilde{\R}(C)$ and $\tilde{\R}({S})$
the slices assigned to $C$ and ${S}$, {resp.}, due to Lemma~\ref{lem:place-slices}. 

\subparagraph{Structuring slices inside subcorridors.}
Consider a subcorridor ${S}$ of a corridor $C\in\C$. The placement
of the slices inside ${S}$ due to Lemma~\ref{lem:place-slices} might
be complicated. Instead, we would like to have a packing where the
slices are packed \emph{nicely}, i.e., they are stacked on top of
each other if ${S}$ is horizontal, and side by side if ${S}$ is vertical
(see Figure~\ref{fig:nice-slice}). This might not be possible to achieve exactly, but
we construct something very similar. We prove that inside ${S}$ we
can place $O_{\epsilon}(1)$ boxes and one subcorridor ${S}'\subseteq {S}$
(which we will call sub-subcorridor in order to distinguish it from
the subcorridors) that are pairwise disjoint and such that inside
them we can nicely place essentially all slices from $\tilde{\R}(C)$
(see Figure~\ref{fig:nice-slice}). We can guess the placement of the $O_{\epsilon}(1)$
boxes inside ${S}$. Unfortunately, we cannot guess ${S}'$ directly,
but we can guess the two edges 
that define {the boundary of $S'$ together with the two axis-parallel curves. 
We refer to them as the \emph{edges of $S'$.} Note that they are horizontal if $S$ (and hence~$S'$) is horizontal, and vertical otherwise.
}
We construct ${S}'$
such that the longer {of these two edges} is always also an edge of ${S}$ (see Figure~\ref{fig:nice-slice}).%
\begin{figure}[!ht]
	\centering
	\includegraphics[width=0.7\linewidth]{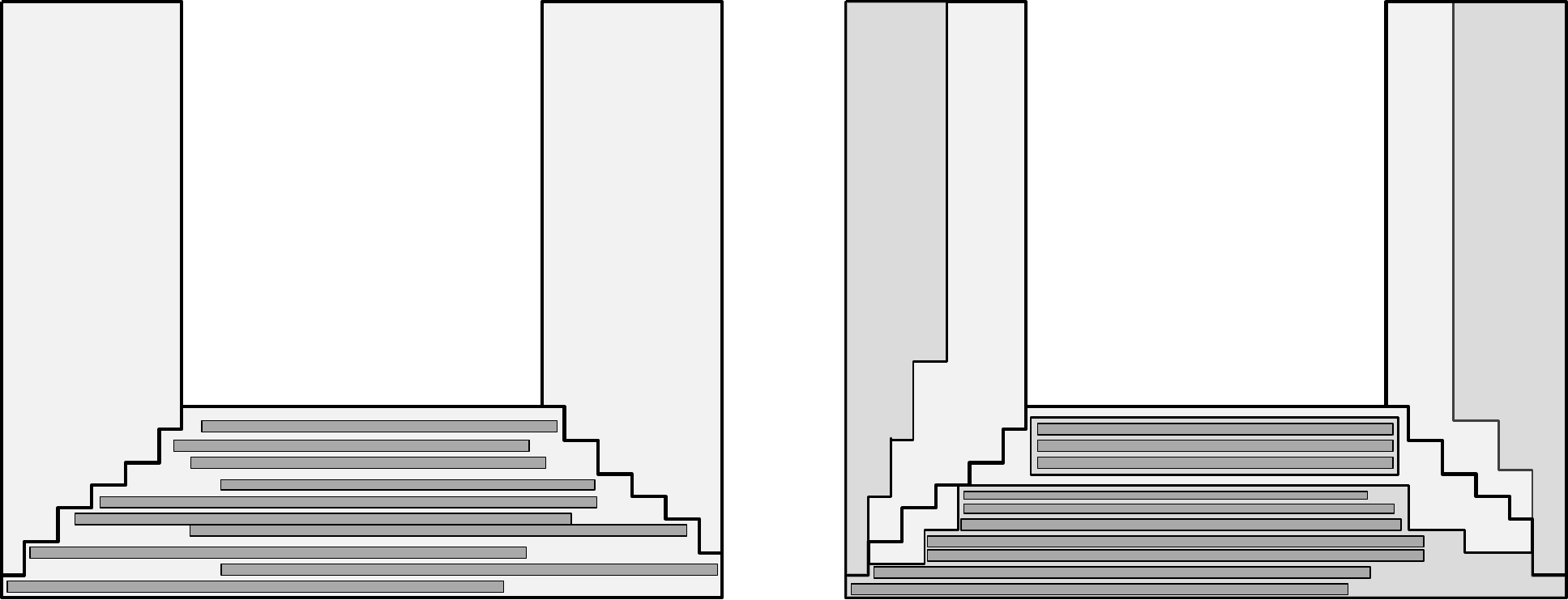}
	\caption{Left: a subcorridor with items packed inside it. Right: A partition of each subcorridor into $O_{\epsilon}(1)$ boxes and a sub-subcorridor, all containing slices which are nicely packed.}
	\label{fig:nice-slice}
\end{figure}

\begin{lem}
\label{lem:partition-subcorridors}For each horizontal/vertical subcorridor
${S}$ we can guess in time $N^{O_{\epsilon}(1)}$ 
\begin{itemize}
\item the two edges of a horizontal/vertical sub-subcorridor ${S}'\subseteq {S}$
such that the longer edge of ${S}'$ coincides with the longer edge
of ${S}$,
\item $O_{\epsilon}(1)$ non-overlapping boxes $\B({S})$ inside ${S}$ that
are disjoint with ${S}'$, 
\end{itemize}
such that we can nicely place slices from $\tilde{\R}({S})$ with a
total profit of $(1-\epsilon)\profit(\tilde{\R}({S}))$ inside ${S}'$
and the boxes~$\B({S})$.
\end{lem}

We apply Lemma~\ref{lem:partition-subcorridors} to each subcorridor
${S}$ of a corridor $C\in\C$. Let $\S$ and $\B$ denote the resulting
set of sub-subcorridors and boxes, respectively. For each $\ell$,
each $j$, and each $F\in\B\cup\S$ denote by $\tilde{\R}_{hor,j}^{(\ell)}(F)$
and $\tilde{\R}_{ver,j}^{(\ell)}(F)$ the respective slices from $\tilde{\R}_{hor,j}^{(\ell)},\tilde{\R}_{ver,j}^{(\ell)}$
in $F$, respectively. Using simple slice reorderings, we can prove the following lemma.
\begin{lem}
\label{lem:slices-structured}There is a packing of the slices in
$\left\{ \tilde{\R}_{hor,j}^{(\ell)}(F),\tilde{\R}_{ver,j}^{(\ell)}(F)\right\} _{j,\ell,F}$
such that for each box or sub-subcorridor $F\in\B\cup\S$ we can assume
w.l.o.g.~that 
\begin{itemize}
\item horizontal/vertical items inside $F$ are ordered {non-increasingly} by width/height,
starting at the longer edge of $F$ if $F$ is a sub-subcorridor,
and starting at an arbitrary edge if $F$ is a box; ties are broken
according to the input items that the slices correspond to, 
\item any two adjacent horizontal/vertical slices of the same width/height
are placed exactly on top of each other/side by side.
\end{itemize}
\end{lem}

Therefore, we can construct this packing of the slices if we knew
the cardinality of $\tilde{\R}_{hor,j}^{(\ell)}(F)$ and $\tilde{\R}_{ver,j}^{(\ell)}(F)$
for each $F\in\B\cup\S$ {and each $\ell$ and $j$}. We guess this cardinality approximately
in the following lemma for each $F,\ell$ and $j$ in parallel. 
\begin{lem}
\label{lem:guess-slices}In time $(nN)^{O_{\epsilon}(1)}$ we can
guess values $\mathrm{opt}_{hor,j}^{(\ell)}(F),\mathrm{opt}_{ver,j}^{(\ell)}(F)$
for each $\ell\in\{0,...,\left\lfloor \log_{1+\epsilon}N\right\rfloor \}$,
$j\in\{1,...,1/\epsilon\}$, $F\in\B\cup\S$ such that
\begin{itemize}
\item $\mathrm{opt}_{hor,j}^{(\ell)}(F)\le\left|\tilde{\R}_{hor,j}^{(\ell)}(F)\right|$
and $\mathrm{opt}_{ver,j}^{(\ell)}(F)\le\left|\tilde{\R}_{ver,j}^{(\ell)}(F)\right|$
for each $\ell,j,F$ and
\item $\sum_{F\in\B\cup\S}\mathrm{opt}_{hor,j}^{(\ell)}(F)\ge(1-\epsilon)\sum_{F\in\B\cup\S}\left|\tilde{\R}_{hor,j}^{(\ell)}(F)\right|$
and $\sum_{F\in\B\cup\S}\mathrm{opt}_{ver,j}^{(\ell)}(F)\ge(1-\epsilon)\sum_{F\in\B\cup\S}\left|\tilde{\R}_{ver,j}^{(\ell)}(F)\right|$
for each $\ell,j$.
\end{itemize}
\end{lem}

\begin{proof}
	For each $\ell,j,$ and $F$ we define $\mathrm{opt}_{hor,j}^{(\ell)}(F)$
	to be the largest integral multiple of $\frac{\epsilon}{|\B|+|\S|}\left|\tilde{\R}_{hor,j}^{(\ell)}(F)\right|$
	that is at most $\left|\tilde{\R}_{hor,j}^{(\ell)}(F)\right|$. 
	Then clearly, $\mathrm{opt}_{hor,j}^{(\ell)}(F)\le\left|\tilde{\R}_{hor,j}^{(\ell)}(F)\right|$
	for each $\ell,j,F$.
	Also, $\sum_{F\in\B\cup\S}\mathrm{opt}_{hor,j}^{(\ell)}(F) \ge \sum_{F\in\B\cup\S} (1-\frac{\epsilon}{|\B|+|\S|})\left|\tilde{\R}_{hor,j}^{(\ell)}(F)\right| \ge (1-\epsilon)\sum_{F\in\B\cup\S}\left|\tilde{\R}_{hor,j}^{(\ell)}(F)\right|$.
	Similarly, for each $\ell,j,F,$ we have $\mathrm{opt}_{ver,j}^{(\ell)}(F)\le\left|\tilde{\R}_{ver,j}^{(\ell)}(F)\right|$
	and $\sum_{F\in\B\cup\S}\mathrm{opt}_{ver,j}^{(\ell)}(F)$$\ge(1-\epsilon)\sum_{F\in\B\cup\S}\left|\tilde{\R}_{ver,j}^{(\ell)}(F)\right|$.
	Note that for $\mathrm{opt}_{hor,j}^{(\ell)}(F)$ there are only $\frac{|\B|+|\S|}{\epsilon}=O_{\epsilon}(1)$
	options. We define the values $\mathrm{opt}_{ver,j}^{(\ell)}(F)$
	similarly. Since there are only $O_{\epsilon}(\log nN)$ of these
	values altogether, we can guess all of them in time ${O_{\epsilon}(1)}^{O_{\epsilon}(\log nN)}=(nN)^{O_{\epsilon}(1)}$.
\end{proof}

\subparagraph{Placing slices inside subcorridors.}
Given the number of slices in each box and each sub-subcorridor due
to Lemma~\ref{lem:guess-slices}, we compute a corresponding packing
for the slices. Inside of each box we simply sort the slices by height
or width, respectively, and then pack them in this order. For packing
the slices inside the sub-subcorridors of a corridor $C$, recall
that we do not know the precise sub-subcorridors, we know only the
guessed edges due to Lemma~\ref{lem:partition-subcorridors}. However,
we can still find a packing for the slices inside of the sub-subcorridors
of $C$. We start with the first sub-subcorridor {$S_1$} of $C$,
sort its slices by height or width, respectively (breaking ties according
to the input items that the slices correspond to), and place them
in this order, starting at the longer edge of {$S_1$}. When we do this,
we push the slices as far as possible to the edge $e_{0}$. The resulting
packing satisfies the properties of Lemma~\ref{lem:slices-structured}.
If $s(C)\ge2$ then we do the same procedure for the last sub-subcorridor
${S}_{s(C)}$ of $C$, and in particular we push its slices as far
as possible to the edge $e_{s(C)}$. If $s(C)\in\{1,2\}$ then we
are done now. Otherwise $s(C)=3$ since $s(C)\le3$ for each $C\in\C$
and the slices of the second sub-subcorridor {$S_2$} are still not
placed. We sort the slices as before and place them in this order,
starting at the longer edge of {$S_2$} and such that their placement
satisfies the properties of Lemma~\ref{lem:slices-structured}. Since
we had pushed the slices in {$S_1$} and {$S_3$} maximally to the
edges $e_{0}$ and $e_{k}$, one can show that this is indeed possible.

\subparagraph{Rounding slices.}
For each set $\Rho^{(\ell)},\Rve^{(\ell)}$, their corresponding slices
induce in total $O_{\epsilon}(1)$ rectangular areas into which we
assigned these slices: at most one for each of the $O_{\epsilon}(1)$
sub-subcorridors and at most one for each of the $O_{\epsilon}(1)$
boxes inside each of the $O_{\epsilon}(1)$ subcorridors. For each
$\ell$ we denote by $\B_{hor}^{(\ell)},\B_{ver}^{(\ell)}$ these
corresponding areas which are in fact boxes. Now the important observation
is that inside the boxes $\B_{hor}^{(\ell)}$ we can place at least
$(1-O(\epsilon))\left|\Rho^{(\ell)}\cap\opt'\right|-2|\B_{hor}^{(\ell)}|$
items from $\Rho^{(\ell)}$ as follows. 
{Based on the slices for $\R^{(\ell)}_{hor}$, we first construct a fractional packing 
of $\frac{1}{1+O(\epsilon)}\mathrm{opt}^{(\ell)}_{hor}$ items from $\R^{(\ell)}_{hor}$
in which
there are at most $2|\B_{hor}^{(\ell)}|$ items that are fractionally assigned to a box. Then we 
simply drop these fractional items.
}

We use a symmetric procedure for the sets $\Rve^{(\ell)}$.
\begin{lem}\label{lem:converting-slices-items}
For each $\ell\in\{0,...,\left\lfloor \log_{1+\epsilon}N\right\rfloor \}$,
in time $O_{\epsilon}(nN)$ we can pack at least $(1-O(\epsilon))\left|\Rho^{(\ell)}\cap\opt'\right|-2|\B_{hor}^{(\ell)}|$
items from $\Rho^{(\ell)}$ into the boxes $\B_{hor}^{(\ell)}$. A
symmetric statement holds for $\Rve^{(\ell)}$ and $\B_{ver}^{(\ell)}$
for each $\ell\in\{0,...,\left\lfloor \log_{1+\epsilon}N\right\rfloor \}$.
\end{lem}

Thus, we obtain a packing with $(1-O(\epsilon))|\opt'|-2\left(\sum_{\ell}|\B_{hor}^{(\ell)}|+|\B_{ver}^{(\ell)}|\right)$
items in total. Note that $\left(\sum_{\ell}|\B_{hor}^{(\ell)}|+|\B_{ver}^{(\ell)}|\right)\le O_{\epsilon}(\log N)$.
Recall that we assumed that $|\opt'|>c_{\epsilon}\log N$. {Thus,} by choosing
$c_{\epsilon}$ sufficiently large, we can ensure that $\left(\sum_{\ell}|\B_{hor}^{(\ell)}|+|\B_{ver}^{(\ell)}|\right)\le\epsilon\cdot|\opt'|$
and hence our packing contains at least $(1-O(\epsilon))|\opt'|$
items in total.
\begin{lem}
\label{lem:opt-large}For each $\epsilon>0$ there is a constant $c_{\epsilon}$
such that if $|\opt'|>c_{\epsilon}\cdot\log N$ we can compute a
solution of size $(1-\epsilon)|\opt'|$ in time $(nN)^{O_{\epsilon}(1)}$.
\end{lem}

\section{\label{sec:DP-color-coding}Dynamic programming with color coding}

Assume that $|\opt'|\le c\cdot\log N$ for some given constant $c$
(which we will {later choose to be the} constant $c_{\epsilon}$ defined
in Section~\ref{sec:large-OPT}). 
We describe an algorithm that computes
a solution of size $|\opt'|$ for this case in time $(nN)^{O(c)}$
such that each item of this solution is contained in a corridor in
$\C$. Our strategy is to use color-coding~\cite{alon1995color} in order to reduce
the setting of $O_{\epsilon}(1)$ L- and U-corridors in $\C$ to the
setting of only one single such corridor. Then we show how to solve
this problem in polynomial time.

First, we guess $|\opt'|$. Then we color each item in $\R$ randomly
with one color in $\left\{ 1,...,|\opt'|\right\} $. {It is easy to} show
that with probability at least $1/e^{|\opt'|}\ge\frac{1}{N^{O(c)}}$
all items in $|\opt'|$ have different colors, in which
case we say that the coloring was \emph{successful}. If this is the
case, then for each color $d\in\left\{ 1,...,|\opt'|\right\} $ we
can guess in time $O_{\epsilon}(1)$ which corridor in $\C$ contains
an item of $\opt'$ that we colored with color $d$. This yields $O_{\epsilon}(1)^{|\opt'|}=N^{O_{\epsilon}(c)}$
guesses overall. By repeating the random coloring $N^{O(c)}$ times,
we can ensure that, with high probability, one of these colorings was
successful. Also, we can derandomize this procedure using {a $k$-perfect family of hash functions} ~\cite{alon1995color, NSS95}, which yields the following lemma.
\begin{lem}\label{lem:color-coding-LU}
In time $N^{O_{\epsilon}(c)}$ we can guess a partition of $\left\{ \R_{C}\right\} _{C\in\C}$
of $\R$ such that for each corridor $C\in\C$ the set $\R_{C}$ contains
all items from $\opt'$ that are placed inside $C$.
\end{lem}

\subsection{Routine for one corridor}

Recall that we are given a corridor $C\in \C$ and an input set $\R_{C}$ of items colored with $\gamma\leq c\cdot\log N$ colors. {W.l.o.g., let $\{1,\ldots,\gamma\}$ be these colors.} Our goal is to place precisely one item per color inside $C$ such that they do not overlap.
{Let {$\opt'_C$} denote the items of $\opt'$ placed inside $C$ and note that also
$\opt'_C$ contains one item of each color.}

{For our}  $(1.6+\epsilon)$-approximation it is sufficient to consider corridors with up to three sub-corridors; {however,} we will next describe a procedure that works for {corridors with} $k$ sub-corridors {for any $k\leq 1/\eps$}. This extension will actually be needed to {obtain} a $(4/3+\epsilon)$-approximation {(see Section~\ref{sec:improve-ratio})}.

Our strategy is to cut $C$ recursively into pieces (see Figure~\ref{fig:consistent}).
Whenever we make a cut, we guess the items from {$\opt'_C$} that are
intersected by this cut and their placement in {$\opt'_C$}. The cut
splits the considered subpart of $C$ into two pieces and we guess
the colors of the items in {$\opt'_C$} in each {one} of the{se} pieces. Then,
we recursively solve the subproblem defined by each piece. Guessing
the colors ensures that we do not place an item twice, e.g., once
in each of the two subproblems. We define our cuts such that there
are only a polynomial number of possible arising pieces during the
recursion, and we observe that for the guesses of the colors there
are only $2^{\gamma} {\le N^c}$ many options. Hence, we can embed this recursion
into a polynomial time dynamic program. 

\subparagraph{Long chords.}
Formally, whenever we cut $C$ we do this along long chords defined
as follows. A \emph{long chord} is a sequence
of $k$ axis-parallel line segments ${f}_{1},...,{f}_{k}$ that
intuitively connect $e_{0}$ with $e_{k+1}$, i.e., such that for
each $j\in\{1,...,k\}$ each end-point of ${f}_{j}$ has integral
coordinates and coincides with an endpoint of ${f}_{j-1}$ or ${f}_{j+1}$
or lies on $e_{0}$ or $e_{k+1}$, see Figure~\ref{fig:consistent}. Note that there
are two special long chords that go along the edges of $C$, defined
by $\ell_{R}:=e_{1},...,e_{k}$ and $\ell_{L}:=e_{k+{2}},...,e_{2k{+}1}$.

We can compute a set $\L$ containing all of the at most $N^{O(k)}$
long chords. We fix an (unknown) partition of $C$ into $s(C)=:k$
subcorridors $S_{1},...,S_{k}$ that is nice for {$\opt'_C$}. We are
particularly interested in the long chords ${f}_{1},...,{f}_{k}$
in $\L$ that have the property that for each $j\in\{1,...,k\}$ the
line segment ${f}_{j}$ is contained in $S_{j}$ and it is parallel
to the two parallel edges that define $S_{j}$ (see Figure~\ref{fig:consistent}). We
say that such a long chord is \emph{consistent with $S_{1},...,S_{k}$
}(or just \emph{consistent }for short). Note that we do not know $S_{1},...,S_{k}$
and hence we cannot determine whether a given long chord is consistent
or not. However, there are two key observations
\begin{itemize}
\item we can subdivide $C$ recursively along the long chords such that
each arising piece is defined as the area enclosed by two given long
chords and $e_{0}$ and $e_{k+1}$ (see Figure~\ref{fig:consistent}),
\item each consistent long chord can intersect 
with at most $k/\epsl$ items
in {$\opt'_C$} since the corridors are thin and all input items are
skewed.
\end{itemize}
\begin{figure}
	\centering
	\includegraphics[width=0.7\linewidth]{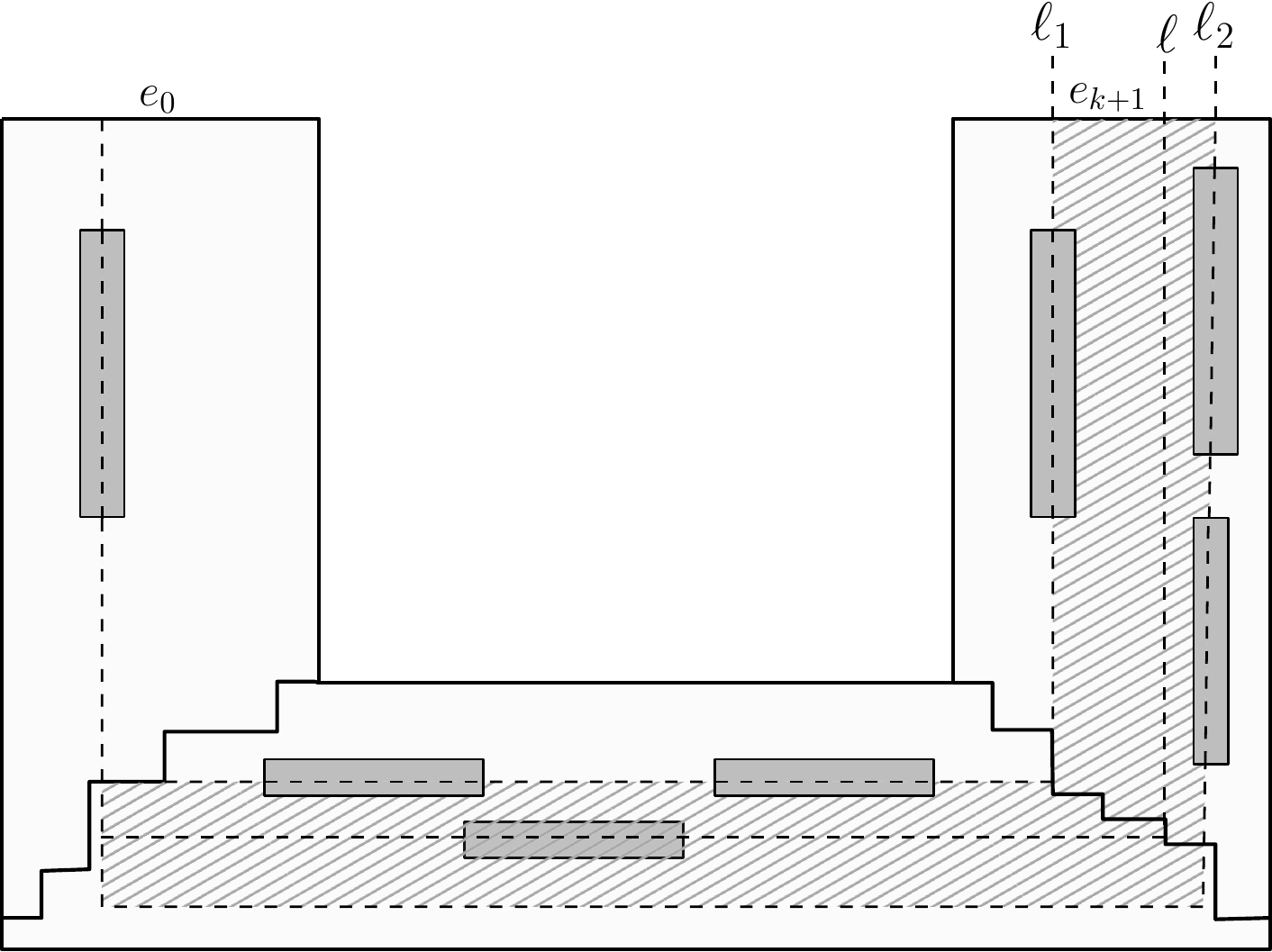}
	\caption{A U-corridor and two consistent long chords $\ell_1$ and $\ell_2$. The long chords intersect only a constant number of items from the optimal solution (shown in the figure). There is a DP-cell that is defined by $\ell_1$, $\ell_2$, and these items, and whose corresponding area is shaded. This cell splits into smaller subproblems by defining a new long chord $\ell$ that lies in-between $\ell_1$ and $\ell_2$.}
	\label{fig:consistent}
\end{figure}

\subparagraph{Subproblems of DP.}
Therefore, we can compute a recursive partitioning of $C$ via a dynamic
program. Each cell of the DP-table is defined by 
\begin{itemize}
\item two long chords $\ell_{1},\ell_{2}\in\L$ that might intersect but
that do not properly cross each other; together with {({a} part of)} $e_{0}$ and
$e_{k+1}$
they define a polygon $C'\subseteq C$,
\item a set of $O(k/\epsl)$ items $\R'_{C}\subseteq\R_{C}$ with a non-overlapping
placement of them inside $C$ such that the interior of each item
in $\R'_{C}$ intersects $\ell_{1}$ or $\ell_{2}$, and 
\item a set of colors $\Gamma\subseteq\left\{ 1,...,\gamma\right\} $.
\end{itemize}
The subproblem encoded in this cell is to place items from $\R_{C}$
inside $C'$ such that they do not overlap with the items in $\R'_{C}$
and such that for each color $d\in\Gamma$ we place exactly one item
of color $d$. If this subproblem has a solution $\opt(\ell_{1},\ell_{2},\R'_{C},\Gamma)$, we store it in the corresponding DP-cell; otherwise we store {\emph{fail}}.

To compute such a solution, we consider any long chord $\ell\in\L$
that lies completely inside $C'$ but is not identical to $\ell_{1}$
or $\ell_{2}$ {(we would like to select a consistent long chord; 
however, we do not know which long chords are consistent and hence
we try all of them)}. 
Let us first assume that at least one such $\ell$ exists. Note that $\ell$ divides
$C'$ into two smaller polygons $C'_{1},C'_{2}$ that are surrounded
by $\ell_{1}$ and $\ell$, and by $\ell$ and $\ell_{2}$, respectively. Then we consider any subset of items $\R''_C\subseteq \R_C$ and a placement of such items inside $C'$ such that: (1) $\R''_C$ are pairwise non-overlapping and not overlapping with $\R'_C$, (2) they are intersected by $\ell$ in their interior, and (3) have distinct colors $\Gamma_\ell\subseteq \Gamma$. Finally, we consider any partition $\Gamma_{1}\dot{\cup}\Gamma_{2}$ of the remaining colors $\Gamma\setminus \Gamma_\ell$. Let $\R'{}_{C,1}$
and $\R'{}_{C,2}$ be the items in $\R'{}_{C}\cup\R''{}_{C}$ that
intersect $C'_{1}$ and $C'_{2}$, respectively. We consider the DP-cells $(\ell_{1},\ell,\R'{}_{C,1},\Gamma_{1})$
and $(\ell,\ell_{2},\R'{}_{C,2},\Gamma_{2})$ and, {if none of them {contains the value ``fail''}}, we store in $(\ell_1,\ell_2,\R'_C,\Gamma)$ the union of $\R''_C$, $\opt(\ell_{1},\ell,\R'{}_{C,1},\Gamma_{1})$, and $\opt(\ell,\ell_{2},\R'{}_{C,2},\Gamma_{2})$ {(together with the placement of the corresponding items) and halt {the computation for the considered DP-cell}}. If the above event never happens, we store ``fail'' in {this DP-cell}.

The {base cases} of the DP are given by pairs $\ell_1,\ell_2$ which are at most one unit apart from each other (everywhere
inside $C$), so that it is not possible to define any long chord $\ell$ between $\ell_{1}$ and $\ell_{2}$ (recall that the endpoints of the line segments of the long chords have integral coordinates). Notice however that in this case at most $O(k/\epsl)$ skewed items can fit {inside} $C'$, hence we can determine whether a feasible solution $\opt(\ell_{1},\ell_{2},\R'_{C},\Gamma)$ exists by {enumeration} in time $(nN)^{O(k/\epsl)}$.

At the end we output the solution stored in the cell $(\ell_{L},\ell_{R},\emptyset,\left\{ 1,...,\gamma\right\} )$.
We will show that this is the optimal solution for $C$. The number
of DP-cells is bounded by $(nN)^{O(k/\epsl)}\cdot2^{\gamma}$ and
the number of possible guesses when computing the entry of a DP-cell
is bounded by $(nN)^{O(k/\epsl)}\cdot2^{O(\gamma)}$. This allows
us to bound the running time of our DP.
\begin{lem}
\label{lem:DP-coloring}Given a path corridor $C$ with $k$ subcorridors
and a set of skewed items $\R_{C}$ {with $\gamma$ distinct colors}. In time $(nN)^{O(k/\epsl)}\cdot2^{O(\gamma)}$ we can determine whether
there exists a set $\R'_{C}\subseteq\R_{C}$ {with $\gamma$ distinct colors} that fits non-overlappingly inside $C$.
\end{lem}

We apply Lemma~\ref{lem:DP-coloring} to each corridor $C\in\C$
which yields the following lemma.
\begin{lem}
\label{lem:opt-small}Assume that $|\opt'|\le c\cdot\log N$ for some
constant $c$. Then we can compute a solution of size $|\opt'|$ in
time $(nN)^{O(c)}$.
\end{lem}

Now Lemmas~\ref{lem:opt-large} and \ref{lem:opt-small} yield the
following theorem.
\begin{thm}
{There is an $(1.6+\epsilon)$-approximation algorithm with a running time of $(nN)^{O_{\epsilon}(1)}$ for unweighted instances of 2DK with only skewed items.}

\end{thm}

\section{Improved approximation ratio and weighted case}

\label{sec:improve-ratio}

In this section we show how to improve our approximation ratio to
$4/3+\epsilon$, even in the weighted case. For any set of items $I'\subseteq I$
we define $p(I'):=\sum_{i\in I'}p(i)$.

We use a slightly different definition of path and cycle corridors.
Recall that in the unweighted case each path corridor is defined via
edges $e_{0},...,e_{2k-1}$. We imposed the condition that each of
these edges has length at least $\epsl/2$ that the line segment $\ell_{i}$
connecting $e_{i},e_{2k-i}$ has length at most $\epsilon\cdot\epsl$
and a similar condition for cycle corridors.

Along this section we may refer to subcorridors equivalently as corridor pieces, which we recall are defined via two parallel horizontal edges $e_{1}=\overline{p_{1}p'_{1}}$, $e_{2}=\overline{p_{2}p'_{2}}$, and by monotone axis-parallel curves connecting $p_{1}=(x_{1},y_{1})$ with $p_{2}=(x_{2},y_{2})$ and connecting $p'_{1}=(x'_{1},y'_{1})$ with $p'_{2}=(x'_{2},y'_{2})$, respectively (the definition for vertical pieces is symmetric). Furthermore, we say that a corridor piece $P$ is \emph{acute} if $x_{1}\le x_{2}\le x'_{2}\le x'_{1}$ or $x_{2}\le x_{1}\le x'{}_{1}\le x'{}_{2}$ and $P$ is an \emph{obtuse
piece }otherwise.

\paragraph*{Relatively thin corridors.}

We will now require stronger conditions, which are intuitively
that $\ell_{i}$ is relatively\emph{ }short compared to $e_{i}$ and
$e_{2k-i}$. Also, we require that inside each horizontal corridor
piece $P$ each item $j\in\opt'$ has small height compared to $P$,
and additionally $j$ either is wide compared to $P$ (like a horizontal
item) or has very small width compared to $P$ (like a small item).
We require a similar condition for vertical corridor pieces. Formally,
let $P$ be a horizontal corridor piece, defined via two horizontal
edges $e_{1}=\overline{p_{1}p'_{1}}$ and $e_{2}=\overline{p_{2}p'_{2}}$,
and additionally two monotone axis-parallel curves connecting $p_{1}=(x_{1},y_{1})$
with $p_{2}=(x_{2},y_{2})$ and connecting $p'_{1}=(x'_{1},y'_{1})$
with $p'_{2}=(x'_{2},y'_{2})$, respectively. Note that $y_{1}=y'_{1}$
and $y'_{2}=y'_{2}$. We define the height of $P$ to be $\height(P):=|y_{1}-y_{2}|$
(intuitively the length of $\ell_{i}$ above) and the width of $P$
to be $\width(P):=\left|\max\{x'_{1},x'_{2}\}-\min\{x_{1},x_{2}\}\right|$
(where we assume w.l.o.g. that $x_{1}\leq x_{1}'$ and $y_{1}<y_{2}$).
We say that $P$ is \emph{relatively thin }if $\height(P)\le\epsilon\cdot\width(P)$
and additionally $|x_{1}-x'_{1}|\le\epsilon\cdot\width(P)$ and $|x_{2}-x'_{2}|\le\epsilon\cdot\width(P)$.
We use an analogous definition for vertical corridor pieces.

We say that a corridor $C$ is \emph{relatively thin} if there is
a partition of $C$ into $s(C)$ relatively thin pieces. Furthermore,
for a relatively thin horizontal corridor piece $P$ and a set of
items $\R'$, we say that $\R'$ \emph{fits well into }$P$ if $\R'$
can be placed non-overlappingly inside $P$ and also for each item
$i\in\R'$ it holds that $\height(i)\le\eps^{4}\cdot\height(P)$
and additionally the existence of constants $0<\epss'<\epsl'$
that differ by a large enough factor such that either $\width(i)>\epsl'\cdot\width(P)$
(and then $i$ is intuitively a horizontal item inside $P$) or $\width(i)\le\epss'\cdot\width(P)$
(and then $i$ is intuitively a small item inside $P$). Again,
we use an analogous definition for vertical corridor pieces.

\begin{figure}
\begin{centering}
\includegraphics[scale=0.5]{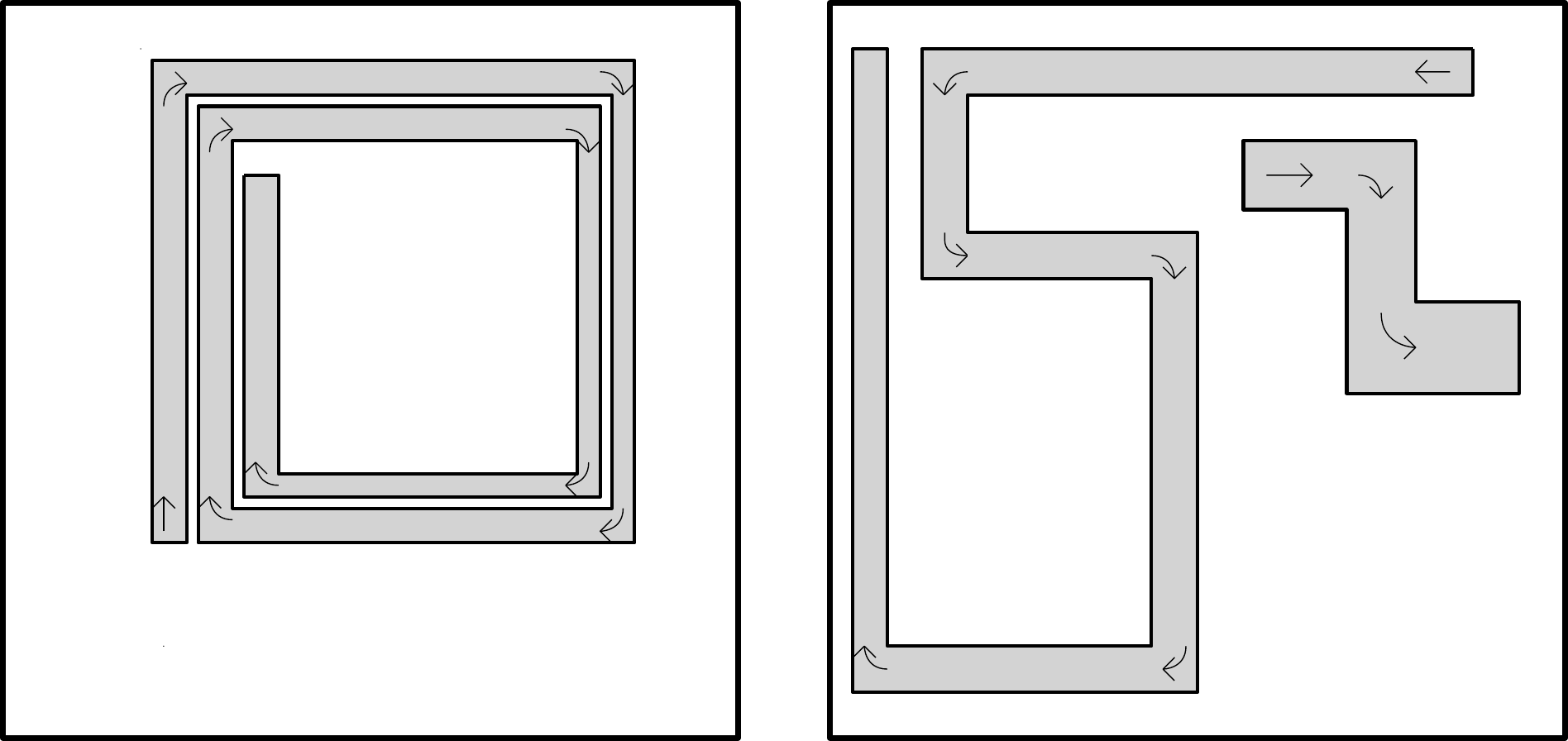} 
\par\end{centering}
\caption{\label{fig:spiral-2-spirals}Left: a spiral where we we turn right
at each bend, i.e., we do not change the direction while traversing
the corridor. Right: two 2-spirals where we change the direction once
when we traverse the corridors.}
\end{figure}

\paragraph*{Spirals and 2-spirals.}

Also, we allow more general shapes than LU-corridors, namely spirals
and 2-spirals. Intuitively, a spiral is a corridor with the property
that when we traverse it, either at each bend we turn left or at each
bend we turn right (see Figure~\ref{fig:spiral-2-spirals}). %
{Formally, a \emph{spiral }is a corridor $C$ with $s(C)$ subcorridors
such that for any set of items $\R'$ placed inside $C$, there is
a partition of $C$ into $s(C)$ pieces that is nice for $\R'$ and
in which each piece is an acute piece.} 
LU-corridors are special cases of spirals.
Intuitively, a 2-spiral is a corridor with the property that when
we traverse it we turn left (right) at each bend until we pass an
obtuse piece and afterwards we turn right (left) at each bend. %
{Formally, a \emph{2-spiral }is a corridor $C$ with $s(C)$
subcorridors such that for any set of items $\R'$ placed inside $C$,
there is a partition of $C$ into $s(C)$ pieces that is nice for
$\R'$ and in which exactly one piece is an obtuse piece and all other
pieces are acute pieces.} We also refer to \emph{LUZ-corridors} as corridors having at most $3$ pieces, and it is easy to see that LUZ-corridors are always spirals or $2$-spirals%

 If a (2-)spiral $C$ is relatively thin then we say that it is a
\emph{relatively thin (2-)spiral}. 
\begin{lem}
\label{lem:partition-knapsack-improved-new}There are values $0<\epss<\epsl\le1$
with $\epss\le\epsilon^{2}\epsl$ with the following properties. There
exists a solution $\opt'\subseteq\opt$ with $p(\opt)\le(4/3+\epsilon)p(\opt')$
and a partition of $K$ into a set $\tilde{\C}$ with $|\tilde{\C}|\le O_{\epsilon}(1)$
where each $C\in\tilde{\C}$ is a box, a relatively thin spiral, or
a relatively thin 2-spiral, such that 
\begin{itemize}
\item each item $i\in\opt'$ is contained in one element of $\tilde{\C}$, 
\item each corridor $C\in\tilde{\C}$ can be partitioned into $s(C)$ corridor
pieces $\P(C)$ such that 
\begin{itemize}
\item each item $i\in\opt'$ inside $C$ is contained in one piece $P\in\P(C)$, 
\item for each $P\in\P(C)$ it holds that the items of $\opt'$ inside $P$
fit well into $P$ (according to $\epss$ and $\epsl$), 
\item no obtuse piece $P\in\P(\C)$ intersects small items (small
with respect to the dimensions of $P$).
\end{itemize}
\item each box $C\in\tilde{\C}$ either contains only one single item $i\in\opt'$
or contains a set of items $\opt'(C)\subseteq\opt'$ that fit well
into $C$. 
\end{itemize}
Also, the pair $(\epsl,\epss)$ is one pair from a set of $O_{\eps}(1)$
pairs which can be computed in polynomial time. 
\end{lem}

We will give the proof of Lemma~\ref{lem:partition-knapsack-improved-new}
in Section~\ref{sec:Improved-partitioning-corridors}. We guess
$\epsl$ and $\epss$ and $\tilde{\C}$ in time $(nN)^{O_{\epsilon}(1)}$.
We need to generalize Lemma~\ref{lem:partition-knapsack} to the more general
case of $\tilde{\C}$ which might contain spirals and 2-spirals (rather
than only LU-corridors like the set $\C$ due to Lemma~\ref{lem:partition-knapsack}).
\begin{lem}
\label{lem:partition-Ls-improved}There is a solution $\opt''\subseteq\opt'$,
a partition $\left\{ \opt''(C),\opt''_{\mathrm{lonely}}(C)\right\} _{C\in\tilde{\C}}$
for it, and for each $C\in\tilde{\C}$ a set of $O_{\epsilon}(\log N)$
boxes $\tilde{\B}(C)$ and a partition $\P(C)$ of $C$ into $s(C)$
pieces, such that 
\begin{itemize}
\item $\opt''(C)\cup\opt_{\mathrm{lonely}}''(C)\subseteq\opt'(C)$ with
$p(\opt''(C))+p(\opt_{\mathrm{lonely}}''(C))\ge(1-\epsilon)p(\opt'(C))$, 
\item the boxes $\tilde{\B}(C)$ and the items $\opt''_{\mathrm{lonely}}(C)$
can be packed non-overlappingly inside $C$ such that apart from $O_{\epsilon}(1)$
elements in $\tilde{\B}(C)\cup\opt''_{\mathrm{lonely}}(C)$, each
element $E\in\tilde{\B}(C)\cup\opt''_{\mathrm{lonely}}(C)$ is contained
in a horizontal (vertical) piece $P\in\P(C)$ such that $\width(E)>\epss\cdot\width(P)$
(such that $\height(E)>\epss\cdot\height(P)$), 
\item the items in $\opt''(C)$ can be nicely packed into the boxes $\tilde{\B}(C)$. 
\end{itemize}
Also, it holds that 
\begin{itemize}
\item $\sum_{C\in\C}|\opt_{\mathrm{lonely}}''(C)|\le O_{\epsilon}(\log N)$, 
\item in time $N^{O_{\epsilon}(1)}$ we can guess the sizes of all boxes
$\tilde{\B}:=\bigcup_{C\in\C}\tilde{\B}(C)$ and a set $\Rlo$ with
$\bigcup_{C\in\C}\opt_{\mathrm{lonely}}''(C)\subseteq\Rlo\subseteq\R\setminus\bigcup_{C\in\C}\opt''(C)$,
and 
\item $p(\opt'')\ge\Omega(|\tilde{\B}|/\epsilon)\cdot\max_{i\in\R\setminus\Rlo}p(i)$. 
\end{itemize}
\end{lem}

We will prove Lemma~\ref{lem:partition-Ls-improved} in Section~\ref{sec:partition-boxes-weighted}.
Now, we {seek to
place the guessed boxes in $\B$, together with a profitable subset
of $\Rlo$. }We will later place items from $\R\setminus\Rlo$ into
the boxes $\B$. We partition the items in $\Rlo$ in the next lemma.

\begin{lem}
\label{lem:assign-items-corridors-improved}In time ${(nN)^{O_{\epsilon}(1)}}$
we can guess a partition of $\Rlo$ into sets $\left\{ \Rlo(C)\right\} _{C\in\tilde{\C}}$
such that $\opt_{\mathrm{lonely}}''(C)\subseteq\Rlo(C)$ for each
$C\in\tilde{\C}$. 
\end{lem}

\begin{proof}
Indeed, notice that by Lemma~\ref{lem:partition-Ls-improved} we
have that $\sum_{C\in\C}|\opt_{\mathrm{lonely}}''(C)|\le O_{\epsilon}(\log N)$,
and therefore we can use color coding with $|\Rlo\cap OPT''|$ colors
in a similar fashion as in Lemma \ref{lem:color-coding-LU}.
With probability $1/N^{O_{\eps}(1)}$ all items in $\Rlo\cap OPT''$
are colored with different colors and we can derandomize this procedure
to a deterministic routine with a running time of $(nN)^{O_{\eps}(1)}$.
 Then, as there are $O_{\eps}(1)$ objects in $\tilde{C}$ according
to Lemma~\ref{lem:partition-knapsack-improved-new}, in the event
that all items in $\Rlo\cap OPT''$ are colored differently we can
guess the right assignment of colors into the objects in $\tilde{C}$
and get a partition $\left\{ \Rlo(C)\right\} _{C\in\tilde{\C}}$ in
time $N^{O_{\eps}(1)}$. 
\end{proof}
Like in the unweighted case, we guess in time $(nN)^{O_{\epsilon}(1)}$
for each box in $\B$ the corridor $C\in\tilde{\C}$ such that $B\in\B(C)$
according to Lemma~\ref{lem:partition-Ls-improved}. Then, for each
corridor $C\in\tilde{\C}$ we place the boxes assigned to it and the
most profitable set of items in $\Rlo(C)$ that fit into $C$ via
the following lemma. 
\begin{lem}
{\label{lem:DP-weighted}Given a corridor $C$, the sizes of
a set of boxes $\bar{\B}$, a set of skewed items $\bar{\R}\subseteq\R$,
and an integer $k$. There is an algorithm with a running time of
$2^{|\B|+k}(nN)^{O_{\epsilon}(1)}$ that finds a packing of $\bar{\B}$
and a subset of $\bar{\R}'\subseteq\bar{\R}$ with at most $k$ items
of maximum total profit, among all packings such that apart from $O_{\epsilon}(1)$
elements in $\bar{\B}\cup\bar{\R}'$, each element $E\in\bar{\B}\cup\bar{\R}'$
is contained in a horizontal (vertical) piece $P\in\P(C)$ such that
$\width(E)>\epss\cdot\width(P)$ (such that $\height(E)>\epss\cdot\height(P)$),
for a suitable partition $\P(C)$ of $C$ into pieces.} 
\end{lem}

Finally, we pack items from $\R\setminus\Rlo$ into the guessed boxes
using the following lemma.

\begin{lem}
\label{lem:pack-into-boxes-weighted} Given a set of boxes $\bar{\B}$
and a set $\bar{I}$ such that a subset $\bar{\R}'\subseteq\bar{I}$
can be nicely packed inside $\bar{\B}$. There is an algorithm with
a running time of $(n|\bar{\B}|)^{O(1)}2^{O(|\B|)}$ that (nicely)
packs a subset $\tilde{\R}'\subseteq\bar{\R}$ into $\bar{\B}$ such
that $p(\tilde{\R}')\ge(1-4\eps)p(\bar{\R}')-|\B|\cdot\max_{i\in\bar{I}}\profit(i)$. 
\end{lem}

Using that $p(\opt'')\ge\Omega(|\B|/\epsilon)\cdot\max_{i\in\R\setminus\Rlo}p(i)$
this completes our improvement to an approximation factor of $4/3+\epsilon$
in the weighted case. 
\begin{thm}
There is an $(4/3+\epsilon)$-approximation algorithm for 2DK with
a running time of $(nN)^{O_{\epsilon}(1)}$. 
\end{thm}

\section{Improved partitioning into corridors}

\label{sec:Improved-partitioning-corridors}

In this section we prove Lemma~\ref{lem:partition-knapsack-improved-new}.
{Our argumentation will be based on the following lemma which follows from~\cite{AHW19}.}
\begin{lem}[Corridor Decomposition Lemma]
	\label{lem:corridorPack-weighted}
	Let $\R$ be a set of items that can be packed inside a given rectangular
	region $\mathcal{K}$ of height and width $T$. Let also $\R'\subseteq\R$
	be a given set of untouchable items, $|\R'|\in O_{\eps}(1)$. Then,
	there exists a corridor partition of $\mathcal{K}$ and a set of items
	$\Rco\subseteq\R$ satisfying: 
	\begin{enumerate}
		\item There exists a set of items $\Rco^{cross}\subseteq\Rco$ such that
		each item in $\Rco\setminus\Rco^{cross}$ is completely contained
		in some corridor of the partition. Furthermore, we have that $\R'\subseteq\Rco\setminus\Rco^{cross}$,
		$a(\Rco^{cross}\cap\Rsm)\le O_{\epsl}(\epss)\cdot a(\mathcal{K})$
		and $|\Rco^{cross}\setminus\Rsm|\in O_{\eps}(1)$. 
		\item $\profit(\Rco)\ge(1-O(\eps))\profit(\R)$. 
		\item The number of corridors is $O_{\eps,\epsl}(1)$ and each corridor
		has at most $\frac{1}{\eps}$ bends and width at most $\epsl N$,
		except possibly for the corridors containing items from $\R'$ which
		correspond to rectangular regions matching exactly the size of these
		items. 
	\end{enumerate}
\end{lem}

{Note that in the previous lemma, there is a set of items $\Rco^{cross}\setminus\Rsm$
that are not small and that are not contained in any corridor. In the unweighted case,
these items are negligible since they are only constantly many (unless the optimal
solution contains only a constant number of items which is a trivial case). However,
in the weighted case those items might have a lot of profit and thus we cannot afford to
lose them. Therefore, we prove}
an alternative version of Lemma~\ref{lem:corridorPack-weighted} 
which intuitively is better suited for the weighted case. 
\begin{lem}
\label{lem:corridor-partition-improved}There are values $0<\epss'<\epsl'\le1$
with $\epss'\le\epsilon^{2}\epsl'$ with the following properties.
There exists a solution $\opt'\subseteq\opt$ with $p(\opt)\le(1+\epsilon)p(\opt')$
and a partition of $K$ into a set $\tilde{\C}$ with $|\tilde{\C}|\le O_{\epsilon}(1)$
where each $C\in\tilde{\C}$ is a box, a relatively thin path corridor,
or a relatively thin cycle corridor, such that 
\begin{itemize}
\item each item $i\in\opt'$ is contained in one element of $\tilde{\C}$, 
\item each corridor $C\in\tilde{\C}$ can be partitioned into $s(C)$ corridor
pieces $\P(C)$ such that 
\begin{itemize}
\item each item $i\in\opt'$ inside $C$ is contained in one piece $P\in\P(C)$, 
\item for each $P\in\P(C)$ it holds that the items of $\opt'$ inside $P$
fit well into $P$ (according to $\epss'$ and $\epsl'$). 
\end{itemize}
\item each box $C\in\tilde{\C}$ either contains only one single item $i\in\opt'$
or contains a set of items $\opt'(C)\subseteq\opt'$ that fit well
into $C$. 
\end{itemize}
Also, the pair $(\epsl',\epss')$ is one pair from a set of $O_{\eps}(1)$
pairs which can be computed in polynomial time. 
\end{lem}

{We will prove Lemma~\ref{lem:corridor-partition-improved} later in Section~\ref{sec:lemma-corridor-improved}}.
Now we will proceed with the proof of Lemma~\ref{lem:partition-knapsack-improved-new}, {assuming Lemma~\ref{lem:corridor-partition-improved}}.
Consider first the corridor partition from Lemma~\ref{lem:corridor-partition-improved}
and the induced pieces from it. In order to construct $\tilde{\C}$
we first add all boxes originally in $\C$. Then, for each relatively
thin path or cycle corridor $C$ we will group {the items inside
the corridor} into four {disjoint} sets according to the pieces
that they belong to, in such a way that if we delete any of these
sets, we can subdivide the corridors and remaining items into {boxes},
spirals and 2-spirals only. {Being that the case we can delete the
least profitable} of the sets, with profit at most $\frac{1}{4}p(C)$,
and pack all items in the remaining pieces in the same way as in the
original packing, concluding the proof.

For each $C\in\tilde{\C}$ and {each piece from the decomposition
of Lemma~\ref{lem:partition-knapsack-improved-new}} we will assign a number $\typeC(P)$
between $1$ and $4$. Now consider an arbitrary path corridor $C=P_{1},\dots P_{k}$,
then we may simply assign to each corridor piece $P_{i}$ by $\typeC(P_{i})=i\mod4$,
and define the four disjoint sets as, for each $i=1,\dots4$, the
items from pieces with $\typeC(P)=i$ which are skewed with respect
to their pieces plus the items from pieces with $\typeC(P_{i})=(i+2)\mod4$
which are small with respect to their pieces. Thus if we consider
$\mathcal{T}_{j}$ to be the set of items assigned to the sets $j=1,..4$,
if we delete any of these sets from $C$, we can further partition
the corridor into only boxes and LUZ-corridors without relatively
small items in the obtuse pieces as in the proof of Lemma~\ref{lem:partition-knapsack}
while losing some negligible profit because of some lost small items
from the boxes.

If $C$ is a cycle corridor instead, composed of the pieces $(P_{1},\dots P_{k})$,
we will divide the analysis into cases. As $C$ is a cycle corridor,
$k$ must be even and therefore we have that either $k\equiv0\mod4$
or $k\equiv2\mod4$. If it is the former case we can proceed exactly
as in the case of path corridors as we will obtain only boxes and
LUZ-corridors.

If it is $k\equiv0\mod4$, we will distinguish two cases. First the
following technical observation that will help to distinguish the
cases.

\begin{remark}\label{rem:U-shapes} For any cycle corridor $C$,
there are four 3-tuples of consecutive corridor pieces denoted by
$U^{i}=(P_{i}^{(1)},P_{i}^{(2)},P_{i}^{(3)})_{i=1}^{4}$ such that
they induce U-corridors contained in $C$ where its central piece
$P_{i}^{2}$ is either the topmost horizontal, leftmost vertical,
rightmost vertical or bottom-most horizontal piece with respect to
the knapsack. \end{remark}

The two cases we consider depend on whether these four induced U-corridors
from Remark~\ref{rem:U-shapes} are disjoint or if they overlap.

Consider a cycle corridor $C=P_{1},\dots P_{k}$ and define $\shapeC(C)$
to be a sequence of $s(C)$ symbols in $\{\text{a,o}\}$ such that
$\shapeC(C)_{i}=\text{a}$ if $P_{i}$ is an acute piece and $\shapeC(C)_{i}=\text{o}$
if $P_{i}$ is an obtuse piece. Now, we first study the case when
$(\text{a,a})$ appears as a substring of $\shapeC(C)$ (w.l.o.g.
we can assume it appears first) then we can assign types as follows (see Figure~\ref{fig:CorridorsnotAA}: 

First we set $\typeC(P_{k})=\typeC(P_{2})=1$
and $\typeC(P_{1})=2$, and then we continue assigning each $P_{i}$
as $\typeC(P_{i})=i-1\mod4$ for $3\leq i\leq k-1$. As $k=2\mod4$,
it must be that $\typeC(P_{k-1})=4$. Under this assignment, if we
hypothetically remove the pieces $P$ such that either $\typeC(P)=1$
or $\typeC(P)=2$, then the resulting corridors we obtain are either
boxes or LUZ-corridors. Instead if we remove all corridor pieces such
that $\typeC(P)=3$ we obtain the corridor LUZ-corridors and a path
corridor formed by three acute pieces, an obtuse piece, and another
acute piece, which is a 2-spiral. If we remove all corridor pieces
$P$ such that $\typeC(P)=4$ then we obtain LUZ-shapes and a path
corridor obtained from 1 acute piece followed by an obtuse piece and
then three acute pieces, which is again a 2-spiral. Having this, we
will group the items in $C$ as follows: for each $i=1,\dots,4$,
we assign to a group the skewed items in pieces with $\typeC(P)=i$
plus the small items from the induced obtuse pieces if we hypothetically
delete the previous pieces. The required properties are satisfied
by construction since it is not difficult to check that no piece becomes
an obtuse piece more than once after each deletion.

If the sequence $(\text{a,a})$ does not appear as a substring of
$\shapeC(C)$, it means that all the U-corridors $(P_{i}^{(1)},P_{i}^{(2)},P_{i}^{(3)})$,
$i=1,\dots,4$, from Remark~\ref{rem:U-shapes} are disjoint. Let
us define $S_{i}$ to be the sequence of consecutive corridor pieces
in $C$ starting at $P_{i}^{(2)}$ and ending at $P_{i+1\mod4}^{(2)}$.
Notice that for each $i$, $S_{i}$ has even length as $S_{i}$ either
starts with an horizontal piece and ends with a vertical piece or
vice versa, therefore there must exist an $i$ such that $S_{i}$
has length $2\mod4$. From here we can classify each corridor piece
as follows:

Let us assume w.l.o.g. that $S_{1}$ consists of $\ell=0\mod4$ pieces.
Then if we let $P_{1}=P_{1}^{(2)}$, we will set $\typeC(P_{i})=i\mod4$
for $1\leq i\leq\ell$, with $\ell$ the last piece in $S_{1}$. It
must be the case that $\typeC({P_{\ell}})=2$. Now consider the permutation
$\pi$ given by $(1,3,4,2)$, and we will assign $\typeC(P_{\ell+i})=\pi(i\mod4)$
for all $i\geq1$. This assignment goes on until all remaining pieces
are considered.

\begin{figure}[ht]
		\centering

		\includegraphics[width=0.47\linewidth]{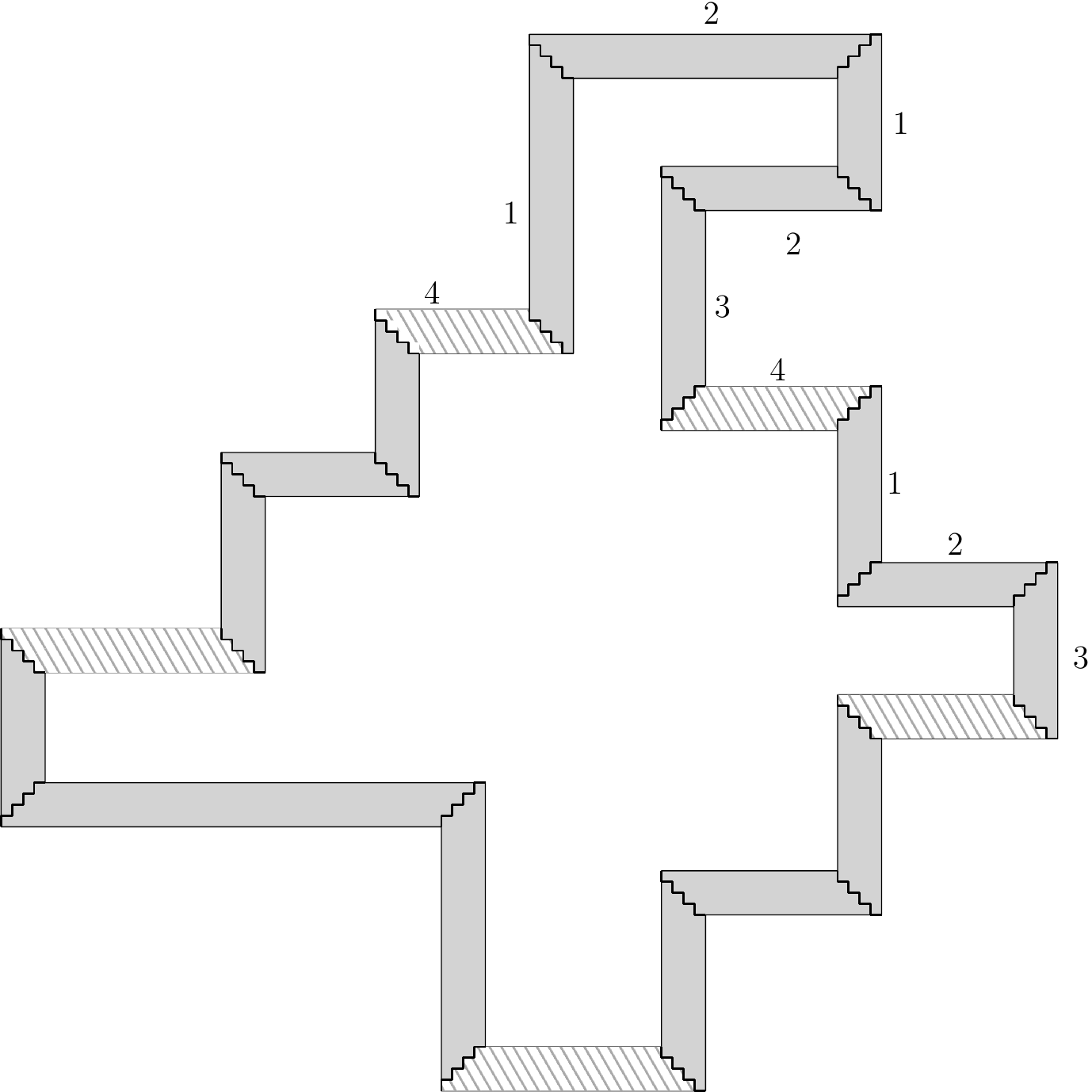}  

	\qquad

		\includegraphics[width=.47\linewidth]{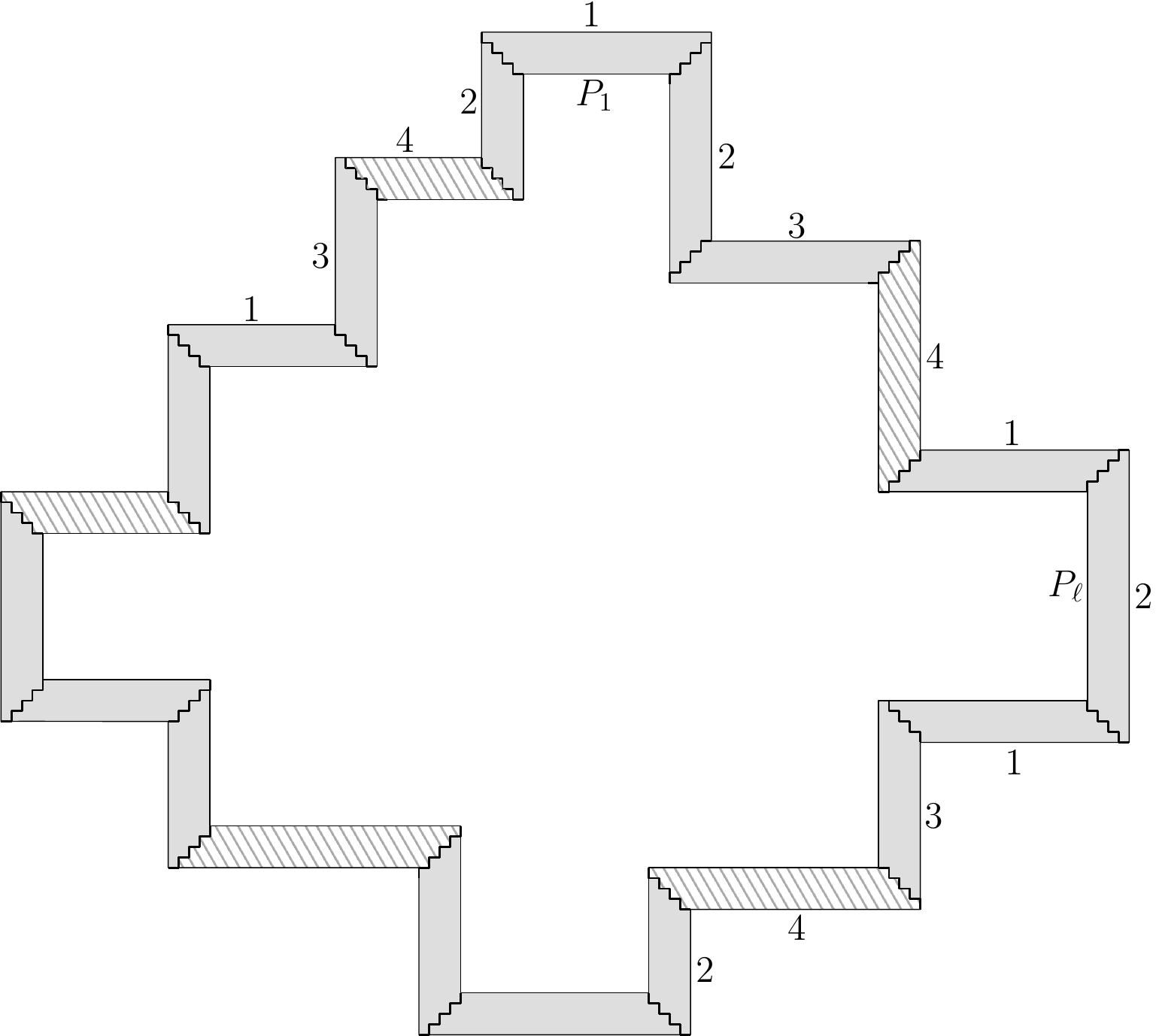}
		\caption{(Left) A cycle corridor with two consecutive acute pieces (on top) and the assigned types. If the pieces of type 4 are removed (dashed regions), we are left only with 2-spirals. (Right) A cycle corridor without consecutive acute pieces. Four disjoint U-subcorridors can be distinguished, and two of them must be at distance $2 \mod 4$ ($P_1$ and $P_{\ell}$). If the pieces of type $4$ are removed (dashed regions) we are left only with $2$-spirals.}
		\label{fig:CorridorsnotAA}
\end{figure}

Let $\bar{S}$ be the sequence of corridor pieces starting from $P_{\ell+1}$
until $P_{k}$. It is not difficult to verify that, if we hypothetically
remove all the pieces of the same type, we are left with only boxes,
spirals and two spirals: this is simple for the pieces of type $1$
and $2$ as we are left only with LUZ-corridors and boxes; if we remove
the pieces of type $3$ (and analogously type $4$) we are left with
only LUZ-corridors plus two 2-spirals involving the disjoint U-corridors
that define $S_{1}$ (see Figure~\ref{fig:CorridorsnotAA}). 

Having
this, again we will group for each $i=1,\dots,4$ the skewed items
in pieces with $\typeC(P)=i$ plus the small items from the induced
obtuse pieces if we hypothetically delete these pieces. Again no piece
becomes an obtuse piece more than once after each deletion, so the
required properties are satisfied.

Finally, we will repack the temporarily removed small items. Indeed,
we have that the total area for the small items removed by the previous
procedure is at most $O_{\epsl}(\epss)N^{2}$, let us denote this
total area by $A_{small}$. Since the knapsack is partitioned into
at least $1/\epsl$ and at most $O_{\epsl}(1)$ many pieces, there
must be at least $1/\eps$ pieces such that a box of height at least
$\frac{\epss}{\eps}N$, width at least $\frac{\epss}{\eps}N$ and
total area at least $(1+2\eps)A_{small}$ can be drawn completely
inside the pieces. By deleting the contents of the least profitable
of these pieces, we can repack all deleted small items into the aforementioned
new box using NFDH (Theorem~\ref{thm:nfdhPack}).

\subsection{Alternative corridor partition}\label{sec:lemma-corridor-improved}
In this section we prove Lemma~\ref{lem:corridor-partition-improved}.

We will first prove the lemma under the assumption that we can
remove $O_{\eps}(1)$ items at no cost. This assumption does not hold
without loss of generality as these items, although being a constant
number, may carry a significant fraction of the profit. However we
will prove how to drop this assumption by standard shifting tricks.

Let $M(0)$ be the set of items that we will delete at no cost during
the argumentation. As there are at most $1/\epsl^{2}$ items in $\optla$
we will just add them to $M(0)$. Let us start initially by applying
Lemma~\ref{lem:corridorPack-weighted} to the optimal solution, which
gives us a corridor partition of $K$ into a constant number of corridors
and boxes and a set $\Rco\subseteq OPT$ of total profit at least
$(1-\eps)p(OPT)$ that can be packed inside the corridors, except
for a set $\optco^{cross}$ consisting of a constant number of skewed
items and a set of small items of total area at most $O_{\epsl}(\epss)N^{2}$.
We will include the skewed items in $\optco^{cross}$ into $M(0)$,
while the small items from that set we will repack later. Let us refine
this partition so as to fulfill all the required extra properties
regarding the relative dimensions of the corridors and the items inside.

Notice that the lines defining the corridors are all longer
than $\frac{\epsl}{2}N$ and
all the pieces of the corridors have height at most $\epsl N$. Let
us start by making the corridors relatively thin, meaning that for
each horizontal (vertical) piece $P$ we will have that $h(P)\le\eps\cdot w(P)$
($w(P)\le\eps\cdot h(P)$) and the width (height) of the neighboring
pieces is at most $w(P)$ ($h(P)$ respectively). 
 Suppose some horizontal piece $P$ is not thin, meaning that its
height is larger than $\eps\cdot w(P)$ or that the width of some
of its neighboring pieces is larger than $\eps\cdot w(P)$ (or the
analogous statement for a vertical piece). We will subdivide the corridor
containing $P$ into $\frac{2}{\eps^{2}}$ thinner corridors in such
a way that each piece $P'$ of the corridor is divided into at least
$\frac{2}{\eps}$ pieces of height/width at most $\frac{\eps}{2}h(P')$
($\frac{\eps}{2}w(P')$), which would result on a relatively thin
corridor as $h(P')\le\epsl N$ and $w(P')\ge\frac{\epsl}{2}N$ for
any horizontal piece $P'$, and the analogous inequalities for vertical
pieces would also hold. This can be achieved as follows: let $P'$
be an horizontal piece in $C$, being the procedure symmetric if the
piece is vertical. We can draw $2/\eps$ horizontal lines across $P'$
equidistantly at distance $\frac{\eps}{2}h(P')$ and then extending
these lines across the corridor without intersecting skewed items
along their short dimension. These lines may intersect some items,
but at most $\frac{1}{\eps\epsl}$ skewed ones as they are crossed
along their long dimension, plus a set of small items of total area
at most $\frac{1}{\eps\epsl}\epss N^{2}$. We apply this procedure
for all the pieces of the corridor, which increases the number of
corridors only by a factor of at most $2/\eps^{2}$. Let $\mathcal{P}'$
be the set of obtained pieces after this procedure, $|\mathcal{P}'|\le O_{\eps}(1)$.
We will add the crossed skewed items, which are at most $\frac{1}{\eps\epsl}|\mathcal{P}'|\le O_{\eps}(1)$
many, to $M(0)$, and will temporarily remove the crossed small items
so as to repack them later. Notice that their total area is at most
$|\mathcal{P}'|\frac{1}{\eps\epsl}\epss N^{2}$.

In the same spirit, we will ensure that the items \emph{fit well}
inside the pieces of the corridor, meaning that for each horizontal
(vertical) piece $P$ the height (width) of the items inside $P$
is at most $\eps^{4}h(P)$ ($\eps^{4}w(P)$ respectively). To this
end we remove, for each horizontal piece $P$, all the items contained
in $P$ whose height is larger than $\eps^{4}h(P)$, and also
we perform the analogous procedure for vertical pieces. Notice that
there are at most $\frac{1}{\eps^{4}\epsl}$ skewed items
which are removed from $P$, and if some small item is removed from
$P$, meaning that $h(P)<\frac{\epss}{\eps^{4}}N$, then the total
area of these removed small items is at most $h(P)\cdot w(P)\le\frac{\epss}{\eps^{4}}N^{2}$.
We do the same for all the pieces in the current partition, again
adding the skewed removed items to $M(0)$ while temporarily keeping
aside the removed small items, of total area at most $|\mathcal{P}'|\frac{\epss}{\eps^{4}}N^{2}$,
to be repacked later. Finally, we will choose values $\epss'$
and $\epsl'$

so that the total profit of items where $w(i)\in(\epss'w(P),\epsl'w(P)]$
is negligible while ensuring that the two values differ by a factor
at least $\eps^{2}$. Indeed, 

these values come from a set of constantly many candidates that can
be computed in polynomial time. Therefore, all the required properties
are satisfied.

Consider now the so far removed small items, let us call them $S(0)$,
whose total area is at most $O_{\epsl}(\epss)N^{2}$. Since $K$ is
partitioned into at least $1/\epsl$ and at most $O_{\epsl}(1)$ pieces,
there must be at least $1/\eps$ pieces such that a box of height
at least $\frac{\epss}{\eps}N$, width at least $\frac{\epss}{\eps}N$
and total area at least $(1+2\eps)a(S(0))$ can be drawn completely
inside the pieces. If we remove the items inside the least profitable
of the pieces, whose total profit is at most $O(\eps)\profit(\opt)$,
we can repack the small items from $S(0)$ into the aforementioned
box using NFDH (Theorem~\ref{thm:nfdhPack}) that fits now inside
the deleted piece. Summarizing, we obtain this way a corridor decomposition
with all the required dimension properties for items of total profit
at least $(1-\eps)\profit(\opt)-\profit(M(0))$.

If the total profit of the items in $M(0)$ is at most $\eps\profit(\opt)$,
then we can safely remove them and obtain the desired partition into
relatively thin corridors. If it is not the case, then we will apply
a shifting argumentation, proceeding again with the same construction
but this time ensuring that the skewed items we tried to remove before
are not removed this time. This may induce a new (disjoint) set of
skewed items that need to be removed, which again if they have small
total profit we can just remove, or otherwise we recourse. After at
most $1/\eps$ iterations we will find a set of small total profit
that can be deleted as all these sets are disjoint by construction.

More in detail, suppose that we are at a further iteration of this
recursion, meaning that so far we have computed disjoint sets $M(0),M(1),\dots,M(t)$
of constantly many skewed items and all of them have significant total
profit. We will again apply Lemma~\ref{lem:corridorPack-weighted}
but this time with the set of untouchable items being $\mathcal{M}(t):=\bigcup_{i=0}^{t}{M(i)}$,
the union of all the items we have tried to remove in previous iterations,
which are in total constantly many. Then we apply again the whole
decomposition process and obtain a constant number of skewed items
that we need to remove: these items form $M(t+1)$. Observe that by
construction, the sets $M(0),M(1),\dots,M(t)$ are pairwise disjoint,
and as a consequence the condition $p(M(t))>\eps\cdot p(OPT)$ can
happen strictly less than $1/\eps$ times.

During this decomposition, as described before, some small items are
temporarily removed that need to be repacked. Aside from them, all
the other items are completely contained in one of the corridors except
for $M(t+1)$. Now we will show how to repack these removed small
items. This time it is not enough to prove that they have negligible
area with respect to $K$, as the set of untouchable items may occupy
almost entirely the knapsack. To this end, we will prove that around
the untouchable items it is possible to repack these small items.

We first define a non-uniform grid $G$ by extending the boundaries
of the items in $\mathcal{M}(t)$ in the optimal solution. This yields
a partition of the knapsack into $O_{\eps}(1)$ rectangular cells,
where each item from $\mathcal{M}(t)$ completely covers one or multiple
cells. Note that items might intersect many cells. We will now classify
the items according to their interaction with the cells. We define
constants $1\ge\epsl''\ge\epss''\ge\Omega_{\eps}(1)$ and denote by
$\R(C)$ the set of items that intersect $C$ for each cell $C$.
Furthermore, $h(C)$ and $w(C)$ denote the height and the width of
the cell $C$ respectively, and $w(i\cap C)$ and $h(i\cap C)$ denote
the height and the width of the intersection of item $i$ with $C$,
respectively. We can then partition $\R(C)$ into $\Rsm(C)$, $\Rla(C)$,
$\Rho(C)$, and $\Rve(C)$ as follows: 
\begin{itemize}
\item $\Rsm(C)$ contains all items $i\in\R(C)$ with $h(i\cap C)\le\epss''h(C)$
and $w(i\cap C)\le\epss''w(C)$, 
\item $\Rla(C)$ contains all items $i\in\R(C)$ with $h(i\cap C)>\epsl''h(C)$
and $w(i\cap C)>\epsl''w(C)$, 
\item $\Rho(C)$ contains all items $i\in\R(C)$ with $h(i\cap C)\le\epss''h(C)$
and $w(i\cap C)>\epsl''w(C)$, and 
\item $\Rve(C)$ contains all items $i\in\R(C)$ with $h(i\cap C)>\epsl''h(C)$
and $w(i\cap C)\le\epss''w(C)$. 
\end{itemize}

We can choose
$\epss''$ and $\epsl''$ in such a way that the items not falling
in any of these categories for any cell $C$ have negligible total
profit, and hence we can safely discard them. For each cell $C$ that
is not entirely covered by some item in $\R$ we add all items in
$\Rla(C)$ that are not contained in $\mathcal{M}(t)$ to $M(t+1)$.

Let us temporarily remove the items which are small with respect to
every cell that they intersect. Imagine that we first stretch the
non-uniform grid into a uniform $[0,1]\times[0,1]$ grid. After this
operation, for each cell $C$ and for each item in $\Rho(C)\cup\Rve(C)\setminus\mathcal{M}(t)$
we know that its height or width is at least $\epsl\cdot\frac{1}{1+2|\mathcal{M}(t)|}$.
We can then apply Lemma~\ref{lem:corridorPack-weighted} with $\mathcal{M}(t)$
being the set of untouchable items which yields a decomposition of
the $[0,1]\times[0,1]$ square into at most $O_{\eps,\epsl}(1)$ corridors.
The decomposition for the stretched $[0,1]\times[0,1]$ square corresponds
to the decomposition for the original knapsack.

We add all items that are not contained in a corridor (at most $O_{\eps}(1)$
many) to $M(t+1)$. Now if we include back the small items that were
temporarily removed before, some of them might not be completely contained
in some corridor. However, this time they are very small with respect
to the cells, and furthermore the total area of unpacked small items
intersecting a cell is very small compared to the area of the cell:
The number of lines defining the corridor partition is $O_{\epsl,\eps}(1)$
while the area of each of these small items is at most $O(\epss)a(C)$.
Since each cell is partitioned into at least $1/\epsl$ pieces of
corridors (this can be ensured thanks to the scaling), we can find
such a piece of negligible total profit such that if we remove these
items, we can place a big enough box to repack the small items from
the cell.

After this construction then we have a corridor decomposition and
a set of items such that every item is contained in some corridor.
We can then again refine this decomposition as explained in the beginning
of the proof in order to ensure that the corridors are relatively
thin and that the items fit well inside their corresponding pieces.
This procedure again removes a constant number of items which we add
to $M(t+1)$ and small items of small total area, which we repack
into the corridors by using similar arguments as before (we argue
about this looking at each cell separately, where items are small
compared to the cells and there are enough corridor pieces to ensure
that items can be repacked at a small loss of profit). This way we
obtain the corridor decomposition with the required properties for
a set of items of total profit at least $(1-\eps)\profit(\opt)-\profit(M(t+1))$.
As mentioned before, after at most $1/\eps$ levels of the recursion
we find a solution of profit at least $(1-O(\eps))\profit(\opt)$.
This concludes the proof of Lemma~\ref{lem:corridor-partition-improved}.

\section{Improved partition into boxes \label{sec:partition-boxes-weighted}}

In this section we prove Lemma~\ref{lem:partition-Ls-improved} in
which the set of corridors $\tilde{\C}$ can contain boxes, spirals,
and 2-spirals. First, we partition $I$ into two subsets $\Rh$ and
$\Rl$. We introduce a constant $c=O_{\epsilon}(1)$ that we will
define later. 
\begin{lem}
\label{lem:shifting-weighted}For any $c=O_{\epsilon}(1)$ there is
a solution $\overline{\opt}'\subseteq\opt'$ and an item $i^{*}\in I$
such that $\Rh=\{i\in I|p(i)>p(i^{*})\}$ and $\Rl=\{i\in I|p(i)\le p(i^{*})\}$
satisfy that 
\begin{itemize}
\item $p(\overline{\opt}')\ge(1-O(\eps))p(\opt')$, 
\item $p(\overline{\opt}')\ge\Omega\left(c\cdot\left(\log(nN)+\left|\overline{\opt}'\cap\Rh\right|\right)\cdot p(i^{*})\right)$, 
\item $\left|\overline{\opt}'\cap\Rh\right|\le O_{c}(\log(nN))$, and 
\item each item $i\in\overline{\opt}'$ it holds that $p(i)\ge\frac{\epsilon}{n}\max_{i'\in\R}p(i')$. 
\end{itemize}
\end{lem}

\begin{proof}
First, we define $\overline{\opt}'$ such that the fourth property
is satisfied and such that we can group the items in $\R$ into groups
$\left\{ I^{(\ell)}\right\} _{\ell\in\mathbb{N}}$ (there might be
items in $\R$ that are not contained in any group) such that (i)
each item $i\in\overline{\opt}'$ is contained in one group $I^{(\ell)}$,
(ii) within each group $I^{(\ell)}$, the profits of the items differ
at most by a factor $c^{1/\epsilon}$, and (iii) for two items $i,i'$
in two different groups $I^{(\ell)},I^{(\ell')}$ the profits $p(i),p(i')$
differ at least by a factor $c$. Using standard shifting steps one
can show that there exists such a solution $\overline{\opt}'\subseteq\opt'$
with $p(\overline{\opt}')\ge(1-O(\eps))p(\opt')$. Then let $\ell^{*}$
be the smallest integer such that $\left|\overline{\opt}'\cap\bigcup_{\ell\ge\ell^{*}}I^{(\ell)}\right|\le\log(nN)c^{1/\epsilon+2}$.
We define $i^{*}$ to be the most profitable item in $I^{(\ell^{*}-1)}$.
By definition, it holds that $\left|\{i\in I|p(i)>p(i^{*})\}\right|=\left|\overline{\opt}'\cap\bigcup_{\ell\ge\ell^{*}}I^{(\ell)}\right|\le\log(nN)c^{1/\epsilon+2}$
and hence the third property is satisfied.

Regarding the second property, first assume that $\left|\overline{\opt}'\cap\bigcup_{\ell\ge\ell^{*}}I^{(\ell)}\right|\ge\log(nN)c^{2}$.
Then
\begin{align*}
	p(\overline{\opt}') & \ge  c\cdot p(i^{*})\left|\overline{\opt}'\cap\bigcup_{\ell\ge\ell^{*}}I^{(\ell)}\right|\\
	& \ge  p(i^{*})\cdot\left(\frac{c\log(nN)}{2}+\frac{c\left|\overline{\opt}'\cap\bigcup_{\ell\ge\ell^{*}}I^{(\ell)}\right|}{2}\right)\\
	& \ge  \Omega\left(\left(c\log(nN)+\left|\overline{\opt}'\cap\Rh\right|\right)p(i^{*})\right).
\end{align*}

On the other hand, assume that $\left|\overline{\opt}'\cap\bigcup_{\ell\ge\ell^{*}}I^{(\ell)}\right|<\log(nN)c$.
Then it must hold that $\left|\overline{\opt}'\cap I^{(\ell^{*}-1)}\right|\ge\log(nN)c^{1/\epsilon+1}/2$.
We obtain that

\begin{align*}
	p(\overline{\opt}') & \ge  p(i^{*})\frac{1}{2c^{1/\epsilon}}\left|\overline{\opt}'\cap I^{(\ell^{*}-1)}\right|\\
	& \ge  p(i^{*})\frac{1}{2}\cdot c\log(nN)\\
	& \ge  p(i^{*})\cdot c\left(\frac{\log(nN)}{4}+\frac{\left|\overline{\opt}'\cap\bigcup_{\ell\ge\ell^{*}}I^{(\ell)}\right|}{4}\right).\\
	& \ge  \Omega\left(c\left(\log(nN)+\left|\overline{\opt}'\cap\Rh\right|\right)p(i^{*})\right).
\end{align*}

\end{proof}
For each $C\in\tilde{\C}$ such that $C$ is a box, we apply the following
lemma. 
Note that it uses that inside each box $C\in\tilde{\C}$ the items
assigned to $C$ in $\opt'$ fit well inside $C$. For each $C\in\tilde{\C}$
denote by $\overline{\opt}'(C)$ the items of $\overline{\opt}'$
contained in $C$. 
\begin{lem}
\label{lem:partition-boxes-weighted}For each box $C\in\tilde{\C}$
and the items from $\overline{\opt}'(C)$ packed inside it,
there is a partition of $C$ into a set of $O_{\epsilon}(1)$ boxes
$\B(C)$, a set $\opt''(C)\subseteq\overline{\opt}'(C)$, and a partition
$\opt''(C)=\dot{\bigcup}_{B\in\B(C)}\opt''(C,B)$ such that 
\begin{itemize}
\item ${\profit(\opt''(C))\ge(1-\eps)\profit(\overline{\opt}'(C))}$, 
\item for each $B\in\B(C)$ the items in $\opt''(C,B)$ can be nicely packed
into $B$, 
\item the boxes $\B(C)$ can be guessed in time $(nN)^{O_{\epsilon}(1)}$. 
\end{itemize}
\end{lem}

\begin{proof}
W.l.o.g. assume that box $C\in\C$ is horizontal, the case of a vertical
box being analogous. Since the items from $\overline{\opt}'(C)$ fit
well inside $C$, we have that the height of each of these items is
at most $\eps^{4}\cdot h(C)$. We partition the box into $2/\eps$
horizontal strips of height $\frac{\eps}{2}h(C)$ each, assuming $2/\eps\in\mathbb{N}$.
Notice that since every item has height at most $\eps^{4}h(C)$, each
item intersects at most two such strips. So there exists a strip $S$
such that the items from $\opt'(C)$ that intersect $S$ have total
profit at most $\eps p(\overline{\opt}'(C))$. We remove all these
items intersected by $S$, and let $\opt_{1}(C)$ be the remaining
items. As mentioned before, $\profit(\opt_{1}(C))\ge(1-\eps)p(\overline{\opt}'(C))$.

Now we have a completely empty strip of height $\frac{\eps}{2}h(C)$
which implies that there exists a packing of $\opt_{1}(C)$ into a
box of size $(1-\frac{\eps}{2})h(C)\times w(C)$. We can now use resource
augmentation (Lemma \ref{lem:augmentPack}) to obtain a nice packing
of $\opt''(C)\subseteq\opt_{1}(C)$ inside $h(C)\times w(C)$ into
$O_{\epsilon}(1)$ boxes $\B'(C)$ such that $p(\opt''(C))\ge(1-O(\eps))p(\overline{\opt}'(C))$.
The obtained boxes can be guessed in time $(nN)^{O_{\eps}(1)}$. 
\end{proof}
Next, for each $C\in\tilde{\C}$ that is not a box, we take each acute
piece $P\in\P(C)$ and apply the following lemma to $P$. 

\begin{lem}
\label{lem:partition-acute-pieces-weighted}
Let $C\in\tilde{\C}$ be a corridor and let $P\in\P(C)$ be an acute
piece. Then in time $(nN)^{O_{\epsilon}(1)}$ we can guess a set of
$O_{\epsilon}(1)$ pairwise disjoint boxes $\tilde{\B}'(P)$ {such
that there is an acute piece $P'\subseteq P$ and} 
\begin{itemize}
\item for each $B\in\tilde{\B}'(P)$ we have that $B\subseteq P$ and $B\cap P'=\emptyset$, 
\item there are pairwise disjoint sets $\overline{\opt}_{1}''(P)\subseteq\overline{\opt}'(P)$,
$\overline{\opt}_{2}''(P)\subseteq\overline{\opt}'(P)$ with 
\begin{itemize}
\item $p(\overline{\opt}_{1}''(P)\cup\overline{\opt}''{}_{2}(P))\ge(1-\epsilon^{2})\left(p(\overline{\opt}'(P))\right)$, 
\item the items in $\overline{\opt}_{1}''(P)$ can be nicely placed inside
$\tilde{\B}'(P)$,%

{} and 
\item the items in $\overline{\opt}_{2}''(P)$ can be nicely placed inside
$P'$. 
\end{itemize}
\end{itemize}
\end{lem}

We apply Lemma~\ref{lem:partition-acute-pieces-weighted}
to each of the $O_{\epsilon}(1)$ acute pieces in $\P:=\bigcup_{C\in\tilde{\C}}\P(C)$.
Let $\tilde{\B}'$ denote the union of all boxes that we guessed,
i.e., $\tilde{\B}':=\cup_{P\in\P:P\,\mathrm{is\,acute}}\tilde{\B}'(C)$.
Similarly, we define $\overline{\opt}_{1}'':=\cup_{P\in\P:P\,\mathrm{is\,acute}}\overline{\opt}_{1}''(P)$,
$\overline{\opt}_{2}'':=\cup_{P\in\P:P\,\mathrm{is\,acute}}\overline{\opt}''{}_{2}(P)$,
and $\overline{\opt}'':=\overline{\opt}_{1}''\cup\overline{\opt}_{2}''$.
Recall that due to Lemma~\ref{lem:partition-knapsack-improved-new}
no obtuse piece $P$ intersects with an item that is small compared
to $P$.

Next, we guess a partition for $\Rh$ into sets $\Rhone,\Rhtwo$ such
that $\Rhone$ is a superset of the items in $\Rh\cap\overline{\opt}_{1}''$
and similarly $\Rhtwo$ is a superset of the other items in $\Rh\cap\overline{\opt}_{2}''$.
\begin{lem} \label{guess-R-prime-weighted}
In time $2^{O(|\overline{\opt}''\cap\Rh|)}n^{O(1)}$ we can guess
sets $\Rhone,\Rhtwo$ such that 
\begin{itemize}
\item $\Rh=\Rhone\dot{\cup}\Rhtwo$, 
\item $\overline{\opt}''{}_{1}\cap\Rh\subseteq\Rhone$, and 
\item $\overline{\opt}''{}_{2}\cap\Rh\subseteq\Rhtwo$. 
\end{itemize}
\end{lem}

\begin{proof}
We will use color coding with
the items in $\Rh$, meaning that we will color all these items randomly
using $|\overline{\opt}''\cap\Rh|$ colors. It is possible to show
that the items in $\overline{\opt}''\cap\Rh$ receive different colors
with probability $1/2^{O(|\overline{\opt}''\cap\Rh|)}$~\cite{cygan2015parameterized}.
If we repeat this coloring procedure $2^{O(|\overline{\opt}''\cap\Rh|)}$
times we can ensure that with high probability one of these runs assigns
different colors to the items in $\overline{\opt}''\cap\Rh$. Let
$\{\R^{(j)}\}{j=1,\dots,|\opt|}$ be the partition induced by the
coloring, we will now guess the color classes of the items in $\overline{\opt}''{}_{1}\cap\Rh$
in time $2^{|\overline{\opt}''\cap\Rh|}$ and let the union of the
items in $\Rh$ having these colors be $\Rhone$, while the rest of
the items will be $\Rhtwo$. It is not difficult to see that the required
properties are satisfied for these sets, and this procedure can also
be derandomized using standard techniques~\cite{NSS95} in time $2^{O(|\opt|)}n^{O(1)}$. 
\end{proof}
To the items in $\Rl$ we apply now a similar procedure as we did
to the items in $\Rsk$ in the unweighted case when $|\opt|>O_{\epsilon}(\log N)$.
For simplicity, partition the items in $\Rl$ into two sets $\Rho,\Rve$
such that $\Rho$ contains each item $i\in\Rl$ such that $\height(i)\le\width(i)$
and $\Rve:=\Rl\setminus\Rho$. We group the items in $\Rl$ into $O_{\epsilon}(\log(nN))$
groups where we group the items in $\Rho$ (in $\Rve$) according
to their densities which are defined as the ratio between their profit
and heights (widths). Formally, for each $\ell\in\mathbb{Z}$ we define
$\Rho^{(\ell)}:=\{i\in\Rl\cap\Rho|\frac{p(i)}{\height(i)}\in[(1+\epsilon)^{\ell},(1+\epsilon)^{\ell+1})\}$
and $\Rve^{(\ell)}:=\{i\in\Rl\cap\Rve|\frac{p(i)}{\width(i)}\in[(1+\epsilon)^{\ell},(1+\epsilon)^{\ell+1})\}$
and observe that for only $O(\log(nN)/\epsilon)$ values $\ell$ the
respective sets $\Rho^{(\ell)},\Rve^{(\ell)}$ are non-empty. Intuitively,
for each $\ell$ the items in $\Rho^{(\ell)}$ (in $\Rve^{(\ell)}$)
essentially all have the same density. Now, for each group $\Rho^{(\ell)}$,
$\Rve^{(\ell)}$ we guess an estimate for $p\left(\Rho^{(\ell)}\cap\opt\right)$
and $p\left(\Rve^{(\ell)}\cap\opt\right)$, respectively. We do this
in time $(nN)^{O_{\epsilon}(1)}$ in a similar fashion as in Lemma~\ref{lem:guess-slices},
using a technique from \cite{chekuri2005polynomial}. 
\begin{lem}
In time $(nN)^{O_{\epsilon}(1)}$ we can guess values $\mathrm{opt}_{hor}^{(\ell)},\mathrm{opt}_{ver}^{(\ell)}$
for each $\ell$ with $\Rho^{(\ell)}\cup\Rve^{(\ell)}\ne\emptyset$
such that 
\begin{itemize}
\item $\sum_{\ell}\mathrm{opt}_{hor}^{(\ell)}+\mathrm{opt}_{ver}^{(\ell)}\ge(1-O(\eps))p(\overline{\opt}'')$
and 
\item $\mathrm{opt}_{hor}^{(\ell)}\le p\left(\overline{\opt}''\cap\Rho^{(\ell)}\right)$
and $\mathrm{opt}_{ver}^{(\ell)}\le p\left(\overline{\opt}''\cap\Rho^{(\ell)}\right)$
for each $\ell$. 
\end{itemize}
\end{lem}

\begin{proof}
First, as $p(\overline{OPT}'')$ is a value between $p_{max}$ and
$n\cdot p_{max}$ (with $p_{max}$ the maximum profit among all items)
we can guess a $(1+\epsilon)$-approximation $P^{*}$ for $p(\overline{OPT}'')$
in time $O(\log_{1+\epsilon}n)$, i.e., $(1-\epsilon)p(\overline{OPT}'')\le P^{*}\le p(\overline{OPT}'')$.
Assume w.l.o.g. that for some $\hat{L}=O(\log(nN)/\epsilon)$ it holds
that for each $\ell\notin[\hat{L}]$ the sets $\Rho^{(\ell)}$ and
$\Rve^{(\ell)}$ are empty.

Now, for each $\ell\in[\hat{L}]$ set $\hat{k}_{hor}^{(\ell)}$ to
be the biggest integer such that $\hat{k}_{hor}^{(\ell)}\frac{\eps}{\hat{L}}\cdot P^{*}\leq p(\overline{OPT}''\cap I_{hor}^{\ell})$.
We define $\hat{k}_{ver}^{(\ell)}$ accordingly. We guess all $O(\hat{L})$
values $\hat{k}_{hor}^{(\ell)},\hat{k}_{ver}^{(\ell)}$ for each $\ell$
in a similar way as in the proof of Lemma~\ref{lem:guess-slices},
adopting an argumentation in \cite{chekuri2005polynomial}. Let $w_{hor}$
be a binary string with at most $O(\hat{L}/\eps)$ digits, out of
which exactly $\hat{L}+1$ entries are 1's and for each $\ell$ the
substring of $w_{hor}$ between the $\ell$-th and $(\ell+1)$-th
1 has exactly $k_{hor}^{(\ell)}$ digits (all 0's). We can guess $w_{hor}$
in time $2^{O(\hat{L}/\epsilon)}=(nN)^{O(1/\epsilon)}$, infer all
values $\hat{k}_{hor}^{(\ell)}$ from it, and define $\mathrm{opt}_{hor}^{(\ell)}=\hat{k}_{hor}^{(\ell)}\cdot\eps P^{*}/\hat{L}$
for each $\ell$. We have that $p(\overline{\opt}'')-\left(\sum_{\ell}opt_{hor}^{(\ell)}+opt_{ver}^{(\ell)}\right)\le O(\epsilon)p(\overline{\opt}'')$
since we make a mistake of at most $\frac{\eps}{\hat{L}}\cdot P^{*}\le\frac{\eps}{\hat{L}}\cdot p(\overline{\opt}'')$
for each of the $\hat{L}$ groups.

We guess the values $\mathrm{opt}_{ver}^{(\ell)}$ in a similar fashion. 
\end{proof}
In other words, in order to obtain a $(1+\epsilon)$-approximation
for each $\ell$ it suffices to pack items from $\Rho^{(\ell)}$ with
a total profit of $\mathrm{opt}_{hor}^{(\ell)}$ (and hence of a total
height of $\frac{\mathrm{opt}_{hor}^{(\ell)}}{(1+\epsilon)^{\ell}}$)
and items from $\Rve^{(\ell)}$ with a total profit of $\mathrm{opt}_{ver}^{(\ell)}$
(and hence of a total height of $\frac{\mathrm{opt}_{ver}^{(\ell)}}{(1+\epsilon)^{\ell}}$).
We guess now boxes for the items in the sets $\Rho^{(\ell)},\Rve^{(\ell)}$
in a similar way as we had done it for the items in $\Rsk$ in the
unweighted case.

Intuitively, for each $\ell$ we slice the items in $\Rho^{(\ell)}\cap\overline{\opt}''$
into horizontal slices of height 1 and profit $(1+\epsilon)^{\ell}$
each, and we want to pack $\frac{\mathrm{opt}_{hor}^{(\ell)}}{(1+\epsilon)^{\ell}}$
of these slices. We select the $\frac{\mathrm{opt}_{hor}^{(\ell)}}{(1+\epsilon)^{\ell}}$
shortest slices among them and round their widths to at most $1/\epsilon$
different values via linear grouping, losing a factor of at most $1+\epsilon$
in the profit. We do a similar operation for each set $\Rve^{(\ell)}$.
The way to place the slices and group them into containers is similar
to the unweighted case. We use
the property $p(\overline{\opt}')\ge\Omega\left(c\left(|\tilde{\C}|\log(nN)+\left|\overline{\opt}'\cap\Rh\right|\right)p(i^{*})\right)$
from Lemma~\ref{lem:shifting-weighted} in order to argue that our
guessed boxes yield enough profit, if we choose $c=O_{\epsilon}(1)$
appropriately.%

\begin{lem}
\label{lem:slices-weighted} Given $\epsilon$, we can compute a value
for $c$ with $c=O_{\epsilon}(1)$ such that in time $(nN)^{O_{\epsilon}(1)}$
we can guess the sizes of $O(|\tilde{\C}|\log N/\epsilon)$ boxes
$\tilde{\B}''$ such that 
\begin{itemize}
\item one can nicely place a set of items $\bar{\R}\subseteq\Rl\cup\Rhone$
with $p(\bar{\R})\ge(1-\epsilon^{2})p(\overline{\opt}'\cap(\Rl\cup\Rhone))$
into $\tilde{\B}'\cup\tilde{\B}''$, and 
\item one can place the boxes $\tilde{\B}'\cup\tilde{\B}''$ together with
the items in $\overline{\opt}''\cap\Rhtwo$ non-overlappingly into
the corridors in $\tilde{\C}$. 
\end{itemize}
\end{lem}

We define $\tilde{\B}:=\tilde{\B}'\cup\tilde{\B}''$, $\opt'':=\bar{\R}\cup(\overline{\opt}''\cap\Rhtwo)$,
and $\Rlo:=\Rhtwo$. This completes the proof of Lemma~\ref{lem:partition-Ls-improved}.

\section{Improved approximation ratio via rotations}\label{sec:rotations}

In this section we present an improved approximation guarantee of
$1.25+\epsilon$ for the cardinality case of two-dimensional geometric
knapsack with rotations (2DKR), i.e. we assume that $\profit(i)=1$
for each item $i\in\R$ and that we are allowed to rotate items by
90 degrees. %

We will prove the following lemma later in this section. 
\begin{lem}
\label{lem:partition-knapsack-rot}There exists a solution $\opt'\subseteq\opt$
with $p(\opt)\le(1.25+\eps)p(\opt')$ and a partition of $K$ into
a set $\hat{\C}$ of $O_{\eps}(1)$ objects where each of them is
either a box, a spiral, or a 2-spiral, such that 
\begin{itemize}
\item each item $i\in\opt'$ is contained in one element of $\hat{\C}$, 
\item each box $B\in\C$ either contains only one single item $i\in\opt'\cap\Rla$,
or only items in $\opt'\cap(\Rho\cup\Rsm)$, or only items in $\opt'\cap(\Rve\cup\Rsm)$, 
\item for each corridor $C\in\hat{\C}$ there is a partition of $C$ into
$s(C)$ pieces $\P(C)$ such that 
\begin{itemize}
\item each obtuse piece in $\P(C)$ does not intersect any item in $\opt'\cap\Rsm$, 
\item each very thin acute piece $P\in\P(C)$ does not intersect any item
in $\opt'\cap\Rsm$. 
\end{itemize}
\end{itemize}
\end{lem}

Like before, for each $C\in\hat{\C}$ let $\opt'(C)\subseteq\opt'$
denote the items from $\opt'$ that are contained in~$C$. We apply
the following lemma to $\hat{\C}$ whose proof is essentially identical
to the proof of Lemma~\ref{lem:partition-Ls-improved}
(the only difference is that we can rotate items, which however affects
the proofs only marginally). 
\begin{lem}
\label{lem:partition-Ls-rotation}Assume that $|\opt'|\ge\Omega_{\epsilon}(1)$.
There is a solution $\opt''\subseteq\opt'$, a corresponding partition
$\left\{ \opt''(C),\opt''_{\mathrm{lonely}}(C)\right\} _{C\in\hat{\C}}$
, and a set of $O_{\epsilon}(\log N)$ boxes $\B(C)$ for each $C\in\hat{\C}$
such that 
\begin{itemize}
\item $\opt''(C)\cup\opt_{\mathrm{lonely}}''(C)\subseteq\opt'(C)$ with
$|\opt''(C)|+|\opt_{\mathrm{lonely}}''(C)|\ge(1-\epsilon)|\opt'(C)|$, 
\item the items in $\opt''(C)$ can be nicely packed into the boxes $\hat{\B}(C)$
such that at most $O_{\epsilon}(1)$ boxes in $\hat{\B}(C)$ contain
an item from $\Rsm$, and 
\item the boxes $\hat{\B}(C)$ and the items $\opt''_{lonely}(C)$ can be
packed non-overlappingly inside $C$. 
\end{itemize}
Also, it holds that 
\begin{itemize}
\item $\sum_{C\in\hat{\C}}|\opt_{\mathrm{lonely}}''(C)|\le O_{\epsilon}(\log N)$, 
\item in time $N^{O_{\epsilon}(1)}$ we can guess the sizes of all boxes
$\B:=\bigcup_{C\in\hat{\C}}\B(C)$ and a set $\Rlo$ with $\bigcup_{C\in\hat{\C}}\opt_{\mathrm{lonely}}''(C)\subseteq\Rlo\subseteq\Rsk\setminus\bigcup_{C\in\hat{\C}}\opt''(C)$,
and 
\item $|\opt''|\ge\Omega(|\B|/\epsilon)$. 
\end{itemize}
\end{lem}

\begin{proof}
Since we can rotate items by 90 degrees, we can assume w.l.o.g.~that
$\Rve=\emptyset$ and hence $\Rsk=\Rho$. Lemmas~\ref{lem:partition-boxes-weighted},
\ref{lem:partition-acute-pieces-weighted}, and \ref{guess-R-prime-weighted} hold
accordingly. In Lemma~\ref{lem:estimate-sizes-1} we do not guess the
values $\opt_{ver}^{(\ell)}$ since $\Rve=\emptyset$. We adjust the
proof of Lemma~\ref{lem:slices-weighted} such that we additionally guess
for each box in $\hat{\B}$ whether it needs to be rotated or not.
We can do this in time $2^{|\hat{\B}|}=(nN)^{O_{\epsilon}(1)}$. The
rest of the proof of Lemma~\ref{lem:slices-weighted} stays unchanged. 
\end{proof}

We guess the boxes $\B$ and like before (in the case without rotations)
we want to place $\B$ into the corridors and boxes in $\hat{\C}$,
together with as many items from $\Rlo$ as possible. We guess the
correct assignment of the boxes in $\B$ to the corridors and boxes
in $\hat{\C}$ in time $(Nn)^{O_{\epsilon}(1)}$. We guess a partition
the items in $\Rlo$ in the next lemma,  which can be proven exactly like Lemma~\ref{lem:assign-items-corridors-improved}. 
\begin{lem}
\label{lem:assign-items-corridors-improved-1}In time $N^{O_{\epsilon}(1)}$
we can guess a partition of $\Rlo$ into sets $\left\{ \Rlo(C)\right\} _{C\in\hat{\C}}$
such that $\opt_{\mathrm{lonely}}''(C)\subseteq\Rlo(C)$ for each
$C\in\hat{\C}$. 
\end{lem}

\begin{proof}
Indeed, as $\sum_{C\in\hat{\C}}|\opt_{\mathrm{lonely}}''(C)|\leq O_{\eps}(\log N)$,
we can use color coding with $O_{\epsilon}(\log N)$ colors, assigning all in $\Rlo$ a different color with probability $1/N^{O_{\eps}(1)}$.
Then, as there are $O_{\eps}(1)$ objects in $\hat{C}$ according
to Lemma~\ref{lem:partition-knapsack-rot}, in the event that all
items in $OPT''$ are colored differently we can guess the right assignment
of colors into the objects in $\hat{C}$ and get a partition $\left\{ \R(C)\right\} _{C\in\hat{\C}}$
in time $N^{O_{\eps}(1)}$. By repeating this procedure $N^{O_{\eps}(1)}$
times we can obtain such an assignment with good probability. Finally,
as in Lemma~\ref{lem:color-coding-LU}, we can derandomize
this procedure to a deterministic routine with a running time of $(nN)^{O_{\eps}(1)}$. 
\end{proof}
Then, for each corridor $C\in\hat{\C}$ we place the boxes assigned
to it together with the maximum number of possible items from $\Rlo(C)$
that fit into $C$ via Lemma~\ref{lem:DP-weighted}. Finally, we pack items
from $\R\setminus\Rlo$ into the guessed boxes $\B$ using Lemma~\ref{lem:pack-into-boxes-weighted}.
This completes our improvement to an approximation factor of $1.25+\epsilon$. 
\begin{thm}
There is an $(1.25+\epsilon)$-approximation algorithm for the unweighted
case of 2DK with rotations with a running time of $(nN)^{O_{\epsilon}(1)}$. 
\end{thm}

\subsection{Partitioning into corridors in rotational case}

In this section we prove Lemma~\ref{lem:partition-knapsack-rot}.
We can assume that $|\opt|\ge\Omega_{\epsilon}(1)$, i.e. $|\opt|$
is larger than any given constant (that might depend on $\epsilon$),
since otherwise we can simply pack each item into a separate box for
it. Hence, by losing only a factor of $1+\epsilon$ in the number
of packed items, we can assume that $\opt\cap\Rla=\emptyset$ since
$|\opt\cap\Rla|\le1/\epsl^{2}$.

We start with the corridor partition due to Lemma~\ref{lem:corridor-partition}
and obtain a partition $\bar{\C}$ into boxes, corridors, and cycles
and a corresponding (essentially optimal) solution $\opt'$. Like
before, for each element $C\in\bar{\C}$ let $\P(C)$ denote a partition
of $C$ into $s(C)$ pieces, see Lemma~\ref{lem:corridor-pieces-1}.
Also, for each piece $P\in\P(C)$ denote by $\opt'(P)$ the items
in $\opt'$ that are contained in $P$. To each box $C\in\bar{C}$
we apply Lemma~\ref{lem:partition-boxes-weighted} as before.

We want to partition the corridors and cycles in $\bar{\C}$ and use
the following lemma as a tool for this. Note that we did not use this
lemma (or a similar statement) for the non-rotational case. Intuitively,
the lemma states that given a corridor or cycle $C\in\bar{\C}$ and
an acute piece $P\in\P(C)$ we can remove skewed items from $P$ with
small total height and also some small items in $C$ of small total
area and rearrange the remaining items from $\opt'$, such that then
we can partition $P$ into $O_{\epsilon}(1)$ boxes and $C\setminus P$
into $O_{\epsilon}(1)$ corridors. We say that a piece $P$ is horizontal
(vertical) if the two parallel edges defining it (that are connected
via monotone curves) are horizontal (vertical). 
\begin{lem}
\label{lem:process-piece}Let $C\in\bar{\C}$ and let $P\in\P(C)$
be a horizontal (vertical) acute piece. Let $\epsilon'>0$. There
a set of $O(1/\epsilon')$ boxes $\B(P)$ with equal height (width)
inside $P$, a partition of $C\setminus\B(P)$ into $O_{\epsilon}(1)$
corridors $\bar{\C}(C)$, and 
\begin{itemize}
\item a set $\Rsk'(P)\subseteq\Rsk\cap\opt'(P)$ with $\height(\Rsk'(P))\le\epsilon'h(P)/\epsl$, 
\item a set $\Rsk''(C)\subseteq\Rsk\cap\opt'(C)$ with $|\Rsk''(C)|\le O(1/(\epsilon\cdot\epsilon'\cdot\epsl))$, 
\item a set $\Rsm'(C)\subseteq\Rsm\cap\opt'(C)$ with $\area(\Rsm'(C))\le\epsilon'h(P)N+O\left(\frac{1}{\epsilon\cdot\epsilon'}\epss N^{2}\right)$, 
\end{itemize}
such that $\opt'(C)\setminus(\Rsm'(C)\cup\Rsk'(P)\cup\Rsk'(C))$ can
be assigned into \textup{$\B(P)$} and $\bar{\C}(C)$. 
\end{lem}

\begin{proof}
Assume w.l.o.g.~that $1/\epsilon'$ is an integer. The proof start
with a similar construction as the proof of Lemma~\ref{lem:partition-acute-pieces-weighted}.
Assume w.l.o.g.~that $P$ is horizontal, i.e., $P$ is defined via
two horizontal edges $e_{1}=\overline{p_{1}p'_{1}}$ and $e_{2}=\overline{p_{2}p'_{2}}$,
and additionally two monotone axis-parallel curves connecting $p_{1}=(x_{1},y_{1})$
with $p_{2}=(x_{2},y_{2})$ and connecting $p'_{1}=(x'_{1},y'_{1})$
with $p'_{2}=(x'_{2},y'_{2})$, respectively. Assume w.l.o.g. that
$x_{1}\le x_{2}\le x'_{2}\le x'_{1}$ and that $y_{1}<y_{2}$ (see
Figure~\ref{fig:construction}). Let $h(P)$ denote the the height
of $P$ which we define as the distance between $e_{1}$ and $e_{2}$.
Intuitively, we place $1/\epsilon'$ boxes inside $P$ of height $\epsilon'\cdot h(P)$
each, stacked one on top of the other, and of maximum width such that
they are contained inside $P$. Formally, we define $1/\epsilon'$
boxes $B_{0},...,B_{1/\epsilon'-1}$ such that for each each $j\in\{0,...,1/\epsilon'-1\}$
the bottom edge of box $B_{j}$ has the $y$-coordinate $y_{1}+j\cdot\epsilon'\cdot h(P)$
and the top edge of $B_{j}$ has the $y$-coordinate $y_{1}+(j+1)\cdot\epsilon'\cdot h(P)$
(see Figure~\ref{fig:construction}). For each such $j$ we define
the $x$-coordinate of the left edge of $B_{j}$ maximally small and
the $x$-coordinate of the right edge of $B_{j}$ maximally large
such that $B_{j}\subseteq P$. We define $\Rsk'(P)$ to be all items
from $\Rsk$ that intersect with $B_{0}$ and we define $\Rsm'(P)$
to be all items from $\Rsm$ that intersect with $B_{0}$. Intuitively,
we remove the items in $\Rsk'(P)\cup\Rsm'(P)$ from $\opt'(C)$. Note
that the total height of the items in $\Rsk'(P)$ can be at most $\epsilon'h(P)/\epsl$.
Intuitively, we move down all remaining items in $\opt'(P)$ by $\epsilon'h(P)$
units. Hence, they fit into the boxes $\left\{ B_{0},...,B_{1/\epsilon'-2}\right\} $.

Then, we take the top left corner $p$ of each box $B\in\left\{ B_{0},...,B_{1/\epsilon'-1}\right\} $
and start at $p$ a sequence of at most $s(C)$ line segments that
do not intersect any skewed item in $\opt'(C)$ parallel to its respective
shorter edge and connects $p$ with a boundary of $C$ or with the
right edge of a box in $\left\{ B_{0},...,B_{1/\epsilon'-1}\right\} $.
Since our corridors are thin, this is always possible. We do the same
operation with the top right corner of each box in $\left\{ B_{0},...,B_{1/\epsilon'-1}\right\} $.
This partitions $P$ into $1/\epsilon'$ boxes and $C\setminus P$
into $1/\epsilon'$ corridors. Let $\Rsk''(C)$ denote the set of
skewed items in $\opt'(C)$ that are intersected by one of the constructed
lines or by an edge of a box in $\left\{ B_{0},...,B_{1/\epsilon'-1}\right\} $.
Note that there can be at most $O(\frac{1}{\epsilon\cdot\epsilon'\cdot\epsl})$
many of them.

Moreover, denote by $\Rsm'(C)$ the union of $\Rsm'(P)$ with all
small items in $\opt'$ that are intersected by one of the constructed
line segments. Since we constructed only $O\left(\frac{1}{\epsilon\cdot\epsilon'}\right)$
line segments, the total area of $\Rsm'(C)$ is bounded by $\epsilon'h(P)N+O\left(\frac{1}{\epsilon\cdot\epsilon'}\epss N\right)$. 
\end{proof}
In the sequel, we will say that we \emph{process} some acute piece,
meaning that we apply Lemma~\ref{lem:process-piece} to it. We say
that an item $i\in\opt'$ is \emph{long }if $\height(i)\ge(1/2+2\epsl)N$
or $\width(i)\ge(1/2+2\epsl)N$ and \emph{short }otherwise. Let $\R_{long}$
and $\R_{short}$ denote the long and short input items, respectively.
We say that a piece $P\in\P:=\cup_{C\in\bar{\C}}\P(C)$ is \emph{long}
if in $\opt'$ it contains a long item, and $P$ is \emph{short} otherwise.
Based on this definition, we obtain some properties of the pieces
that we will use later. 
\begin{lem}
\label{lem:long-corridors-properties}Let $C\in\bar{\C}$ and let
$P_{1},...,P_{k}$ be the pieces of $C$ in the order in which they
appear within $C$. 
\begin{enumerate}
\item If $C$ is a corridor then $P_{1}$ and $P_{k}$ are acute. 
\item For each $j,j'$, if $P_{j}$ and $P_{j'}$ are both horizontal (resp.
both vertical) long pieces, then there must be an acute vertical (resp.
horizontal) piece $P_{j''}$ with $j<j''<j'$. 
\item If $C$ is a corridor and contains at most one short piece then $C$
is a spiral or a 2-spiral. 
\end{enumerate}
\end{lem}

\begin{proof}
To show property 1, note that both $P_{1}$ and $P_{k}$ have at least
two corners where the edges meet at $90^{\circ}$. Thus $P_{1}$ and
$P_{k}$ must be acute.

Assume $P_{j}$ and $P_{j'}$ are both horizontal long pieces. Let
$\ell$ be the vertical line at $x=N/2$. Then $P_{j}$ and $P_{j'}$
both intersect $\ell$ as their length is at least $(1/2+\epsl)N$.
Now assume all intermediate vertical pieces are obtuse. 
Then $\ell$ can not intersect both $P_{j}$ and $P_{j'}$. This is
a contradiction. So there must be an acute vertical piece $P_{j''}$
with $j<j''<j'$. This proves property 2.

Before proving property 3, first let us show that if $C$ contains
only long pieces it forms a spiral. For contradiction, assume $C$
contains an obtuse piece $P_{O}$. Now $P_{O}$ has width $\le\epsl$
and its two neighbor pieces have length $\ge(\frac{1}{2}+2\epsl)N$.
Hence, the distance between the left most endpoint and the rightmost
endpoint of $C$ is at least $2(\frac{1}{2}+2\epsl)N-\epsl N>N$.
Which is a contradiction.

Now assume $C$ contains one short piece $P_{S}$. Let $P_{L},P_{R}$
be its neighbor pieces. Say $P_{x}(\neq P_{L},P_{R})$ be a corridor,
the neighbor pieces of $P_{x}$ are both horizontal long (or both
vertical long) pieces. So, using property 2, $P_{x}$ is acute. $P_{L}$
and $P_{R}$ both can not be obtuse as then the distance between the
right most endpoint and the left most endpoint in $C$ is at least
$2(\frac{1}{2}+2\epsl)N-2\epsl N>N$. Hence, at most one of them is
obtuse, implying $C$ to be spiral or 2-spiral. 
\end{proof}
We present now several ways to partition the corridors and cycles
in $\bar{\C}$ and prove later that one of them yields an approximation
ratio of $1.25+\epsilon$.

\subsubsection{Packings 1 and 2}

Let $C\in\bar{\C}$. First, we {\em process} each horizontal acute
piece in $C$. Consider one of the resulting corridors $C'$. By Lemma~\ref{lem:long-corridors-properties},
$C'$ have at most one long piece. Also, by Lemma~\ref{lem:long-corridors-properties}
the first and the last piece of $C'$ are acute. Unless $C'$ consists
of only one piece, one of the latter pieces must be short and we process
this short acute piece. We continue until there is no short acute
piece left to be processed. We do this operation with every corridor
$C\in\bar{\C}$. We obtain a partition $\hat{\C}_{1}$ of $K$ into
boxes. Let $\R'_{skew,1}$ denote the set of all items in the respective
sets $\Rsk'(P)$ when we processed a piece $P$ (see Lemma~\ref{lem:process-piece}).
Intuitively, these items are lost at this moment. Then from Lemma
\ref{lem:process-piece}, $\height(\R'_{skew,1})\le\eps'N/\epsl^{2}$,
as the sum of heights of all pieces $\le N/\epsl$. By choosing $\eps'\le\epsl^{4}\eps N$,
we make $\height(\R'_{skew,1})\le\epsl^{2}\eps N$. Similarly, let
$\R_{small,1}$ denote the set of all items in the respective sets
$\Rsm'(C)$. Now items in $\Rsm'(C)$ comes from piece in $C$ that
have length $\ge\epsl N$. Hence, their total area $\sum_{C}a(\Rsm'(C))\le\sum_{C}\eps'h(P)N+O\left(\frac{1}{\eps\cdot\eps'}\epss N^{2}\right)\le O(\eps')N^{2}$.
Since they have small area, we can put them back into the packing
using the following lemma. 
\begin{lem}
\label{lem:rotsmall} Given a partition $\hat{\C}$ of $K$ into $|\hat{\C}|$
(which is $O_{\eps}(1)$) boxes, spirals, and 2-spirals, a set of
items $\R'$ that are placed inside the objects in $\hat{\C}$, and
a set $\R'_{small}\subseteq\Rsm$ with $\area(\R'_{small})\le\eps'N^{2}$.

Then $(1-\epsilon)|\R'|+|\R'_{small}|$ items from $\R'\cup\R'_{small}$
can be placed inside the objects in $\hat{\C}$. 
\end{lem}

\begin{proof}
Here, we will show the existence of a free empty rectangular region
that has height, width $\ge2\epss$ and area $\ge2\area(\R'_{small})$.
Then we can pack all items in $\R'_{small}$ inside the strip using
NFDH. Let $n_{p}$ be the number of pieces in $\hat{\C}$. As each
corridor can have at most $1/\eps$ pieces, $n_{p}\le|\hat{\C}|/\eps$.
So, there exists a piece $P_{B}$ of area $\ge N^{2}/n_{p}$.  As corridors have width $\le\eps\epsl N$.
There exists a maximal rectangle $R_{M}$ contained in $P_{B}$ of
area at least $N^{2}/n_{p}-2\eps^{2}\epsl^{2}N^{2}$. Let us assume
$w(R_{M})\ge h(R_{M})$. Now we add two constraints: 
\begin{equation}
\eps/(|\hat{\C}|)-2\eps^{2}\epsl^{2}\ge4\epss/\eps\label{roteq1}
\end{equation}
\begin{equation}
(\eps/3)(\eps/(|\hat{\C}|)-2\eps^{2}\epsl^{2})\ge\eps'\label{roteq2}
\end{equation}

From \eqref{roteq1}, we get $h(R_{M})\ge a(R_{M})/N\ge4\epss N/\eps$.
From \eqref{roteq2}, we get $\eps a(R_{M})/3\ge\eps'N^{2}$.

Let us partition $R_{M}$ into $1/\eps$ equal strips of height $\eps h(R_{M})$
and let $R_{F}$ be the strip that contains the minimum number of
rectangles. We remove these rectangles (completely contained inside
$R_{F}$) loosing only $\eps$ fraction of profit. Note that we do
not remove rectangles cut by the boundary of the strip. Now within
this strip we obtain a height of $\eps h(R_{M})-2\epss N\ge4\epss N-2\epss N\ge2\epss N$,
which is completely free. The area of this free region $\ge3\eps'N^{2}-2\epss N^{2}\ge\eps'N^{2}$.
Thus we obtain the required box to pack small items using NFDH. 
\end{proof}
In some of the later packings, we will include $\R'_{skew,1}$ so
that these items are not entirely lost. Also, there is one case in
which we can put back $\R'_{skew,1}\cap\R_{short}$ into our packing
as follows: Let $B_{hor}^{*}$ denote the box $B\in\hat{\C}_{1}$
of maximum height with $w(B)\ge(1/2+2\epsl)N$ and similarly let $B_{ver}^{*}$
denote the box $B'\in\hat{\C}_{1}$ of maximum width with $h(B')\ge(1/2+2\epsl)N$.
Assume that $\max\left\{ h(B_{hor}^{*}),w(B_{hor}^{*})\right\} \ge\epsl^{2}N/3$
and assume w.l.o.g that $h(B_{hor}^{*})\ge\epsl^{2}N/3$. Then we
pack all items in $\R'_{skew,1}\cap\R_{short}$ into $B_{hor}^{*}$
(or $B_{ver}^{*}$) while losing only an $\epsilon$-fraction of the
items packed in $B_{hor}^{*}$ (or $B_{ver}^{*}$). 
\begin{lem}
Assume that $h(B_{hor}^{*})\ge\epsl^{2}N/3$ and let $\opt'(B_{hor}^{*})$
denote the items packed in $B_{hor}^{*}$. Then we can partition $B_{hor}^{*}$
into smaller boxes and pack $(1-\epsilon)|\opt'(B_{hor}^{*})|+|\R'_{skew,1}\cap\R_{short}|$
items into these boxes. 
\end{lem}

\begin{proof}
The proof is similar to Lemma \ref{lem:rotsmall}. We again divide
$B_{hor}^{*}$ into $1/3\eps$ strips and remove the one that contains
smallest number of rectangles. Thus we can obtain an empty strip with
sufficient height to pack rectangles $\R'_{skew,1}$ as a stack. 
\end{proof}
Let $\opt''_{1}\subseteq\opt'_{1}$ denote set of items packed into
$\hat{\C}_{1}$. Packing 2 uses the same strategy, the only difference
is that at the beginning we process each \emph{vertical} acute piece
in each $C\in\bar{\C}$ (and afterwards all remaining short pieces).
Let $\R'_{skew,2}$, $\hat{\C}_{2}$, and $\opt''_{2}$ denote the
respective sets of removed items, the partition of $K$ into boxes,
and the set of packed items.

\subsubsection{Packings 3 and 4}

Let $C\in\bar{\C}$. Let $P_{1}^{S},P_{2}^{S},\dots,P_{k}^{S}$ be
the short pieces of $C$ in the order in which they appear within
$C$. Consider each $P_{j}^{S}$ such that $j$ is even. We delete
all skewed items of $\opt'$ that are contained in $P_{j}^{S}$. We
replace $P_{j}^{S}$ intuitively by the largest box $B_{j}^{S}$ that
fits into $P_{j}^{S}$ and place essentially all small items from
$\opt'$ in $P_{j}^{S}$ inside $B_{j}^{S}$. 
\begin{lem}
There exists a box $B_{j}^{S}\subseteq P_{j}^{S}$ such that $\area(B_{j}^{S})\ge(1-\epsilon)\area(P_{j}^{S})$,
$C\setminus B_{j}^{S}$ consists of one or two corridors, and into
$B_{j}^{S}$ we can assign $(1-\epsilon)\left|\opt'(P_{j}^{S})\cap\Rsm\right|$
items from $\opt'(P_{j}^{S})\cap\Rsm$. 
\end{lem}

\begin{proof}
Similarly to the creation of boxes for small items from Lemma~\ref{lem:partition-knapsack}, we can temporarily remove
the items from the piece, place a box inside whose total area is close
to the area of the piece due to the fact that the pieces are thin,
and place back almost all the small items from the piece by means
of NFDH (Theorem~\ref{thm:nfdhPack}). 
\end{proof}
By Lemma~\ref{lem:long-corridors-properties} each remaining corridor
is a spiral or a 2-spiral. Let $\hat{\text{\ensuremath{\C}}}_{3}$
denote the resulting a partition of $K$ and let $\opt''_{3}\subseteq\opt'$
denote the corresponding set of packed items. Packing 4 is defined
similarly, the only difference is that we do the above operation with
each $P_{j}^{S}$ such that $j$ is odd; let $\hat{\text{\ensuremath{\C}}}_{4}$
denote the resulting partition of $K$ and $\opt''_{4}\subseteq\opt'$
the set of packed items.

\subsubsection{Packing 5}

In our final packing we start with removing items from $\opt'$ such
that we can free up a strip of height $\epsilon N$. Intuitively,
we do this via a random placement of a horizontal or vertical strip
of width $\epsilon N$ and arguing that in expectation we lose only
relatively few items. 
\begin{lem}
\label{lem:randstrip} There exists a solution $\overline{\opt}'\subseteq\opt'$
with at least \[(1-O_{\epsl,\epsst}(1))\left(\frac{\left|\opt'\cap\R_{long}\right|}{2}+\frac{\left|\opt'\cap\R_{short}\right|}{4}\right)\]
items in which no item intersects with $[0,N]\times[0,\epsst N]\subseteq K$. 
\end{lem}

\begin{proof}
Let $X_{H}$ (resp. $X_{V}$) be a random horizontal (resp. vertical)
strip of height (resp. width) $\epsst N$ and width (resp. height)
$1$, fully contained in the knapsack. We choose the bottom (resp.
left) boundary of the strip uniformly at random over $[0,N-\epsst N]$.
Given a packing $\opt'$ of rectangles in knapsack $K$, a rectangle
$i\in\opt'\cap(\Rho\cup\Rsm)$ is intersected by $X_{H}$ with probability
at most $\epsst+2\epsl$. A long (resp.~short) rectangles $i\in\opt'\cap\Rve$
is intersected by $X_{H}$ with probability at most $1$ (resp. $\frac{1}{2}+O_{\epsst,\epsl}(1)$).
Now with equal probability we either choose $X_{H}$ or $X_{V}$,
and remove all rectangles intersected by the strip. Then the profit
of remaining rectangles $\overline{\opt}'$ is at least $(1-O_{\epsl,\epsst}(1))\left(\frac{\left|\opt'\cap\R_{long}\right|}{2}+\frac{\left|\opt'\cap\R_{short}\right|}{4}\right)$.
W.l.o.g. assume the strip is horizontal (otherwise, use rotations).
The items below the strip can be translated vertically by $\epsst N$
amount to obtain a completely free region of $[0,N]\times[0,\epsst N]$. 
\end{proof}
Next, we use some standard techniques in order to place the items
in $\overline{\opt}'$ into $O_{\epsilon}(1)$ boxes, while keeping
a thin strip empty. 
\begin{lem}
There is a set of boxes $\B$ such that the items in $\overline{\opt}'$
are nicely placed in $\B$ and the boxes in $\B$ are placed non-overlappingly
in $[0,N]\times[0,(1-\epsst/2)N]\subseteq K.$ 
\end{lem}

\begin{proof}
Follows from Resource Augmentation Lemma (Lemma \ref{lem:augmentPack}). 
\end{proof}
Recall that for packings 1 and 2 we obtained an improved packing if
$\max\left\{ h(B_{hor}^{*}),w(B_{hor}^{*})\right\} \ge\epsl^{2}N/3$.
Now we obtain an improved packing of packing 5 for the converse case
that $\max\left\{ h(B_{hor}^{*}),w(B_{hor}^{*})\right\} <\epsl^{2}N/3$.
We claim that in this case we can pack into the empty area $[0,N]\times[(1-\eps N/20)N,N]$
(by taking $\epsst=\eps N/40$) all items in $\R'_{skew,1}\cup\R'_{skew,2}\cup(\R_{long}\cap\opt')$:
The items in $\R'_{skew,1}\cup\R'_{skew,2}$ fit into the empty area
because their total height is at most $\epsilon N/10$. Regarding
the items in $\R_{long}\cap\opt'$, observe that packings 1 packs
all items in $(\R_{long}\cap\opt')\setminus\R'_{skew,1}$ in $O_{\epsilon}(1)$
boxes $\hat{\C}_{1}$, and since $\max\left\{ h(B_{hor}^{*}),w(B_{hor}^{*})\right\} <\epsl^{2}N/3$
all long items in these $O_{\epsilon}(1)$ boxes have a total height
of $\epsilon N/10$. Let $\hat{\C}_{5}$ denote the resulting partition
of $K$ into boxes and let $\opt''_{5}\subseteq\opt'$ denote the
set of packed items.

We obtained partitions $\hat{\C}_{1},...,\hat{\C}_{5}$ of $K$ into
boxes, spirals, and 2-spirals, and corresponding sets $\opt''_{1},...,\opt''_{5}$
which are nicely packed into them. In the next lemma we show that
one of these packings yields an approximation ratio of $1.25+\epsilon$
which completes the proof of Lemma~\ref{lem:partition-knapsack-rot}. 
\begin{lem}
It holds that $\max\left\{ |\opt''_{1}|,|\opt''_{2}|,|\opt''_{3}|,|\opt''_{4}|,|\opt''_{5}|\right\} \ge\left(\frac{4}{5}-\epsilon\right)|\opt'|$. 
\end{lem}

\begin{proof}
Let us define, $LT:=(\R'_{skew,1}\cup\R'_{skew,2})\cap\R_{long}$,
$ST:=(\R'_{skew,1}\cup\R'_{skew,2})\cap\R_{short}$, $LF:=(\opt'\cap\R_{long})\setminus LT$,
$SF:=(\opt'\cap\R_{long})\setminus ST$. Also assume $w=|LT|,x=|ST|,y=|LF|,z=|SF|$.
Then $w+x+y+z=(1-\eps)|\opt'|$. We have two cases.

\begin{itemize}
	\item Case 1. $\max\left\{ h(B_{hor}^{*}),w(B_{hor}^{*})\right\} \ge\epsl^{2}N/3$.
	
	Now, using an averaging argument, it follows that 
	\[\max\left\{ |\opt''_{1}|,|\opt''_{2}|\right\} \ge y+w/2+x+z
	\text{ and }\max\left\{ |\opt''_{3}|,|\opt''_{4}|\right\} \ge y+w+(z+x)/2.\]
	
	Hence, $\max\left\{ |\opt''_{1}|,|\opt''_{2}|,|\opt''_{3}|,|\opt''_{4}|,|\opt''_{5}|\right\} $
	\begin{align*}
		& \ge(1-\eps)\max\left\{ y+w/2+x+z,y+w+(z+x)/2,3z/4+y/2+w+x)\right\} \\
		& \ge(1-\eps)\left(\frac{2}{5}(y+w/2+x+z)+\frac{1}{5}(y+w+(z+x)/2)+\frac{2}{5}(\frac{3}{4}z+\frac{1}{2}y+w+x)\right)\\
		& \ge\left(4/5-\epsilon\right)(y+w+z+x)\ge\left(4/5-\epsilon\right)|\opt'|
	\end{align*}
	Here the first inequality follows from the properties of packings
	1-5. The second inequality follows from the fact that $\max\{a,b,c\}\ge\alpha a+\beta b+\gamma c$
	for $\alpha+\beta+\gamma=1,\alpha\ge0,\beta\ge0,\gamma\ge0$.
	\item Case 2. $\max\left\{ h(B_{hor}^{*}),w(B_{hor}^{*})\right\} <\epsl^{2}N/3$.
	\begin{align*}
		\max\left\{ |\opt''_{1}|,|\opt''_{2}|,|\opt''_{5}|\right\} &\ge(1-\eps)\max\left\{ y+w/2+z,\frac{3}{4}z+w+x+y)\right\} \\
		&\ge(1-\eps)\left(\frac{1}{5}(y+w/2+z)+\frac{4}{5}(\frac{3}{4}z+w+x+y)\right)\\
		&\ge\left(\frac{4}{5}-\epsilon\right)(y+w+z+x)\ge\left(\frac{4}{5}-\epsilon\right)|\opt'|
	\end{align*}

\end{itemize}

\end{proof}

\bibliography{bibliography.bib}
 
\appendix

\section{Known subroutines}

\subsection{Next Fit Decreasing Height (NFDH)}

The following theorem summarizes the properties of the NFDH algorithm
\cite{CGJT80} in terms of area covered by the solution, which is
useful when the items are small compared to the region where they
are packed.
\begin{thm}[\cite{CGJT80}]
	\label{thm:nfdhPack} Let $C$ be a rectangular region of height
	$h$ and width $w$. Assume that, for some given parameter $\eps\in(0,1)$,
	for each $i\in\R$ one has $\width(i)\leq\eps w$ and $\height(i)\leq\eps h$.
	Then NFDH is able to pack in $C$ a subset $\R'\subseteq\R$ of area
	at least $a(\R')\geq\min\{a(\R),(1-2\eps)w\cdot h\}$. In particular,
	if $a(\R)\leq(1-2\eps)w\cdot h$, all items in $\R$ are packed. 
\end{thm}

\begin{proof}
	The claim trivially holds if all items are packed. Thus suppose that
	this is not the case. Observe that $\sum_{j=1}^{r+1}\height({i_{n(j)}})>h$,
	otherwise item $i_{n(r+1)}$ would fit in the next shelf above $i_{n(r)}$;
	hence $\sum_{i=2}^{r+1}\height({i_{n(j)}})>h-\height({i_{n(1)}})\geq(1-\eps)h$.
	Observe also that the total width of items packed in each round $j$
	is at least $w-\eps w=(1-\eps)w$, since $i_{n(j+1)}$, of width at
	most $\eps w$, does not fit to the right of $i_{n(j+1)-1}$. It follows
	that the total area of items packed in round $j$ is at least $(w-\eps w)\height({n(j+1)-1})$,
	and thus 
	\begin{eqnarray*}
		a(\R') & \geq & \sum_{j=1}^{r}(1-\eps)w\cdot\height({n(j+1)-1})\geq(1-\eps)w\sum_{j=2}^{r+1}\height({n(j)})\\
		& \geq & (1-\eps)^{2}w\cdot h\geq(1-2\eps)w\cdot h.
	\end{eqnarray*}
\end{proof}


\subsection{Resource Augmentation}

If a set of items can be packed non-overlappingly in a container,
then under resource augmentation it can be packed nicely in a constant
number of boxes within this container. This has been shown in previous
work and we state here a lemma in \cite{GGHIKW17} summarizing this.
\begin{lem}[Resource Augmentation Packing Lemma \cite{GGHIKW17}]
	\label{lem:augmentPack} Let $\R$ be a collection of items that
	can be packed into a box of size $a\times b$, and $\eps_{au}>0$
	be a given constant. Then there exists a {nice} packing  of $\R'\subseteq\R$ inside a
	box of size $a\times(1+\eps_{au})b$ (resp., $(1+\eps_{au})a\times b$)
	such that:
\end{lem}

\itemsep0pt 
\begin{enumerate}
	\item $p(\R')\geq(1-O(\eps_{au}))p(\R)$; 
	\item The number of containers is $O_{\eps_{au}}(1)$ and their sizes belong
	to a set of cardinality $n^{O_{\eps_{au}}(1)}$ that can be computed
	in polynomial time. 
\end{enumerate}

\section{Omitted proofs}

The following lemma ensures that the intermediate items can be neglected
by losing only a factor of $1+\epsilon$ in the approximation ratio while maintaining that $\epsl$ and $\epss$ are lower-bounded by
some constant depending only on $\epsilon$. 

 \begin{lem}\label{lem:item-classification}Given $\epsilon>0$. There are values
 $0<\epss<\epsl\le1$ with $\epss\le\epsilon^{2}\epsl$ such that the
 total number of intermediate rectangles in the optimal solution is bounded by $\eps\cdot|OPT|$.
 The pair $(\epsl,\epss)$ is one pair from a set of $O_{\eps}(1)$
 pairs which can be computed in polynomial time. 
 \end{lem}
\begin{proof}
	Define $k+1=2/\eps+1$ constants $\eps_{1},\ldots,\eps_{k+1}$, such
	that $\epsilon_{1}:=\epsilon$ and $\epsilon_{i+1}:=f(\epsilon,\eps_{i})$
	for each~$i$. Consider the $k$ ranges of widths and heights of
	type $(\eps_{i+1}N,\eps_{i}N]$. By an averaging argument there exists
	one index $j$ such that the total profit of items in $\opt$ with
	at least one side length in the range $(\eps_{j+1}N,\eps_{j}N]$ is
	at most $2\cdot\frac{\eps}{2}p(\opt)$. It is then sufficient to set
	$\epsl=\eps_{j}$ and $\epss=\eps_{j+1}$.
\end{proof}

\subsection{Proof of Lemma~\ref{lem:corridor-partition}}
\begin{proof} This result can be seen as a corollary of Lemma~\ref{lem:corridorPack-weighted}. In fact, since we assume there are no small items, we only need to care about the items that are intersected by the lines of the corridor partition which are at most $K\in O_{\eps}(1)$. If $|OPT|\le K/\eps$ then we can simply set all the items as untouchable and prove the claim (each item will be associated to a corridor which is in turn a box matching its size and nothing is discarded), and otherwise we can safely discard these $K$ items as $K\le \eps OPT$. \end{proof}

\subsection{Proof of Lemma~\ref{lem:partition-knapsack}}
\begin{proof}
	Consider the solution $\overline{\opt}$ together with the corridor partition from Lemma~\ref{lem:corridor-partition},
	and its further nice partition into subcorridors from Lemma~\ref{lem:corridor-pieces-1}. Recall that each subcorridor contains either only items from $\R'\cap\Rho$ or
	only items from from $\R'\cap\Rve$.
	
	For each corridor $C$ we will group its subcorridors (let us say they are numbered $P_1, P_2, \dots, P_{s(C)}$ in such a way that, for each $i=1,\dots,s(C)-1$, $P_{i}$ and $P_{i+1}$ are consecutive and also $P_{s(C)}$ and $P_1$ are consecutive) into sets (not necessarily disjoint) in such a way that if we remove the items contained in
	any of the sets of subcorridors, then we can subdivide the corridors into L- and U-corridors
	only. In each case we will prove that one of the removed set of items
	will have cardinality at most $\frac{3}{8}|\overline{\opt}(C)|$, where $\overline{\opt}(C)$ denotes the set of items contained in corridor $C$, which would conclude the proof.

	If $C$ is a path-corridor, then we will group the subcorridors $P_j$, where $j\equiv\alpha\bmod3$,
	into one set for each $\alpha\in\{0,1,2\}$ (notice that the groups
	are disjoint so one of them must have at most $\frac{1}{3}|\overline{\opt}(C)|$ items associated to it).
	If we remove the items from a group, we can subdivide the corridors
	into L-corridors by simply extending the boundaries of the subcorridors which are next to the deleted ones; this induces a constant number of L-corridors (and possibly rectangular regions in the boundaries of $C$).
	
	If instead $C$ is a cycle-corridor (notice that a cycle-corridor
	always has an even number of pieces), then we distinguish the following cases: 
	\begin{itemize}
		\item If $C$ has $4$ subcorridors then we can assign each subcorridor to a different set, in which case the deletion of one of the sets and associated items would induce a U-corridor if we extend the boundaries of the subcorridors which are next to the deleted one as before; since the sets are disjoint, one of them will have associated at most $\frac{1}{4}|\overline{\opt}(C)|$ items.
		\item If $C$ has $6$ or $12$ subcorridors then it can be treated exactly as in the case of path-corridors and obtain only L-corridors since these numbers are divisible by $3$. This way, one of the sets will have associated at most $\frac{1}{4}|\overline{\opt}(C)|$ items as required.
		\item If $C$ has $8$ or $14$ subcorridors, we will group them into sets as follows: let us choose a subcorridor uniformly at random and relabel them starting from there. We will then group the subcorridors $P_j$, where $j\equiv\alpha\bmod3$,
		into one set for each $\alpha\in\{0,1,2\}$ but we will also assign the subcorridor $P_{1}$ into the set containing $P_{3}$. This way we obtain only L-corridors when deleting any of the sets and the expected number of items in the deleted set of subcorridors if $s(C)=8$ is at most $\left(\frac{7}{8}\cdot \frac{1}{3} + \frac{1}{8}\cdot \frac{2}{3}\right)|\overline{\opt}(C)| = \frac{3}{8}|\overline{\opt}(C)|$ since every subcorridor is removed in one out of the three sets except for the one chosen randomly which is removed twice (if $s(C)=14$ then this expected value is at most $\left(\frac{13}{14}\cdot \frac{1}{3} + \frac{1}{14}\cdot \frac{2}{3}\right)|\overline{\opt}(C)| = \frac{5}{14}|\overline{\opt}(C)|$ which is better). This procedure can be derandomized by simply choosing the subcorridor with less items inside as the new $P_1$.
		\item If $C$ has $10$ subcorridors, we will use the fact that there must be four consecutive subcorridors such that any three consecutive subcorridors out of them induce a U-corridor. To see this, notice that the leftmost vertical, rightmost vertical, topmost horizontal and bottommost horizontal subcorridors must be the middle subcorridors of U-corridors, and since $s(C)<12$ these U-corridors cannot be disjoint. We will then group the subcorridors as follows: We will first relabel the subcorridors so that $P_1$ becomes the third out of the four aforementioned subcorridors. Then we will group the subcorridors $P_j$, where $j\equiv\alpha\bmod3$,
		into one set for each $\alpha\in\{0,1,2\}$. This leads to disjoint sets so one of them will have at most $\frac{1}{3}|\overline{\opt}(C)|$ items associated to it. Furthermore, it is not difficult to see that the deletion of any of the sets induces only L-corridors except possibly inside the sequence of subcorridors $P_{9}, P_{10}, P_1$ and $P_2$, which belong to sets $0, 1, 1$ and $2$ respectively. However, due to our choice of the subcorridor $P_1$, both the corridors induced by $P_9-P_{10}-P_1$ and $P_{10}-P_1-P_2$ are U-corridors, obtaining then only L- and U-corridors when deleting any of the sets.
		\item Finally if $C$ has at least $16$ subcorridors, we can delete one subcorridor chosen uniformly at random and then processing the induced path-corridor as described in the beginning. This way we obtain only L-corridors and the expected number of items of the deleted set is at most $\left(\frac{1}{16} + \frac{15}{16}\cdot \frac{1}{3}\right)|\overline{\opt}(C)| = \frac{3}{8}|\overline{\opt}(C)|$. This can again be derandomized by removing first the subcorridor with less items inside instead.
\end{itemize} \end{proof}

\subsection{Proof of Lemma~\ref{lem:partition-subcorridors}}
We will first consider the simpler case where the subcorridor is just a rectangular region and then use this idea to address the more general case.

\begin{lem}\label{lem:box-into-boxes} Given a box $S$ of height $h(S)$ and width $w(S)$, we can guess in time $N^{O_{\eps}(1)}$ a set of $O_{\eps}(1)$ non-overlapping boxes $\mathcal{B}(S)$ inside $S$ such that we can nicely place slices from $\tilde{I}(S)$ with total profit of $(1-\eps)\profit(\tilde{I}(S))$ inside the boxes $\mathcal{B}(S)$. \end{lem}

\begin{proof} Assume w.l.o.g. that $S$ is horizontal and that the slices inside correspond to horizontal items. We will partition $S$ into $\frac{1}{\eps}$ horizontal strips of height $\eps h(S)$. Since every slice is completely contained in one of the strips, there must be a strip such that the total profit of slices inside is at most $\eps \profit(\tilde{I}(S))$, and we remove this strip and the corresponding slices completely. Since we now have an empty strip of height $\eps h(S)$ in the region, this means that there exists a packing of the slices with total profit at least $(1-\eps)\profit(\tilde{I}(S))$ into a box of height $(1-\eps)h(S)$ and width $w(S)$. Consequently we can use resource augmentation (Lemma~\ref{lem:augmentPack}) to obtain the claimed packing of the slices. The obtained boxes can be guessed in time $N^{O_{\eps}(1)}$. \end{proof}

Now we will proceed with the proof for subcorridors.

\begin{proof}[Proof of Lemma~\ref{lem:partition-subcorridors}] Let $S$ be a subcorridor and assume w.l.o.g.~that $S$ is horizontal,
	i.e., it is defined via two horizontal edges $e_{1}=\overline{p_{1}p'_{1}}$
	and $e_{2}=\overline{p_{2}p'_{2}}$, and additionally two monotone
	axis-parallel curves connecting $p_{1}=(x_{1},y_{1})$ with $p_{2}=(x_{2},y_{2})$
	and connecting $p'_{1}=(x'_{1},y'_{1})$ with $p'_{2}=(x'_{2},y'_{2})$,
	respectively. Assume w.l.o.g. that $x_{1}\le x_{2}\le x'_{2}\le x'_{1}$
	and that $y_{1}<y_{2}$ (see Figure~\ref{fig:construction} for a depiction). Let
	$h(S)$ denote the the height of $S$ which we define as the distance
	between $e_{1}$ and $e_{2}$. Intuitively, we place $1/\epsilon^{2}$
	boxes inside $S$ of height $\epsilon^{2}h(S)$ each, stacked one
	on top of the other, and of maximum width such that they are contained
	inside $S$. Formally, we define $1/\epsilon^{2}$ boxes $B_{0},...,B_{1/\epsilon^{2}-1}$
	such that for each each $j\in\{0,...,1/\epsilon^{2}-1\}$ the bottom
	edge of box $B_{j}$ has the $y$-coordinate $y_{1}+j\cdot\epsilon^{2}h(S)$
	and the top edge of $B_{j}$ has the $y$-coordinate $y_{1}+(j+1)\cdot\epsilon^{2}h(S)$. For each such $j$ we define
	the $x$-coordinate of the left edge of $B_{j}$ maximally small and
	the $x$-coordinate of the right edge of $B_{j}$ maximally large
	such that $B_{j}\subseteq S$.
	
	Now consider the first $1/\epsilon$ boxes, i.e., $B_{0},...,B_{1/\epsilon-1}$.
	For each box $B_{j}\in\left\{ B_{0},...,B_{1/\epsilon-1}\right\} $
	consider the horizontal stripe $S_{j}:=[y_{1}+j\cdot\epsilon^{2}h(S),y_{1}+(j+1)\cdot\epsilon^{2}h(S)]\times[0,N]$
	(i.e., the horizontal stripe of height $\epsilon^{2}h(S)$ that contains
	$B_{j}$). By the pidgeon hole principle, one of the stripes $j^*$ must contain slices of total profit at most $\eps\profit(\tilde{I}(S))$. We delete all such slices. Next, we move down all the slices that remain and that intersect the boxes $B_{j^{*}+1},...,B_{1/\epsilon^{2}-1}$
	by $\epsilon^{2}h(P)$ units. Note that then they fit into the area
	defined by the union of the boxes $B_{j^{*}},...,B_{1/\epsilon^{2}-2}$.
	
	We define $w(S):=x'_{1}-x_{1}$, i.e., the length of $e_{1}$ (which
	is longer than $e_{2}$). Let $w'(S):=x'_{2}-x_{2}$, i.e., the length
	of $e_{2}$. Next, we would like to ensure that below the box $B_{j^{*}}$ there is no item $i$
	with $\width(i)<w'(S)$ (we want to achieve this since then we can
	stack the items underneath $B_{j^{*}}$ on top of each other).
	
	Therefore, consider the topmost $1/\epsilon^{2}-1/\epsilon$ boxes.
	We group them into $1/12\epsilon-1$ groups with $12/\epsilon$ boxes
	each, i.e., for each $k\in\{0,...,1/12\epsilon-2\}$ we define a group
	$\B_{k}:=\{B_{j}|j\in\{1/\epsilon+12k/\epsilon,...,1/\epsilon+12(k+1)/\epsilon-1\}\}$.
	Note that each group $\B_{k}$ contains exactly $12/\epsilon$ boxes
	and below $B_{j^{*}}$ there are at most $1/\epsilon$ boxes. By the
	pidgeon hole principle, there is a value $k^{*}\in\{0,...,1/\epsilon-2\}$
	such that the boxes in the group $\B_{k^{*}}$ intersect with
	slices of total profit of $O(\epsilon)\profit(\tilde{I}(S))$. We proceed then to delete all the slices that intersect a box in $\B_{k^{*}}$.
	
	Consider all slices that intersect one of the stripes in $\left\{ S_{0},...,S_{j^{*}-1}\right\} $ and
	that satisfy that $\width(i)\le w'(P)$. Due to Steinberg's algorithm~\cite{S97}
	they fit into a box of height $3\epsilon\cdot h(S)$ and width $w'(S)$.
	Therefore, they fit into $3/\epsilon$
	boxes in $\B_{k^{*}}$. We assign them to these boxes in $\B_{k^{*}}$.
	
	Now we define $S'$ as the sub-subcorridor induced by $e_{1}$, the bottom
	edge of $B_{j^{*}}$, and the respective part of the two monotone
	axis-parallel curves connecting $p_{1}=(x_{1},y_{1})$ with $p_{2}=(x_{2},y_{2})$
	and connecting $p'_{1}=(x'_{1},y'_{1})$ with $p'_{2}=(x'_{2},y'_{2})$,
	respectively. Each remaining slice intersecting $S'$
	satisfies that $\width(i)\ge w'(P)$. Therefore, we can stack these
	items on top of each other (using that $w'(P)>w(P)/2$).
	
	We obtain that each remaining slice is assigned to
	a box in $\left\{ B_{j^{*}},...,B_{1/\epsilon^{2}-1}\right\} $ or
	lies in $S'$. We finally apply Lemma~\ref{lem:box-into-boxes}
	to each box $B\in\left\{ B_{j^{*}},...,B_{1/\epsilon^{2}-1}\right\} $
	in order to partition $B$ further and such that the slices assigned
	to $B$ are nicely packed inside $B$.\end{proof}
\subsection{Proof of Lemma~\ref{lem:slices-structured}}
\begin{proof}
	Without loss of generality,  our subcorridors and boxes are horizontal. If $F$ is a box, it follows directly from the fact that we can push all slices to the left side of the box, and then sort them by width, resolving ties considering the input items to which each slice belongs to. If $F$ is a sub-subcorridor instead, we need to consider the case that two slices from the same item are at different positions of $F$ with respect to its defining curves $C_1$ and $C_2$. For this, consider $i^*$ to be the slice positioned the highest in $F$, then we move all other slices along the $x$-axis until they are all aligned at position $lc(i^*)$. Given that all slices in $F$ are nicely packed and $C_1$ and $C_2$ are monotonic, we can do this without intersecting other slices. Finally, we sort them from top to bottom non-increasingly by width.
\end{proof}

\subsection{Proof of Lemma~\ref{lem:converting-slices-items}}
\begin{proof}
Consider a value $\ell\in\{0,\dots,\lfloor\log_{1+\eps}N\rfloor\}$
and the set $\R_{hor}^{(\ell)}$. Denote by $\bar{\R}_{hor}^{(\ell)}$
its corresponding slices and by $\mathcal{B}_{hor}^{(\ell)}$ the
corresponding boxes. We interpret the assignment of the slices in
$\bar{\R}_{hor}^{(\ell)}$ due to Lemma~\ref{lem:place-slices} as
a fractional assignment of the items in $\R_{hor}^{(\ell)}$ to the
boxes in $\R_{hor}^{(\ell)}$. We model this as the solution to a
linear program where for each $i\in\R_{hor}^{(\ell)}$ and each $j\in\mathcal{B}_{hor}^{(\ell)}$
we introduce a variable $x_{i,j}$ that denotes the fractional extent
by which the item $i$ assigned to the box $j$. The constraints of
the linear program model that each item is assigned \emph{at most}
once (note that maybe not all slices of an item $i$ are assigned
to some box) and for each box $j\in\mathcal{B}_{hor}^{(\ell)}$ the
resulting height achieved by stacking the packed (fractional) items
must not exceed the height of the box. For each box $j\in\mathcal{B}_{hor}^{(\ell)}$
we denote by $h(j)$ its height. 
\begin{align*}
 & \max\sum_{i\in I,j\in\mathcal{B}}x_{i,j}\\
\text{s.t. } & \sum_{j}x_{i,j}\leq1,\quad\text{for any }i\in I\\
 & \sum_{i\in I}h(i)x_{i,j}\leq h(j),\quad\text{for any }j\in\mathcal{B}\\
 & x_{i,j}\geq0
\end{align*}
Let $x^{*}$ denote the optimal fractional solution and assume w.l.o.g.~that
$x^{*}$ is an extreme point solution. Since our constructed assignment
of the slices yields a feasible solution, we know that the profit
of $x^{*}$ is at least $(1-O(\epsilon))\left|\Rho^{(\ell)}\cap\opt'\right|$.
By the rank lemma, there are at most $|\R_{hor}^{(\ell)}|+|\mathcal{B}_{hor}^{(\ell)}|$
variables in the support of $x^{*}$. Also, using the rank lemma we
can assume that there are at most $|\mathcal{B}_{hor}^{(\ell)}|$
items which are assigned to more than one box. We drop all slices
of these items, i.e., the slices of at most $|\mathcal{B}_{hor}^{(\ell)}|$
items. Since $x^{*}$ is an extreme point solution, for each $j\in\mathcal{B}_{hor}^{(\ell)}$
there can be at most one item $i\in\R_{hor}^{(\ell)}$ with $0<x_{i,j}^{*}<1$.
We drop also these fractionally assigned items, at most $|\mathcal{B}_{hor}^{(\ell)}|$
many.%

Hence, we dropped all items that were not integrally assigned, which
are at most $2|\mathcal{B}_{hor}^{(\ell)}|$ items in total. Therefore,
we obtain an integral packing of $(1-O(\epsilon))\left|\Rho^{(\ell)}\cap\opt'\right|-2|\B_{hor}^{(\ell)}|$
items in $\R_{hor}^{(\ell)}$. A similar argumentation works for sets
of vertical items $\R_{ver}^{(\ell)}$.
\end{proof}

\subsection{Proof of Lemma~\ref{lem:opt-large}}
\begin{proof}
	Recall that $OPT'$ is the solution obtained thanks to Lemma~\ref{lem:partition-knapsack}. By using Lemma~\ref{lem:place-slices} we can turn this solution into a packing of slices whose corresponding profit is at least $(1-O(\eps))OPT$ and that obey a restricted number of classes $\tilde{\R}_{hor}^{(\ell)}, \tilde{\R}_{ver}^{(\ell)}$. This solution can be rearranged and decomposed into a set of $O_\eps(\log N)$ boxes as stated in Lemmas~\ref{lem:partition-subcorridors} and \ref{lem:slices-structured}. Using the existence of this structured solution, we can actually compute a feasible solution in time $(nN)^{O_{\eps}(1)}$ consisting of at least $(1-O(\eps))|OPT'| - |\mathcal{B}|$ items by means of Lemmas~\ref{lem:guess-slices} and \ref{lem:converting-slices-items}, where $|\mathcal{B}|$ is the number of obtained boxes. Since $|OPT'|> c_{\eps} \log N$ and $|\mathcal{B}| \le O_{\eps}(\log N)$, we can choose $c_{\eps}$ to be sufficiently large so that $|\mathcal{B}|\le \eps|OPT'|$, concluding the proof. 
\end{proof}
\subsection{Proof of Lemma~\ref{lem:color-coding-LU}}
\begin{proof}
	Let $k = |OPT'|$. As $k$ is bounded by $c\cdot \log N$, we can guess its value in $c\cdot \log N$ time. We color each item in $I$ uniformly at random using $k$ colors. Then, the probability that we get a successful coloring is $k!/ k^{k} \geq 1/e^k$ which is again at least $1/N^{O(c)}$. Now, given a successful coloring, we can guess which item goes to which corridor $C\in \mathcal{C}$ in time $O_\eps(1)^{k} = N^{O_\eps(c)}$. Finally, we can derandomize this algorithm using an $(n, k)$-perfect hash family as described in~\cite{NSS95} which can be constructed in time $(2e)^k k^{{O}(\log k)} n^{O(1)} = N^{O_\eps(c)}n^{O(1)}$.
\end{proof}
\subsection{Proof of Lemma~\ref{lem:DP-coloring}}

\begin{proof} We will use the previously described Dynamic Program to determine whether there exists the set $I_C'\subseteq I_C$ with $\gamma$ different colors that can be packed inside $C$. Recall that $I_C$ contains all the items from $OPT'$ that are placed inside $C$ and they are colored with $\gamma$ different colors in such a way that the items from $OPT' \cap I_C$ have different colors as we assume a successful coloring. Formally a DP-cell is defined by two long chords $\ell_1, \ell_2$ that do not cross, a set of items $\tilde{I_C}\subseteq I_C$ of cardinality at most $O(k/\epsl)$ plus a placement of them inside $C$ so that they intersect at least one of the chords, and a set of color $\Gamma$. We will first prove the following claim: For any DP-cell $(\ell_1, \ell_2, \tilde{I_C}, \Gamma)$ and $C'$ the polygon surrounded by $\ell_1, \ell_2, e_0$ and $e_{k+1}$, there exists a long chord $\ell$ different from $\ell_1$ and $\ell_2$ that lies within or $|OPT(\ell_1, \ell_2, \tilde{I_C}, \Gamma)| \le O(k/\epsl)$, and furthermore this chord is consistent if both $\ell_1$ and $\ell_2$ are consistent. This implies that if the DP cannot recourse it is because it is at a base case.
	
	To see the proof of the claim, recall that each chord $\ell_i$ is defined by a sequence of axis-parallel lines $\ell_i^{(1)}, \dots, \ell_i^{(k)}$, where $i\in\{1,2\}$ and $k$ is the number of subcorridors of $C$. Now consider two cases: If there exists an index $j\in\{1,\dots,k\}$ such that $\ell_1^{(j)}$ and $\ell_2^{(j)}$ have distance at least $2$ (meaning that there exists some segment of length at least $2$ which is perpendicular to both lines and does not properly cross them), then we can define a new chord by ``following'' $\ell_1$ and $\ell_2$. More in detail, suppose that $\ell_1^{(j)}$ and $\ell_2^{(j)}$ are horizontal (the remaining case being symmetric), which implies that there exists an horizontal line $\ell^{(j)}$ that is completely disjoint from both the previous lines and it lies in between them. We will then extend this line to the left until intersecting $\ell_1$ or $\ell_2$ (notice that we cannot intersect both), and turn following the direction of the corridor iterating the same procedure; if at some point the line does not intersect any of the two chords, it is because we reached $e_0$. We do the analogous procedure extending $\ell^{(j)}$ to the right and turning according to the corridor once we intersect any of the chords. 

	This new sequence of lines $\ell$ is indeed a long chord, different from $\ell_1$ and $\ell_2$ thanks to the segment in subcorridor $j$, and it is contained in $C'$. Also if $\ell_1$ and $\ell_2$ are consistent, $\ell$ can be made consistent by ``turning'' every time it intersects the boundary of a subcorridor rather than when it intersects $\ell_1$ or $\ell_2$.
	
	On the other hand, if $\ell_1$ and $\ell_2$ are at distance at most $1$ everywhere in the corridor, then since items have integral heights and widths it is possible to place at most $O(k/\epsl)$ items in $C'$ as they will have to be one next to the other along their long dimension, or equivalently $|OPT(\ell_1, \ell_2, \tilde{I_C}, \Gamma)| \le O(k/\epsl)$. This can indeed be computed optimally by brute force.
	
	The final solution that we output is encoded in the DP-cell $(\ell_L, \ell_R, \emptyset, \{1,\dots,\gamma\})$, and it is actually optimal as we can completely recover the optimal solution by decomposing the corridor via a sequence of consistent long chords that intersect the items. If there is no such placement the DP will simply return ``fail''. The running time is bounded the number of cells which is at most $(nN)^{O(k/\epsl)}$ and the number of guesses it does when computing a cell which is at most $(nN)^{O(k/\epsl)}\cdot 2^{O(\gamma)}$.  \end{proof}

\subsection{Proof of Lemma~\ref{lem:DP-weighted}}

Notice first that, since there are at most $O_{\eps}(1)$ boxes $E$ in $\bar{\B}$
for the unweighted case ( at most $O_{\eps}(1)$ elements $E$ in
$\bar{\B}\cup\bar{I}'$ for the weighted case) such that $\width(E)<\epsl\cdot N$
and $\height(E)<\epsl\cdot N$, denote by $\bar{E}_{small}\subseteq\bar{\B}$
($\bar{E}_{small}\subseteq\bar{\B}\cup\bar{I}'$ for weighted, with
$\bar{\B}_{small}=\bar{E}_{small}\cap\bar{\B}$) such elements. We
can guess a way to pack all these elements inside the corridor in
time $N^{O_{\eps}(1)}$. Let now $\B:=\bar{\B}\setminus\bar{\B}_{small}$
be the set of all other boxes, which we can consider among the set
of items to be packed, and even enforce that they are included in
the final solution by assigning them a large enough profit.

Now we use color coding with parameter $k$, meaning that we will
randomly color all the items in $\bar{\R}$ using $k$ colors, and
it is possible to show that with probability at least $1/2^{k}$ all
the items from an optimal solution with at most $k$ items receive
different colors, e.g.,~\cite{cygan2015parameterized}. If we repeat
this procedure $2^{O(k)}$ times we can ensure that with high probability
one of the runs delivers such a coloring, and this can be derandomized
efficiently using standard techniques~\cite{NSS95}. {Also,
we color each box with a different color}.

We devise a dynamic programming algorithm to pack items of different
colors inside the corridor while maximizing the profit. This algorithm
is inspired on GEO-DP, a dynamic programming algorithm originally
developed for the Maximum Independent Set of Rectangles problem by
Adamaszek and Wiese~\cite{AW13}.

Let us fix a parameter $t\in\mathbb{N}$ and an embedding of $C$
into the plane such that all the coordinates of the vertices of the
corridor are integral. Let $\mathcal{P}$ denote the set of all polygons
inside the corridor having integral coordinates and at most $t$ axis-parallel
edges {which are not overlapping with the already placed boxes
$\bar{\B}_{small}$}. These polygons may not be simple and have holes,
at which case the bound on the number of edges counts both the outer
edges and the boundaries of the holes. We will introduce a DP-cell
for each polygon $P\in\mathcal{P}$ and each possible set of colors
$K\subseteq\{1,\dots,k+|\mathcal{B}|\}$, where such a cell will store
the optimal solution for the problem of packing at most one item of
each color {in $K$}inside the polygon while maximizing the
total profit. Following from the result in~\cite{AW13}, the number
of cells is at most $N^{O(t)}\cdot2^{k+|\B|}$. {We choose $t={5(\frac{1}{\eps}\cdot\frac{1}{\epsl}+|\bar{\B}_{small}|)+\frac{3}{\eps}}=O_{\epsilon}(1)$.}

To compute the solution for a given cell, consisting of a polygon
$P\in\mathcal{P}$ and a set of colors $K\subseteq\{1,\dots,k+|\B|\}$,
we proceed as follows: If $K=\emptyset$ we simply return an empty
solution and terminate; if $|K|=1$, we return the item colored as
the element of $K$ of maximum profit that can be packed inside $P$
(checking if an item can be packed inside a polygon in $\mathcal{P}$
can be done efficiently as its coordinates can be restricted to combinations
of the coordinates of the vertices that define the polygon) returning
$\emptyset$ if no such packing is possible; otherwise, we enumerate
all the possible ways to partition $P$ into $t'$ polygons in $\mathcal{P}$
and $K$ into $t'$ subsets of $K$ {for each $t'\le t$},
and return the solution with the maximum value among all these
DP-cells. For each partition we look up the DP table value for the
polygons with their corresponding sets of colors, and return the partition
of maximum total profit (where the total profit is the sum of the
profits of the solutions returned for each polygon in the partition).
At the end, the algorithm outputs the value in the DP-cell corresponding
to the polygon defined by the corridor minus the boxes from $\bar{\B}_{small}$
and the set of colors $\{1,\dots,k+|\B|\}$ such that its profit,
among all possible packings, is maximized. It is not difficult to
see that the running time of this algorithm can be bounded above by
$N^{O(t^{2})}2^{O(k+|\B|)}$.

Now we will argue about the correctness of the algorithm, meaning
that we will specify a sequence of polygons such that the optimal
solution can be feasibly packed inside them, and this packing can
be found by our dynamic program. Since the running time of the algorithm
in this case would be bounded by $N^{O_{\eps}(1)}2^{O(k+|\B|)}$,
this would conclude the proof of the lemma.

Consider initially the corridor $C$ and an optimal solution for the
problem with the boxes from $\bar{\B}_{small}$ inside as we guessed.
Let us first define a $C$-curve, which consists of an axis-parallel
curve consisting of $s(C)$ edges that is completely contained in
$C$ and intersects both $e_{0}$ and $e_{s(C)+1}$. We say that a
$C$-curve is \emph{feasible} if whenever it intersects an item or
a box from $\B$, it completely crosses it along its longer dimension
(which by assumption has at least length $\epsl\cdot N$). Notice
that any feasible $C$-curve can cross at most $\frac{1}{\eps}\cdot\frac{1}{\epsl}$
items or boxes from $\B$.

{ Let us assume w.l.o.g. that $C$ contains at least two items (or
boxes from $\B$) such that one of them, which we call $i$, is horizontal,
i.e., has width at least $\epsl\cdot N$. Consider any feasible $C$-curve
$L$ that crosses $i$. Then we can partition the corridor into a
collection of disjoint polygons as follows: we consider a rectangular
region for each box or item crossed by $L$ such that it matches both
its size and position. Then we include {all the} regions enclosed
by the border of the previous containers, along with the curve $L$
and the edges of the corridor. Notice that this gives us at most $3\left(\frac{1}{\eps}\cdot\frac{1}{\epsl}+|\bar{\B}_{small}|\right)\le t$
polygons such that every item is completely contained in one of them.
Furthermore, we can easily check that there are at most $t$ edges
involved in the construction of these polygons. From here we will
recourse the procedure on the obtained polygons until the solution
is completely partitioned.}

We will make sure that the constructed polygons always have the following
structure: there exist two feasible $C$-curves that do not cross,
such that the boundary of the polygon is contained in the union of
the curves, the boxes for the items that they cross and the boundary
of the corridor. Provided with this it is not difficult to see that
all the obtained polygons belong to $\mathcal{P}$. Notice that this
property is satisfied for the polygons of the first level of the recursion
as the boundaries of the corridor define also feasible $C$-curves.

Now suppose now we are at a deeper level of the recursion, meaning
that we have a polygon defined by two feasible $C$-curves containing
at least two items (or boxes from $\B$) inside. We perform a similar
procedure as before, identifying an item $i'$ and taking a feasible
$C$-curve that crosses the item. However this time we will take the
following specific $C$-curve: We draw a line crossing item $i'$
along its long dimension until it intersects the feasible $C$-curves
that bound the polygon; if this is possible while only crossing completely
items along its long dimension, then we continue along those feasible
$C$-curves completing a different feasible $C$-curve; {if on the
other hand some item blocks the candidate line, we just bend and continue
following the shape of the corridor until intersecting the feasible
$C$-curve or intersecting $e_{0}$ or $e_{s(C)+1}$}. We again define
rectangular regions for the items crossed by our new feasible $C$-curve,
and it is easy to see that the constructed polygons satisfy the required
property, concluding the proof.

\subsection{Proof of Lemma~\ref{lem:pack-into-boxes-weighted}}
In order to prove Lemma~\ref{lem:pack-into-boxes-weighted}  we will make use of
the following well-known result for the Generalized Assignment problem
from Shmoys and Tardos~\cite{ST93}. In the Generalized Assignment
problem we are given a set of $m$ containers where each container
$i$ has a given capacity $C_{i}$, and a set of $n$ tasks where
each task $j$ has, for every possible container $i$, a profit $p_{ij}$
and size $s_{ij}$ if it is assigned to $i$. The goal is to find
a subset of the tasks and a way to assign them to the containers so
that the total size of the tasks in each container is at most its
capacity and the total profit of the assignment is maximized. The
following theorem states that it is possible to compute an assignment
of the tasks into the containers having at least the optimal profit
but where the capacity constraints are not satisfied.
\begin{thm}[Shmoys and Tardos~\cite{ST93}]
\label{thm:GAP-shmoystardos} Given an instance of GAP, it is possible
to compute an assignment of the items into the containers in time
$(nm)^{O(1)}$ such that: 
\begin{itemize}
\item The total profit of the items assigned to the containers is at least
the profit of the optimal solution for the instance, and 
\item For each container $i$, the total size of the items assigned to $i$
is at most $C_{i}+\underset{j\text{ assigned to }i}{\max}{s_{ij}}$. 
\end{itemize}
\end{thm}

\begin{proof}[Proof of Lemma~\ref{lem:pack-into-boxes-weighted}]
Consider the set of boxes $\bar{\B}$ and the way $\bar{\R}'$ are
nicely packed into them. We will create an instance of GAP so as to
use Theorem~\ref{thm:GAP-shmoystardos} to decide which items are
assigned to the boxes and how, and then remove some rectangles to
turn the assignment into a feasible nicely packed solution. For each
box $B\in\bar{\B}$ we create a container whose capacity will be $h(B)$
if $B$ contains items one on top of the other, $w(B)$ if $B$ contains
items one next to the other or $a(B)$ if $B$ contains only items
relatively small compared to it (this with respect to the packing
of $\bar{\R}'$ into $\bar{\B}$). We will also create for each item
$j\in\tilde{I}$ a task of profit $p_{ij}=\profit(j)$ for every container
$i$ and of size $s_{ij}$ equal to 
\begin{itemize}
\item $h(j)$ if $i$ corresponds to a box for items one on top of the other
and $j$ fits inside the box associated to $i$, 
\item $w(j)$ if $i$ corresponds to a box for items one next to the other
and $j$ fits inside the box associated to $i$, 
\item $a(j)$ if $i$ corresponds to a box for relatively small items and
$h(j)\le\eps h(B)$ and $w(j)\le\eps w(B)$, where $B$ is the box
associated to container $i$, or 
\item $N+1$ otherwise. 
\end{itemize}
In principle we do not know the way items from $\bar{\R}'$ are nicely
packed into the boxes, but we can guess a label for each box saying
if it will contain items one on top of the other, one next to the
other or only relatively small items in time $2^{O(|\B|)}$. Notice
that the optimal profit for this GAP instance is at least $\profit(\bar{\R}')$
as the way $\bar{\R}'$ is nicely packed into $\bar{\B}$ induces
a feasible solution. We now apply Theorem~\ref{thm:GAP-shmoystardos}
to obtain an assignment of a set of items $\R_{GAP}\subseteq\bar{\R}$
into the boxes of total profit at least $\profit(\bar{\R}')$ where
the containers are slightly overloaded. Although this does not induce
a feasible packing, due to the size guarantees of Theorem~\ref{thm:GAP-shmoystardos}
we can remove the task of largest size from each container (i.e. the
item of largest width from each box with items one next to the other,
the item of largest height from each box with items one on top of
the other and the item of largest area from each box for relatively
small items) and obtain a set of items $\R_{GAP}'\subseteq\bar{\R}$
of total profit at least $p(\bar{\R}')-|\B|\cdot\max_{i\in\bar{I}}\profit(i)$
such that items assigned to boxes that stack their items are nicely
packed into their boxes, and for each remaining box we have a set
of items of total area at most the area of the box and whose dimensions
are at most an $eps$ fraction of the dimensions of the box.

In order to turn this into a nicely packed solution we need to pack
(a subset of) the small items into their corresponding boxes. Consider
a box $B$ for relatively small items and the set of items that were
assigned to it from $\R_{GAP}'$, which we will denote by $\R_{GAP}'(B)$.
If $a(\R_{GAP}'(B))\le(1-2\eps)a(B)$ then we can pack all of them
into the box using NFDH (Theorem~\ref{thm:nfdhPack}). If it is not
the case, then we will partition the items by means of the following
procedure: let us greedily pick items (in any order) until their total
area becomes larger than $2\eps a(B)$ and call this a set, and restart
the procedure with the remaining items until we finish. Notice that
the total area of each set will be at least $2\eps a(B)$ and at most
$3\eps a(B)$ (except maybe for the last one). The number of obtained
sets with total area at least $2\eps a(B)$ is at least $\frac{1-5\eps}{3\eps}$
and hence one of such sets must have total profit at most $\frac{3\eps}{1-5\eps}\profit(\R_{GAP}')\le4\eps\profit(\R_{GAP}')$.
If we remove this set then the rest can be packed inside the box using
NFDH (Theorem~\ref{thm:nfdhPack}).

If we apply this procedure to all the boxes for relatively small items,
then overall we get a set of items $\tilde{\R}'\subseteq\bar{\R}$
of total profit at least $(1-4\eps)p(\bar{\R}')-|\B|\cdot\max_{i\in\bar{I}}\profit(i)$
and a way to nicely pack them into $\bar{\B}$ in time $(n|\bar{\B})^{O(1)}$.

\end{proof}

\subsection{Proof of Lemma~\ref{lem:partition-acute-pieces-weighted}}

Let $P$ be an acute piece and assume w.l.o.g.~that $P$ is horizontal,
i.e., $P$ is defined via two horizontal edges $e_{1}=\overline{p_{1}p'_{1}}$
and $e_{2}=\overline{p_{2}p'_{2}}$, and additionally two monotone
axis-parallel curves connecting $p_{1}=(x_{1},y_{1})$ with $p_{2}=(x_{2},y_{2})$
and connecting $p'_{1}=(x'_{1},y'_{1})$ with $p'_{2}=(x'_{2},y'_{2})$,
respectively. Assume w.l.o.g. that $x_{1}\le x_{2}\le x'_{2}\le x'_{1}$
and that $y_{1}<y_{2}$ (see Figure~\ref{fig:construction}). Let
$h(P)$ denote the the height of $P$ which we define as the distance
between $e_{1}$ and $e_{2}$. Intuitively, we place $1/\epsilon^{2}$
boxes inside $P$ of height $\epsilon^{2}h(P)$ each, stacked one
on top of the other, and of maximum width such that they are contained
inside $P$. Formally, we define $1/\epsilon^{2}$ boxes $B_{0},...,B_{1/\epsilon^{2}-1}$
such that for each each $j\in\{0,...,1/\epsilon^{2}-1\}$ the bottom
edge of box $B_{j}$ has the $y$-coordinate $y_{1}+j\cdot\epsilon^{2}h(P)$
and the top edge of $B_{j}$ has the $y$-coordinate $y_{1}+(j+1)\cdot\epsilon^{2}h(P)$
(see Figure~\ref{fig:construction}). For each such $j$ we define
the $x$-coordinate of the left edge of $B_{j}$ maximally small and
the $x$-coordinate of the right edge of $B_{j}$ maximally large
such that $B_{j}\subseteq P$.
%

For proving the unweighted case, we delete all skewed (i.e., horizontal) items in $\opt'(P)$
that intersect a horizontal edge of a box in $\left\{ B_{0},...,B_{1/\epsilon^{2}-1}\right\} $.
Note that there can be at most $1/\epsilon^{2}\cdot1/\epsl$ many.

Now consider the first $1/\epsilon$ boxes, i.e., $B_{0},...,B_{1/\epsilon-1}$.
For each box $B_{j}\in\left\{ B_{0},...,B_{1/\epsilon-1}\right\} $
consider the horizontal stripe $S_{j}:=[y_{1}+j\cdot\epsilon^{2}h(P),y_{1}+(j+1)\cdot\epsilon^{2}h(P)]\times[0,N]$
(i.e., the horizontal stripe of height $\epsilon^{2}h(P)$ that contains
$B_{j}$). 

Each item $i$ contained in $P$ satisfies that $\height(i)\le\eps^{4}\cdot\height(P)$
and therefore each item $i$ contained in $P$ intersects at most
2 stripes in $\left\{ S_{0},...,S_{1/\epsilon^{2}-1}\right\} $.

Therefore, by the pigeon hole principle, two consecutive boxes $B_{j^{*}},B_{j^{*}+1}\in\left\{ B_{0},...,B_{1/\epsilon-1}\right\} $
have the property that the stripe $S'_{j^{*}}:=[0,N]\times[y_{1}+j^{*}\cdot\epsilon^{2}h(P),y_{1}+(j^{*}+2)\cdot\epsilon^{2}h(P)]$
(containing $B_{j^{*}}$ and $B_{j^{*}+1}$) intersects at most $4\epsilon|\opt'(P)|$
of the remaining items in $\opt'(P)$. We delete all items in $\opt'(P)$
that are intersected by $S'_{j^{*}}$. Next, we move down all items
in $\opt'(P)$ that intersect the boxes $B_{j^{*}+2},...,B_{1/\epsilon^{2}-1}$
by $\epsilon^{2}h(P)$ units. Note that then they fit into the area
defined by the union of the boxes $B_{j^{*}+1},...,B_{1/\epsilon^{2}-2}$.
In the weighted case, we assign to $B_{j^{*}}$ all items that intersect
a horizontal edge of a box in $\left\{ B_{j^{*}+1},...,B_{1/\epsilon^{2}-1}\right\} $.
This can be done since each such item has a height of at most $\eps^{4}\cdot\height(P)$
which implies that $\eps^{4}\cdot\height(P)/\epsilon^{2}\le\epsilon^{2}\cdot\height(P)$.

We define $w(P):=x'_{1}-x_{1}$, i.e., the length of $e_{1}$ (which
is longer than $e_{2}$). Let $w'(P):=x'_{2}-x_{2}$, i.e., the length
of $e_{2}$. Due to the definition of corridors, we have that $w'(P)\ge(1-2\epsilon)w(P)$. Next, we would
like to ensure that below the box $B_{j^{*}}$ there is no item $i$
with $\width(i)<w'(P)$ (we want to achieve this since then we can
stack the items underneath $B_{j^{*}}$ on top of each other) and
no small item intersects the boundary of a box.

\begin{figure}
\begin{centering}
\includegraphics[scale=0.8]{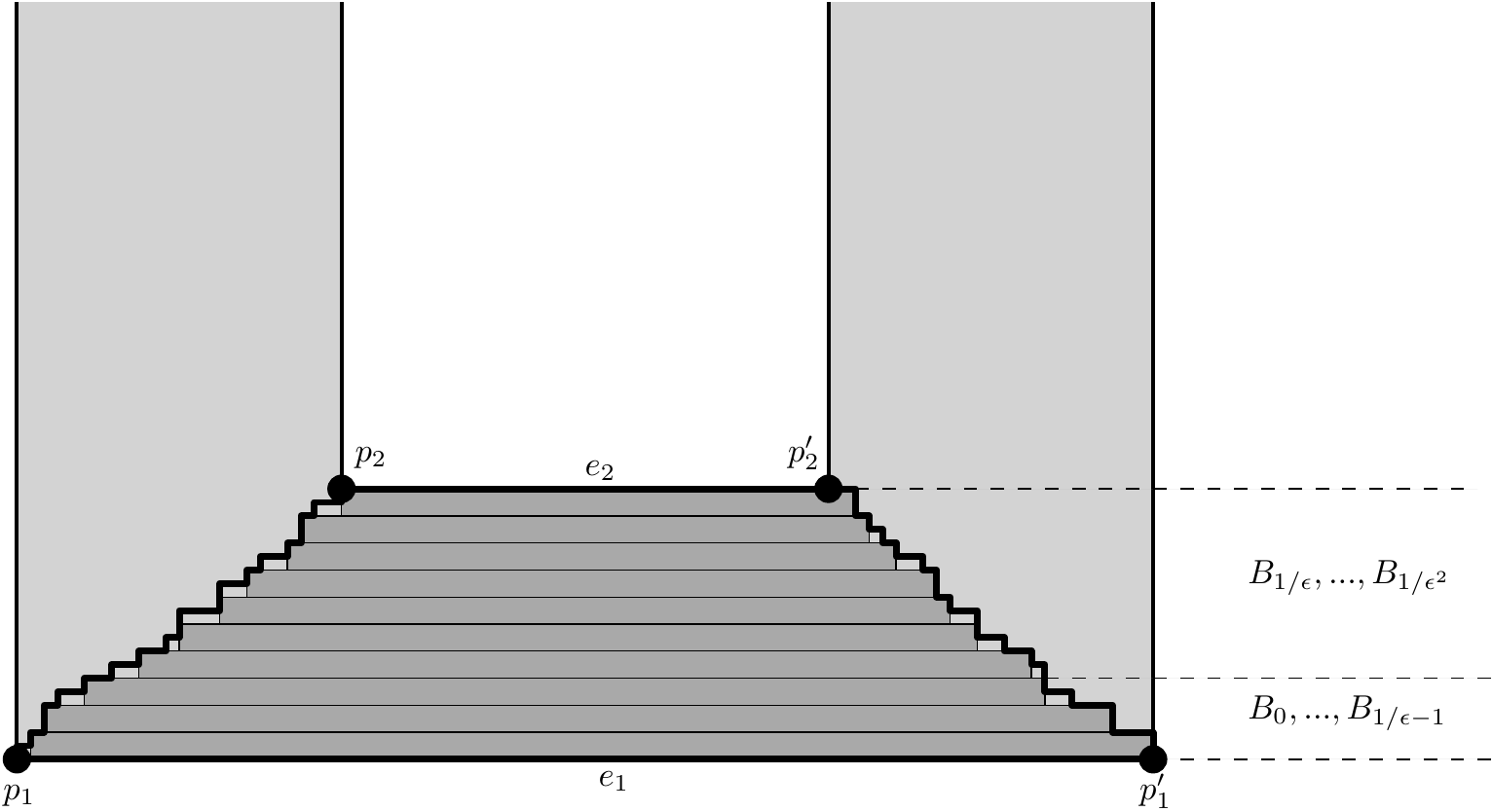} 
\par\end{centering}
\caption{\label{fig:construction}The construction used in the proof of Lemma~\ref{lem:partition-acute-pieces-weighted}.}
\end{figure}

Therefore, consider the topmost $1/\epsilon^{2}-1/\epsilon$ boxes.
We group them into $1/12\epsilon-1$ groups with $12/\epsilon$ boxes
each, i.e., for each $k\in\{0,...,1/12\epsilon-2\}$ we define a group
$\B_{k}:=\{B_{j}|j\in\{1/\epsilon+12k/\epsilon,...,1/\epsilon+12(k+1)/\epsilon-1\}\}$.
Note that each group $\B_{k}$ contains exactly $12/\epsilon$ boxes
and below $B_{j^{*}}$ there are at most $1/\epsilon$ boxes. By the
pigeon hole principle, there is a value $k^{*}\in\{0,...,1/\epsilon-2\}$
such that the the boxes in the group $\B_{k^{*}}$ intersect with
items of total weight of $O(\epsilon)w(\opt'(P))$ where $w(\opt'(P))$
denotes the weight of the remaining items in $\opt'(P)$. Therefore,
we delete all items that intersect a box in $\B_{k^{*}}$.

Consider all items $i\in\opt'(P)\cap\Rsk$ that intersect at least
one of the stripes in $\left\{ S_{0},...,S_{j^{*}-1}\right\} $ and
that satisfy that $\width(i)\le w'(P)$. Due to Steinbergs algorithm~\cite{S97}
they fit into a box of height $3\epsilon\cdot h(P)$ and width $w'(P)$.
Therefore, all but $O(1/\epsl\cdot1/\epsilon)$ of them fit into $3/\epsilon$
boxes in $\B_{k^{*}}$. We assign them to these boxes $3/\epsilon$
boxes in $\B_{k^{*}}$.

Now consider all items $i\in\opt'(P)\cap\Rsm$ that intersect the
boundary of a box in $\left\{ B_{0},...,B_{1/\epsilon^{2}-1}\right\} $
or that are contained in one of the stripes in $\left\{ S_{0},...,S_{j^{*}-1}\right\} $.
In the unweighted case, if there are such items we know that $h(P)>\epss N/\epsilon^{4}$
and hence their total area is at most $\epsilon\cdot h(P)\cdot w(P)+w(P)N\epss/\epsilon^{2}\le4\epsilon\cdot h(P)\cdot w'(P)$,
using that $w'(P)\ge(1-2\epsilon)w(P)$. Also, note that for each
item $i\in\opt'(P)\cap\Rsm$ we have that $\width(i)\le\epss N\le O(\epsilon^{2})w'(P)$
and $\height(i)\le\epss N\le O(\epsilon^{2})w'(P)$.  Therefore, we can partition
all these small items into 9 groups such that each group contains
items with a total area of $\epsilon\cdot h(P)w'(P)/2$ and each group
can be packed into a box in $\B_{k^{*}}$ using Steinbergs algorithm.
In the weighted case, the total area of items $i$ that intersect
the boundary of a box in $\left\{ B_{0},...,B_{1/\epsilon^{2}-1}\right\} $
or that are contained in one of the stripes in $\left\{ S_{0},...,S_{j^{*}-1}\right\} $
with $\height(i)\le\eps^{4}\cdot\height(P)$ and $\width(i)\le\epss\cdot\width(P)$
is bounded by $\epsilon\cdot h(P)\cdot w(P)+O(w(P)h(P)\epss/\epsilon^{2})\le O(\epsilon)\cdot h(P)\cdot w'(P)$.
Hence, we can pack them into the boxes $\B_{k^{*}}$ like before.

Now we define $P'$ as the acute piece induced by $e_{1}$, the bottom
edge of $B_{j^{*}}$, and the respective part of the two monotone
axis-parallel curves connecting $p_{1}=(x_{1},y_{1})$ with $p_{2}=(x_{2},y_{2})$
and connecting $p'_{1}=(x'_{1},y'_{1})$ with $p'_{2}=(x'_{2},y'_{2})$,
respectively. Each remaining item $i\in\opt'(P)$ intersecting $P'$
satisfies that $\width(i)\ge w'(P)$. Therefore, we can stack these
items on top of each other (using that $w'(P)>w(P)/2$).

We obtain that each remaining item from $\opt'(P)$ is assigned to
a box in $\left\{ B_{j^{*}},...,B_{1/\epsilon^{2}-1}\right\} $ or
lies in $P'$. We define $\opt'_{1}(P)$ to be the former set of items
and $\opt'_{2}(P)$ to be the latter set. Finally, we apply Lemma~\ref{lem:partition-boxes-weighted}
to each box $B\in\left\{ B_{j^{*}},...,B_{1/\epsilon^{2}-1}\right\} $
in order to partition $B$ further and such that the items assigned
to $B$ are nicely packed inside $B$.

\subsection{Proof of Lemma \ref{lem:slices-weighted}}

For each corridor $C\in\C$ let $\P(C)$ be a partition of $C$ into
$s(C)$ pieces such that each skewed item in $\opt'$ is contained
in a piece in $\P(C)$. Let $\P:=\bigcup_{C\in\C}\P(C)$.

Consider a set $\R_{hor}^{(\ell)}$. Observe that its items might
be contained in the acute piece $P'$ corresponding to some acute
piece $P$ according to Lemma~\ref{lem:partition-acute-pieces-weighted}
or in a box $B\in\tilde{\B}'(P)$ (according to Lemma~\ref{lem:partition-acute-pieces-weighted}).
Note that in the latter case, an item $i\in\R_{hor}^{(\ell)}$ might
be placed as a ``small item'', i.e., such that $\width(i)\le\epsilon\cdot\width(B)$
and $\height(i)\le\epsilon\cdot\height(B)$. To this end, we partition
$\R_{hor}^{(\ell)}$ into $2^{O_{\epsilon}(1)}$ subgroups $\left\{ \R_{hor,j}^{(\ell)}\right\} _{j}$
such that for each subgroup $\R_{hor,j}^{(\ell)}$ it holds that for
each acute piece $P$ and each box $B\in\tilde{\B}'(P)$ either each
item $i\in\R_{hor,j}^{(\ell)}$ satisfies that $\width(i)\le\epsilon\cdot\width(B)$
and $\height(i)\le\epsilon\cdot\height(B)$, or no item $i\in\R_{hor,j}^{(\ell)}$
satisfies that $\width(i)\le\epsilon\cdot\width(B)$ and $\height(i)\le\epsilon\cdot\height(B)$.
Similarly, we require that for each acute piece $P$ (according to
Lemma~\ref{lem:partition-acute-pieces-weighted}) either each item
$i\in\R_{hor,j}^{(\ell)}$ satisfies that $\width(i)>\epsl\cdot\width(P)$
and $\height(i)\le\epsilon^{4}\cdot\height(P)$, or no item $i\in\R_{hor,j}^{(\ell)}$
satisfies that $\width(i)>\epsl\cdot\width(P)$ and $\height(i)\le\epsilon^{4}\cdot\height(P)$.
Hence, intuitively all items in $\R_{hor,j}^{(\ell)}$ are small w.r.t.
exactly the same set of previously guessed boxes and wide w.r.t. exactly
the same set of acute pieces $P$. Let $L=O_{\epsilon}(\log nN)$
denote the total number of resulting groups (for horizontal and vertical
items).

Then, with each set $\R_{hor,j}^{(\ell)}$ we  compute
a set $\hat{\R}_{hor,j}^{(\ell)}\subseteq\R_{hor}^{(\ell)}$ of items
of minimum width among the items in $\R_{hor}^{(\ell)}$, such that
$\height(\hat{\R}_{hor,j}^{(\ell)})\le\height(\overline{\opt}'\cap\R_{hor,j}^{(\ell)})$
but also $\profit(\hat{\R}_{hor,j}^{(\ell)})\ge(1-O(\epsilon))\profit(\overline{\opt}'\cap\R_{hor,j}^{(\ell)})$.
Then we define slices based on $\hat{\R}_{hor,j}^{(\ell)}$ and round
them to $1/\epsilon$ different widths, losing a factor of at most
$1+\epsilon$. Again, let $\tilde{\R}_{hor,j}^{(\ell)}$ denote the
resulting set of rounded slices and let $\tilde{\R}_{hor,j}^{(\ell)}=\tilde{\R}_{hor,j,1}^{(\ell)}\dot{\cup}...\dot{\cup}\tilde{\R}_{hor,j,1/\epsilon}^{(\ell)}$
denote a partition of $\tilde{\R}_{hor,j}^{(\ell)}$ according to
the widths of the slices. For each $k\in[1/\epsilon]$ we group the
slices in $\tilde{\R}_{hor,j,k}^{(\ell)}$ into \emph{packs }such
that the items in each pack have a total profit of $\epsilon p(\overline{\opt}')/(c'_{\epsilon}(L+|\overline{\opt}'\cap\Rh|))$
up to factors of $1+\epsilon$, for a constant $c'_{\epsilon}$ to
be defined later, apart from one pack that might have smaller profit.
Let $\hat{\B}$ denote the set of boxes constructed in this way. Note
that by construction $|\hat{\B}|\le c'_{\epsilon}(L+|\overline{\opt}'\cap\Rh|)/\epsilon+L/\epsilon$.

Then we argue that we can find a packing of the slices $\left\{ \tilde{\R}_{hor,j}^{(\ell)}\right\} _{j,\ell},\left\{ \tilde{\R}_{ver,j}^{(\ell)}\right\} _{j,\ell}$
such that they are packed into few containers. Now, consider the optimal packing and slice each
item in $\overline{\opt}'\cap\Rho$ horizontally and each item in
$\opt'\cap\Rve$ vertically. For each $\ell$ such that $\opt_{hor}^{(\ell)}>0$
recall that $\height(\hat{\R}_{hor}^{(\ell)})\le\height(\overline{\opt}'\cap\R_{hor}^{(\ell)})$
by construction and, therefore, we can replace the items in $\opt'\cap\R_{hor,j}^{(\ell)}$
by the slices in $\tilde{\R}_{hor,j}^{(\ell)}$ for each $j,\ell$.
Now we reorder the packing of these. For each horizontal acute piece
$P$ and each guessed box $\bar{\B}(P)$ (see Lemma~\ref{lem:partition-acute-pieces-weighted})
containing skewed items, we order the horizontal slices and the items
$\opt'\cap\Rh$ by width. Also, in the corresponding acute piece $P'$
(in which the items in $\overline{\opt}'$ were stacked on top of
each other) we sort the horizontal slices and the items in $\opt'\cap\Rh$
by width. 
Next, we partition each obtuse piece $P$
into boxes. In our case, these boxes are induced by the horizontal
slices, the items in $\opt'\cap\Rh$ in the acute pieces adjacent
to $P$, and the items in $\overline{\opt}'\cap\Rh$ contained in
$P$. 
We partition each of the resulting boxes in $P$ into $O_{\epsilon}(1)$ subboxes.
Overall, this induces $c''_{\epsilon}(L+|\overline{\opt}'\cap\Rh|)$
containers for the slices for some constant $c''_{\epsilon}$. Note
that in some of them the slices are contained as ``small items''.
Then we argue that we can assign almost all boxes in $\hat{\B}$ into
these containers, more precisely, we can assign at least $|\hat{\B}|-c''_{\epsilon}(L+|\overline{\opt}'\cap\Rh|)$
of them. When we assign a box $B$ to a container that corresponds
to small items, we ensure only that we do not pack items and boxes
of too much area into the container, rather than trying to find an
actual packing of the boxes in the containers. Let $\hat{\B}'\subseteq\hat{\B}$
denote the subset of boxes from $\hat{\B}$ that we packed in this
way. We choose $c'_{\epsilon}:=c''_{\epsilon}$ and then the lost
boxes have only small total profit of 
\[
\frac{\epsilon p(\opt')}{c'_{\epsilon}(L+|\overline{\opt}'\cap\Rh|)}\cdot c''_{\epsilon}(L+|\overline{\opt}'\cap\Rh|)\le\epsilon p(\opt').
\]
It remains to argue that we can pack many items into $\hat{\B}'$.
Recall the item $i^{*}$ due to Lemma~\ref{lem:shifting-weighted}. 
\begin{lem}
If $p(\overline{\opt}')\ge p(i^{*})|\hat{\B}|/\epsilon$ then we can
nicely pack $(1-O(\epsilon))p(\overline{\opt}'\cap\Rl)$ items inside
the boxes $\hat{\B}'$. 
\end{lem}

\begin{proof}
We can nicely pack slices from the sets $\left\{ \tilde{\R}_{hor,j}^{(\ell)}\right\} _{j,\ell},\left\{ \tilde{\R}_{ver,j}^{(\ell)}\right\} _{j,\ell}$
with a total profit of $(1-\epsilon)p(\opt'\cap\Rl)$ into the boxes
$\hat{\B}'$. 
We take out the horizontal slices, order them increasingly by width, order
their boxes increasingly by width, and put them back greedily. This
yields a fractional assignment of the horizontal slices for almost
all items in each set $\hat{\R}_{hor,j}^{(\ell)}$ (losing only a
factor $1+\epsilon$ here) to the boxes $\hat{\B}'$ such that in
each box $B\in\hat{\B}'$ there are at most two items that are fractionally
packed into $B$. We drop these fractionally assigned items and keep
only the integrally assigned ones. This loses at most $2|\hat{\B}|$
items which have a total profit of at most $2|\hat{\B}|p(i^{*})$. 
\end{proof}
We guess the subset $\hat{\B}''\subseteq\hat{\B}'$ of boxes that
we packed above into the obtuse pieces and into the respective acute
pieces $P'$ according to Lemma~\ref{lem:partition-acute-pieces-weighted}.
We can do this in time $2^{|\hat{\B}|}=(nN)^{O_{\epsilon}(1)}$. Finally,
in order to argue that $p(\overline{\opt}')\ge p(i^{*})|\hat{\B}|/\epsilon$
recall that Lemma~\ref{lem:shifting-weighted} gives that $p(\overline{\opt}')\ge\Omega\left(\frac{c\log(nN)+\left|\overline{\opt}'\cap\Rh\right|}{\epsilon}p(i^{*})\right)$
for any given constant $c$. Since $|\hat{\B}|\le c'_{\epsilon}(L+|\overline{\opt}'\cap\Rh|)/\epsilon+L\le O_{\epsilon}(\log(nN)+|\overline{\opt}'\cap\Rh|)$
we can define $c$ such that $p(\overline{\opt}')\ge p(i^{*})|\hat{\B}|/\epsilon$.

\end{document}